\documentclass[journal]{IEEEtran}
\usepackage{amsmath,amsfonts}
\usepackage{array}
\usepackage[caption=false,font=normalsize,labelfont=sf,textfont=sf]{subfig}
\usepackage{textcomp}
\usepackage{stfloats}
\usepackage{url}
\usepackage{verbatim}
\usepackage{graphicx}
\usepackage{cite}
\usepackage{float}
\usepackage{graphicx}
\usepackage{amssymb}
\usepackage{epsfig}
\usepackage{amsmath}
\usepackage{amsthm}
\usepackage{amsfonts}
\usepackage{setspace}
\usepackage{url}
\usepackage{psfrag}
\usepackage{graphicx}
\usepackage{dsfont}
\usepackage{units}
\usepackage{graphicx}
\usepackage{wrapfig}
\usepackage[linguistics]{forest}
\usepackage{algorithm}
\usepackage{algpseudocode}   % algorithmicx's pseudocode\usepackage{algorithm}       % float
\usepackage{dsfont}

\newcommand{\RomanNumeralCaps}[1]
{\MakeUppercase{\romannumeral #1}}

\usepackage{amsthm}

\newtheorem{thm}{Theorem}
\newtheorem{cor}{Corollary}
\newtheorem{lem}{Lemma}

\usepackage{enumitem}
\usepackage{authblk}

\begin{document}

\title{Misspecified Universal Learning}

\author{
    Shlomi~Vituri\\
    Electrical and Computer Engineering, Tel-Aviv University\\
    {shlomivituri@mail.tau.ac.il}
    \and\\
    Meir~Feder\\
    Electrical and Computer Engineering, Tel-Aviv University\\
    {meir@tau.ac.il}
}

\maketitle

	\begin{abstract}
    This paper addresses the problem of universal learning under model misspecification with log-loss. In this setting, the learner operates with a hypothesis class of models denoted by $\Theta$, while the true data-generating process belongs to a broader class $\Phi \supset \Theta$, and may lie outside the assumed hypothesis space. Classical approaches have characterized the minimax regret and identified optimal universal learners in both the well-specified stochastic and individual deterministic frameworks.
    The misspecified setting has received comparatively less attention, although several important results have emerged in recent years.
    Extending these foundations, we analyze the minimax regret in the misspecified setting and derive the corresponding optimal universal learner. We propose this formulation as a unified framework for universal learning, applicable to any form of uncertainty in the data-generating process, across both online and batch data arrival modes, as well as supervised and unsupervised learning tasks.
    \end{abstract}


    \section{Introduction}
    \IEEEPARstart{I}{n} recent years, the field of inductive learning has seen a surge of interest driven by both theoretical advancements and practical applications. Traditionally, the primary theoretical framework for inductive learning has been statistical learning theory.
In the standard formulation of the learning problem, a batch of training samples drawn from an unknown data source is used to predict a future ``test'' outcome as accurately as possible, given a predefined hypothesis class and a loss function. The learning task may be unsupervised, such as estimating the probability of the next outcome \( y_n \) based on a sequence of previous observations \( y^{n-1} \), see \cite{kamath2015learning}, or supervised, where the training data consists of feature-label pairs \( (x^{n-1}, y^{n-1}) \), and the goal is to predict the test label \( y_n \) given the training data and the test feature \( x_n \), see, e.g., \cite{bousquet2004introduction_to_STL}, \cite{Vapnik}.

One of the most widely adopted frameworks in learning theory is the Probably Approximately Correct (PAC) model~\cite{valiant1984theory}. In PAC learning, data samples are assumed to be independently drawn from a distribution where each sample consists of an input and its corresponding label. The likelihood of observing a specific input-label pair is determined by the probability of the input occurring, combined with the probability of the label conditioned on that input. Notably, this conditional relationship between inputs and labels may fall outside the scope of the hypothesis class. The goal is to develop algorithms that, with high probability over the training samples, achieve a loss that is close to the minimum possible within the hypothesis class.

A central concept in statistical learning theory is the Vapnik-Chervonenkis (VC) dimension, introduced in \cite{VC}. The VC dimension quantifies the capacity or expressiveness of a hypothesis class and plays a key role in understanding generalization.
The relationship between model complexity, data availability, and generalization performance is captured by the fundamental theorem of statistical learning. If the VC dimension is finite, then with high probability, the empirical loss over the training data uniformly approximates the expected loss over new data for all hypotheses in the class. This guarantees that the hypothesis minimizing empirical risk will also generalize well; hence, empirical risk minimization is a common strategy in batch learning.

Another important complexity measure in statistical learning theory is the Rademacher complexity \cite{bartlett2002rademacher}, which, unlike VC dimension, extends naturally to real valued hypothesis classes and general loss functions. Although both the VC dimension and Rademacher complexity provide strong theoretical foundations for learning simple hypothesis classes, they fail to explain the generalization behavior of modern models, most notably deep neural networks (DNNs); see \cite{zhang2016understanding}. This limitation has sparked a growing interest in the development of new theoretical tools to analyze and understand the performance of contemporary learning algorithms.

As an alternative to statistical learning theory, the field of universal prediction in information theory offers an elegant approach to learning problems. Broadly speaking, universal prediction addresses scenarios where one seeks to estimate the probability of a sequence $y^n$ under uncertainty about the true distribution from which it is generated.
Similarly to statistical learning, universal learning or prediction also assumes a hypotheses class $\Theta$ and a set of possible data-generating distributions, denoted by $\Phi$, where $\Theta \subseteq \Phi$.
In contrast to a non-universal learner who has access to the true data-generating model and can select the optimal hypothesis from $\Theta$, the universal learner operates without such knowledge and aims to perform competitively against this reference.

Two extreme cases of this setting have been extensively studied. The first is the \textit{well-specified stochastic setting}, where \( \Theta \equiv \Phi \); that is, the data are assumed to be generated by an unknown distribution within the reference class \( \Theta \). At the other extreme lies the \textit{deterministic individual setting}, in which no assumptions are made about the data, it is treated as an arbitrary individual sequence. Between these extremes is the \textit{misspecified stochastic setting}, where the data are generated from a distribution in a broader set \( \Phi \), while the hypothesis class \( \Theta \), from which the reference learner is selected, is only a subset of \( \Phi \). This corresponds to the common case in agnostic statistical learning, where the data generating process may follow any i.i.d. distribution, while the learner is restricted to a hypothesis class \( \Theta \).
The individual setting can be viewed as an extreme edge case of the misspecified setting, where the model class is taken to be \( \Phi = \mathcal{P} \), with \( \mathcal{P} \) denoting the set of all probability distributions. In this sense, the misspecified setting provides a general and unified framework for universal prediction theory.

The universal learning problem can vary not only in its set of data-generating distributions and hypothesis class, but also in the manner of data arrival, either \textit{online} or \textit{batch}, and in the type of data, which may be \textit{supervised} or \textit{unsupervised}. The online setting, which is classical in information theory, assumes that data arrives sequentially, one sample at a time. In this case, the universal predictor estimates the probability of the next outcome based on all previously observed samples. In contrast, the batch setting assumes that the data is provided all at once as a complete set of training samples, and the universal predictor aims to infer the probability distribution of future outcomes based on this batch. Regarding the data type, in the supervised case, the data features, denoted by \( x^n \), are accompanied by a corresponding label sequence \( y^n \). In the unsupervised case, such labels are not available, and the learning process relies solely on the structure of the input data.

The well-specified and individual settings in the unsupervised online case have been thoroughly investigated under various data-generating assumptions using the log-loss measure; see, e.g., \cite{universal98prediction}, \cite{davisson1980source}, \cite{ryabko1979coding}, \cite{shtar1987universal}. In addition, preliminary studies of online supervised universal learning have also been investigated; see \cite{online17learning}.
More recently, progress has been made in formulating and solving universal batch learning problems for both supervised and unsupervised cases, from an information-theoretic perspective; see, e.g., \cite{batch18learning}, \cite{individual18learning}, \cite{fogel2019individual}, \cite{deep20pnml}.
On the other hand, within the information-theoretic framework, the first mention of the misspecified setting appeared in \cite{robust1998minimax}. However, it remained largely unexplored until recent works began addressing both online and batch learning under log-loss in the misspecified setting; see \cite{robust21inference}, \cite{online21misspecified}, \cite{misspecified24batch}, \cite{misspecified24learning_arxiv}, \cite{constrained25learning}.

One of the main contributions of this paper is a comprehensive analysis of the minimax regret in the misspecified setting, addressing all data arrival modes online and batch and encompassing both supervised and unsupervised data types. We show that the minimax regret can be interpreted as a constrained version of the capacity between the data and the set of data-generating distributions \( \Phi \). Interestingly, the minimax regret is shown to be approximately equal to the regret in the well-specified stochastic case, assuming the data were generated by a distribution from the hypothesis class \( \Theta \) rather than the larger set \( \Phi \). This implies that the complexity of the learning problem is governed primarily by the hypothesis class \( \Theta \), rather than the broader family of distributions that may have generated the data.
To illustrate our results, we consider the case where observations are drawn from a multinomial distribution, while the hypothesis class \( \Theta \) is a strict subset of this family. We numerically evaluate both the regret and the capacity-achieving prior by developing an extension of the Arimoto-Blahut algorithm; see, e.g., \cite{blahut72},\cite{arimoto72},\cite{arimoto2004blahut_alg},\cite{arimoto2018blahut_alg},\cite{arimoto2013blahut_alg},\cite{arimoto2019blahut_alg},\cite{combined24batch_online},\cite{robust21inference},\cite{misspecified24batch},\cite{misspecified24learning_arxiv}.

Another key contribution of this paper is the introduction of a new framework, referred to as the \textit{constrained misspecified setting}.
In all settings previously studied, the universal learner is typically defined as a mixture distribution over the data-generating distribution set \( \Phi \). 
In addition, as mentioned previously, in the misspecified setting, under certain regularity conditions, this mixture tends to concentrate around the hypothesis class \( \Theta \), and the minimax regret becomes approximately equal to the regret in the well-specified stochastic case, assuming the data were generated by a distribution in \( \Theta \) rather than the larger set \( \Phi \); see \cite{online21misspecified}, \cite{misspecified24batch}, \cite{misspecified24learning_arxiv}.
Moreover, Barron et al.\ analyzed the asymptotic minimax regret in the individual setting, where \( \Theta \) is a set of smooth parametric models and the universal predictor is constrained to a mixture distribution with a prior over \( \theta \in \Theta \), rather than using the Normalized Maximum Likelihood (NML) predictor \cite{shtar1987universal}. They showed that in such models, the asymptotically optimal prior is a modified Jeffreys prior, see \cite{jeffreys46priors}, \cite{clarke1994jeffreys}, which coincides with the capacity-achieving prior in the well-specified stochastic setting. 
For specific model families, see Xie and Barron \cite{barron2000individual} for discrete memoryless multinomial distributions, Takeuchi and Barron \cite{barron2013nonExponential} for exponential families, Atteson and Barron \cite{markov99redundancy}, Kawabata and Barron \cite{barron2013markov} for Markov models, and Gotoh \cite{gotoh1998fsmx} for finite state machines (FSM).

Motivated by these insights, we propose a new approach for the misspecified setting, in which the universal learner is \textit{constrained} to a mixture distribution over the hypothesis class \( \Theta \), rather than over the full set \( \Phi \). We analyze the regret associated with this constrained universal learner, denoted by
\[
Q(y^n) = \int \pi(\theta) P_\theta(y^n) \, d\theta,
\]
which represents the expected difference, under the true data-generating distribution, between the log-loss of the constrained universal predictor and that of the best-fitting hypothesis in \( \Theta \).

Some specific results in this direction can be found in \cite{nir2022misspecified}, which analyzes convergence rates in the misspecified setting, and in \cite{markov99redundancy}, which studies constrained regret under a given \( \phi \)-mixing data-generating distributions \cite{processes69learning}, where the hypothesis class consists of Markov chain processes.
In this work, we further analyze this setting in terms of minimax regret and optimal universal learners, and demonstrate their performance across various examples of smooth parametric models, including memoryless multinomial distributions, binary Bernoulli distributions, Markov models, and exponential families of distributions. In addition, due to the analytical complexity of evaluating the optimal learner and its minimax regret, we develop an extension of the Arimoto-Blahut algorithm tailored to the constrained setting and present corresponding numerical results. Finally, we discuss the analogies and distinctions between the misspecified setting, the well-specified stochastic setting, and the individual deterministic setting.

In summary, the theory of universal prediction presents a compelling alternative to classical statistical learning theory. As depicted in Figure~\ref{figLearningTheoryTree}, this framework can be organized into a three-layer hierarchical structure.
The first layer characterizes the nature of the data-generating process, which may follow a stochastic, deterministic individual, misspecified, or the newly proposed constrained misspecified setting. These variants can be unified under a general framework grounded in the misspecified setting.
The second layer addresses the mode of data arrival, distinguishing between online learning, where data are received sequentially, which is classical in information theory, and batch learning, which is standard in statistical learning.
The third layer pertains to the learning paradigm, differentiating between supervised learning, where inputs are paired with corresponding labels, and unsupervised learning, where labels are absent.
This hierarchical organization enables a unified perspective on diverse learning scenarios and highlights the adaptability of universal prediction as a foundational theoretical framework.

The outline of the paper is as follows. In Section~\RomanNumeralCaps{2}, we provide a formal definition of the misspecified universal learning problem. Sections~\RomanNumeralCaps{3} and~\RomanNumeralCaps{4} then present our main results for the misspecified and constrained misspecified settings, respectively.  Finally, Section~\RomanNumeralCaps{5} concludes the paper.

%\begin{figure}[ht]
\begin{figure*}[!b]
\centering
\begin{forest}
  for tree={
    if n children=0{
      font=\itshape
    }{},
  }
  [{\textbf{Learning Theory}}
    [{\textbf{Statistical Learning Theory}}
        [\footnotesize{\textnormal{PAC Model}}]
    ]
    [{\textbf{Information Theory}}
        [\footnotesize{Universal Prediction}
            [\scriptsize{Stochastic}
                [\scriptsize{Online}
                    [\scriptsize{\textnormal{Un-/Supervised}}
                        %[\scriptsize{\textnormal{\textnormal{\cite{universal98prediction}} / \cite{online17learning}}}]
                    ]
                ]
                [\scriptsize{Batch}
                    [\scriptsize{\textnormal{Un-/Supervised}}
                        %[\scriptsize{\textnormal{\cite{batch18learning}}}]
                    ]
                ]
            ]
            [\scriptsize{Individual}
                [\scriptsize{Online}
                    [\scriptsize{\textnormal{Un-/Supervised}}
                        %[\scriptsize{\textnormal{\cite{universal98prediction} / \cite{online17learning}} }]
                    ]
                ]
                [\scriptsize{Batch}
                    [\scriptsize{\textnormal{Un-/Supervised}}
                        %[\scriptsize{\textnormal{? / \cite{batch18learning},\cite{individual18learning},\cite{fogel2019individual},\cite{deep20pnml}}}]
                    ]
                ]
            ]            
            [\scriptsize{Misspecified}
                [\scriptsize{Online}
                    [\scriptsize{\textnormal{Un-/Supervised}}
                        %[\scriptsize{\textnormal{\cite{online21misspecified},\cite{robust21inference} / \cite{misspecified24learning_arxiv}*}}]
                    ]
                ]
                [\scriptsize{Batch}
                    [\scriptsize{\textnormal{Un-/Supervised}}
                        %[\scriptsize{\textnormal{\cite{misspecified24batch}* / \cite{misspecified24learning_arxiv}*}}]
                    ]
                ]
            ]            
            [\scriptsize{Constrained}
                [\scriptsize{Online}
                    [\scriptsize{\textnormal{Un-/Supervised}}
                        %[\scriptsize{\textnormal{\cite{constrained25learning}* / ?}}]
                    ]
                ]
                [\scriptsize{Batch}
                    [\scriptsize{\textnormal{Un-/Supervised}}
                        %[\scriptsize{\textnormal{\cite{constrained25learning}* / ?}}]
                    ]
                ]
            ]
        ]
    ]
  ]
\end{forest}
\caption{The Learning Theory Tree.}
\label{figLearningTheoryTree}
%\footnote{*}{Part of this research}
\end{figure*}
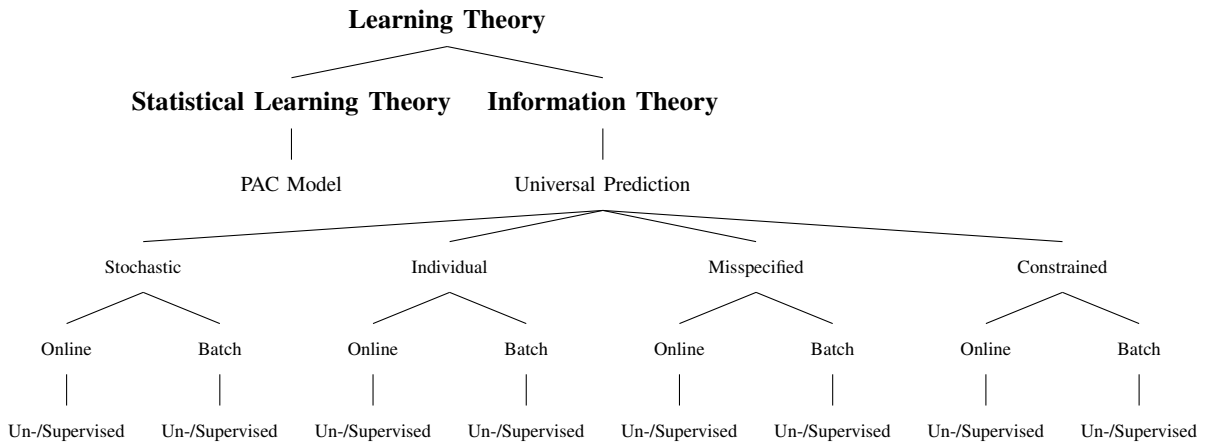

    \section{Formal Problem Definition} 
    \label{sec:basic_definitions}
    This section presents the fundamental definitions, problem formulation, and notation employed throughout the paper. The exposition is organized into two parts, addressing the misspecified setting and the constrained misspecified setting, respectively, as detailed in the subsequent subsections.

\subsection{Misspecified Setting Definition}
We denote the hypothesis class of distributions by $\Theta$, where each $\theta \in \Theta$ induces a probability distribution $P_\theta$. With a slight abuse of notation, we occasionally write $P_\theta \in \Theta$ when the meaning is clear from the context. 
Similarly, we denote by $\Phi$ the set of true data-generating distributions, with $\Theta \subseteq \Phi$, where each $\phi \in \Phi$ induces a distribution $P_\phi$. Here, also, with a slight abuse of notation, we may write $P_\phi \in \Phi$ when no ambiguity arises

In the misspecified batch unsupervised setting, data is generated according to an unknown distribution \( P_\phi \in \Phi \). Given the observed batch \( y^{n-1} \), the learner predicts the next outcome via a universal predictor \( Q(y_n | y^{n-1}) \), without knowledge of \( P_\phi \). Performance is measured under the log-loss relative to a hypothesis distribution \( P_\theta \in \Theta \). The resulting regret associated with the triplet \( (P_\phi, P_\theta, Q) \) is defined as
\begin{align}
    \begin{aligned}
        R_{n}(\theta,\phi,Q) &= \mathbb{E}_{P_\phi}\left\{ \log{\frac{P_\theta(Y_{n}|Y^{n-1})}{Q(Y_n|Y^{n-1})}} \right\}.
    \end{aligned}
\end{align}

In the misspecified online unsupervised setting, the task is to sequentially predict the entire sequence \( y^n \) generated according to an unknown data-generating distribution \( P_\phi \). The learner employs a universal predictor \( Q(y^n) \), without knowledge of \( P_\phi \), and its performance is evaluated under the log-loss measure relative to a hypothesis distribution \( P_\theta \in \Theta \). The corresponding regret is defined as
\begin{align}
    \begin{aligned}
        R_{n}(\theta,\phi,Q) &= \mathbb{E}_{P_\phi}\left\{ \log{\frac{P_\theta(Y^{n})}{Q(Y^n)}} \right\}.
    \end{aligned}
\end{align}

In the supervised setting, the regret is defined analogously. Given the data features distribution $P(x^n)$, the regret associated with predicting the label $y_n$ given a batch of samples, conditioned on the past observations $z^{n-1} \equiv (x^{n-1},y^{n-1})$ and the current feature $x_n$, is defined as
\begin{align}
    \begin{aligned}
    \nonumber
        &R_{n}(\theta,\phi,Q,P) = \mathbb{E}_{P_\phi P}\left\{ \log{\frac{P_\theta(Y_n|X^n,Y^{n-1})}{Q(Y_n|X^n,Y^{n-1})}} \right\}        
        \\& = \sum_{x^n}\sum_{y^n}{ P(x^n)P_\phi(y^n|x^n) }\log{\frac{P_\theta(y_n|x^n,y^{n-1})}{Q(y_n|x^n,y^{n-1})}}.
    \end{aligned}
\end{align}
In the supervised online setting, where the task is to sequentially predict each label $y_t$ given the past observations $(x^{t-1}, y^{t-1})$ and the current feature $x_t$ for $t = 1, \dots, n$, the regret is defined as
\begin{align}
    \begin{aligned}
    \nonumber
        &R_{n}(\theta,\phi,Q,P) = \mathbb{E}_{P_\phi P}\left\{ \log{\frac{P_\theta(Y^n\|X^n)}{Q(Y^n\|X^n)}} \right\}        
        \\& = \sum_{x^n}\sum_{y^n}{ P(x^n)P_\phi(y^n|x^n) }\log{\frac{\prod_{t=1}^nP_\theta(y_t|x^t,y^{t-1})}{\prod_{t=1}^nQ(y_t|x^t,y^{t-1})}},
    \end{aligned}
\end{align}
which are valid for causal models, where $P(y^n\|x^n) = \prod_{t=1}^nP(y_t|x^t,y^{t-1})$ denotes a \textit{causal} (sequential) conditional distribution.

The performance of the universal learner is evaluated via the classical minimax regret criterion, which measures its worst-case log-loss over \( \Phi \) relative to that of the best hypothesis in \( \Theta \). As the unsupervised setting constitutes the main focus of this work, its minimax formulation is presented in detail below, while the supervised case is defined analogously. The unsupervised misspecified minimax regret is defined as
\begin{align}
    \begin{aligned}
    \label{eq:misspecified_minimax_regret_definition}
        F_{s,n}(\Theta,\Phi)
        &= \min_{Q}\; \max_{P_\phi \in \Phi}\; 
           \max_{P_\theta \in \Theta} R_n(\theta,\phi,Q) \\        
        &= \min_{Q}\; \max_{P_\phi \in \Phi} 
           \left(D(P_\phi\|Q) - D(P_\phi\|P_{\theta^*})\right),
    \end{aligned}    
\end{align}
where 
\[
    \theta^* 
        = \arg\max_{P_\theta \in \Theta} D(P_\phi \,\|\, P_\theta)
\]
denotes the KL-projection of $P_\phi$ onto the hypothesis class $\Theta$,  and \( s \in \{ b, o\} \) denotes the learning mode, with \( s = o \) corresponding to the online setting and \( s = b \) to the batch setting. The minimax regret in the supervised setting is defined analogously and denoted by \( R_{s,n}^*(\Theta,\Phi) \).
%for both batch and online settings.

%Throughout, \( F_n(\Theta,\Phi) \) denotes the misspecified minimax regret in the unsupervised setting, encompassing both batch and online cases, with the intended interpretation determined by the context. 
Throughout the paper, with a slight abuse of notation, we omit the subscripts and denote by \( F_n(\Theta,\Phi) \) and \( R_n^*(\Theta,\Phi) \) the misspecified minimax regret in the unsupervised and supervised settings, respectively. Unless stated otherwise, this notation encompasses both batch and online cases, with the intended interpretation clearly derived from the context.

The above definitions provide a general framework for universal prediction problems, and this framework can be illustrated through two extensively studied 
extreme cases: the well-specified setting and the deterministic individual-sequence setting. These two settings, reviewed briefly below, will serve as reference points for evaluating the performance of our proposed misspecified framework.

In the first case, where $\Theta \equiv \Phi$, the minimax regret reduces to the classical redundancy-capacity expression. In the online setting, see \cite{gallager1974source},\cite{davisson1980source},\cite{ryabko1979coding}, this quantity is
\begin{align}
    C_n(\Theta) = \max_{\pi(\theta)} I(Y^n ; \Theta),
\end{align}
where the maximization is over all priors $\pi(\theta)$ on $\Theta$. In the batch setting, see \cite{batch18learning}, the corresponding conditional redundancy-capacity is
\begin{align}
    C_{c,n}(\Theta) = \max_{\pi(\theta)} I(Y^n ; \Theta | Y^{n-1}),
\end{align}
and in both settings the universal predictor achieving these capacities is the Bayesian mixture distribution
\[
    Q(y^n) = \int_{\Theta} \pi(\theta)\, P_\theta(y^n)\, d\theta.
\]

At the other extreme, where
\[
    \Phi = \mathcal{P}
    = \left\{ P : P(y^{n}) \ge 0,\;
    \int_{\mathcal{Y}^{n}} P(y^{n})\, dy^{n} = 1 \right\},
\]
corresponding to the set of all distributions on $\mathcal{Y}^n$, we obtain the deterministic individual-sequence setting. Here, the regret is defined with respect to the worst-case sequence $y^n$. Shtarkov~\cite{shtar1987universal} showed that the optimal universal predictor under log-loss is the NML distribution,
\begin{align}
    \begin{aligned}
        Q^*(y^n) = \frac{\max_\theta P_\theta(y^n)}{\sum_{y^n}\max_\theta P_\theta(y^n)},
    \end{aligned}
\end{align}
and the minimax regret is
\begin{align}
    \begin{aligned}
        \Gamma_n(\Theta) = \log{\sum_{y}\max_{\theta}P_\theta(y^n)}.
    \end{aligned}
\end{align}
In the batch individual-sequence setting, the regret is evaluated with respect to predicting a single next outcome $y$ given the observed sequence $y^n$. Fogel and Feder~\cite{fogel2019individual,fogel_batch_individual,Fogel2018} established that
\begin{align}
    \begin{aligned}
        \Gamma_n(\Theta) = \log{\sum_{y}\max_{\theta}P_\theta(y|y^n)}
    \end{aligned}
\end{align}
and that the optimal universal predictor in this setting is the predictive Normalized Maximum Likelihood (pNML) distribution,
\begin{align}
    \begin{aligned}
        Q^*(y|y^n) = \log{ \frac{\max_{\theta}P_\theta(y|y^n)}{\sum_{y}\max_{\theta}P_\theta(y|y^n)} }.
    \end{aligned}
\end{align}

For the supervised case, partial results for the online and batch settings have been established in \cite{online17learning} and \cite{batch18learning,individual18learning,fogel2019individual}, respectively.

Furthermore, as noted above, the misspecified online setting has recently been studied in \cite{online21misspecified} and \cite{robust21inference}. A brief overview of these results will be provided in the sequel.

\subsection{Constrained Misspecified Setting Definition}
In this work, we also introduce a new variant of the misspecified setting, in which the universal predictor is constrained to be a mixture distribution over the convex hull of the hypothesis class $\Theta$. Specifically, in the online setting, the predictor is restricted to the form
\begin{align}
    Q_{\pi_0(\theta)}(y^n) = \int_{\Theta} \pi_0(\theta)\, P_\theta(y^n)\, d\theta,
\end{align}
where $\pi_0(\theta)$ is a prior distribution over $\Theta$. We refer to this framework as the \emph{constrained misspecified setting}. 

The constrained misspecified regret in the case of data generating distribution $P_\phi \in \Phi$, hypothesis $P_\theta \in \Theta$ and a constrained mixture distribution $Q_{\pi_0(\theta)}$ is defined by:
\begin{equation}
    R_n(\theta,\phi,Q_{\pi_0(\theta)}) = \mathbb{E}_{P_\phi}\left\{ \log {\frac{P_{\theta}(y^n)}{Q_{\pi_0(\theta)}(y^n)}} \right\}.
\end{equation}
Its performance is evaluated via the following minimax regret:
\begin{align}
    \begin{aligned}
    \label{eq:minmaxRegretDefinition_constrained}
        R_{o,n}^*(\Theta,\Phi) &= \min_{\pi_0(\theta)}\max_{P_\phi\in\Phi}R_n(\theta^*,\phi,Q_{\pi_0(\theta)}),
    \end{aligned}
\end{align}
where the subscript ``o'' denotes the \emph{online} setting, and $P_{\theta^*} \equiv \arg\min_{P_\theta\in\Theta} D(P_\phi\|P_\theta)$ is the projection of $P_\phi$ into the set $\Theta$.	

In words, the constrained minimax regret minimizes the regret of the worst data generating distribution and the best matching hypothesis to the data, via a universal distribution $Q_{\pi_0(\theta)}$, which is constrained to be a mixture distribution over the set of hypotheses $\Theta$, for any series of data samples $y^n$. 

By relaxing this constraint and allowing the regret to be minimized over all universal predictors \( Q(y^n) \), we recover the misspecified setting, in which the minimax regret is given by \( F_{\text{o},n}(\Theta,\Phi) \). In the well-specified stochastic case, corresponding to \( \Phi = \Theta \), the regret further reduces to the redundancy-capacity \( C_n(\Theta) \).

The case \( \Phi = \mathcal{P} \), corresponding to the constrained deterministic individual-sequence setting, requires a more delicate analysis. Its connection to the individual-sequence regret \( \Gamma_n(\Theta) \) is discussed in detail in Section~\ref{sec:constrained_setting}.

Similarly, the constrained misspecified minimax regret in the batch setting is defined as
\begin{align}
    \begin{aligned}
        R_{b,n}^*(\Theta,\Phi)
        &= \min_{\pi_0(\theta)}
           \max_{P_\phi \in \Phi}
           \mathbb{E}_{P_\phi}\left\{
           \log \frac{P_{\theta^*}(Y_n | Y^{n-1})}
                    {Q_{\pi_0(\theta)}(Y_n | Y^{n-1})}
           \right\}.
    \end{aligned}
\end{align}
Here, the subscript ``b'' denotes the \emph{batch} setting, and
\[
    Q_{\pi_0(\theta)}(y_n | y^{n-1})
    = \frac{\int_\Theta \pi_0(\theta) P_\theta(y^{n}) \, d\theta}
           {\int_\Theta \pi_0(\theta) P_\theta(y^{n-1}) \, d\theta},
\]
where \( \theta^* \) denotes the KL-divergence projection of \( P_\phi \) onto the hypothesis class \( \Theta \).

%Note that, to simplify notation, we omit the subscripts and denote the constrained universal predictor by \( Q \) and the associated minimax regret by $R_n^*(\Theta,\Phi)$ whenever its definition is clear from the context.
Note that, for notational simplicity, we omit subscripts and denote the constrained universal predictor by \( Q \) and the associated minimax regret by \( R_n^*(\Theta, \Phi)\) whenever their definitions and the underlying setting (online or batch) are clear from the context.

In our analysis of the constrained misspecified setting, we focus primarily on smooth parametric models satisfying the standard \emph{regularity conditions} of \cite{clarke1990infoasymptotics,clarke1994jeffreys}, and thus rely on several statistical and information‑theoretic quantities that require precise definition. Recall that \( \theta^* \) denotes the projection of \( P_\phi \)
onto the hypothesis class \( \Theta \). We further define
\[
\hat{\theta}(y^n) \triangleq \arg\max_{\theta \in \Theta} P_\theta(y^n)
\]
to be the Maximum Likelihood Estimator (MLE) based on the sample \( y^n \). Under
\( P_\phi \in \Phi \), we define the expected Hessian of the log-loss as follows:
\[
J_\phi(\theta) = \mathbb{E}_{P_\phi}\big\{ -\nabla^2_\theta \log P_\theta(Y) \big\},
\]
and the score variance as:
\[
K_\phi(\theta) = \mathbb{E}_{P_\phi}\big\{ \nabla_\theta \log P_\theta(Y)\,\nabla_\theta \log P_\theta(Y)^\top \big\},
\]
where the score function is given by $\nabla_\theta \log P_\theta(Y)$.

In addition, we introduce the empirical Fisher information based on the sample $y^n$, whose expectation equals the Fisher information matrix. The empirical Fisher information is defined as:
\[
\hat{I}_{ij}(\theta, y^n) = -\frac{1}{n}\,\frac{\partial^2 \log P_\theta(y^n)}{\partial \theta_i \,\partial \theta_j},
\]
while the Fisher information is given by:
\[
I_{ij}(\theta) = \mathbb{E}_{P_\theta}\big\{ \hat{I}_{ij}(\theta, y^n) \big\},
\]
where \( i,j \in \{1,\dots,d\} \), and \( d \) denotes the dimension of the parameter vector \( \theta \).

Moreover, since a central component of the constrained misspecified analysis for smooth parametric models involves exponential families, we present a formal definition of these families and summarizing their key structural properties, which play a crucial role in simplifying the subsequent theoretical development.

We adopt the terminology of \cite{barron2013nonExponential},\cite{TakeuchiBarron_ExponentialFamilies} for exponential families. Given a Borel measurable function $T:\mathcal{Y}\rightarrow\mathbb{R}^d$, define:
\begin{align}
    \begin{aligned}
    \nonumber
        \Theta = \{ \theta\ : \theta\in \mathbb{R}^d, \int_{\mathcal{Y}}{h(y)\exp{\left(\theta^{\top} \cdot T(y)\right)}}\,dy < \infty\}.
    \end{aligned}
\end{align}
Define a function $\psi$ and a probability density $P_\theta$ on $\mathcal{Y}$ by:
\begin{align}
    \begin{aligned}
    \nonumber
        \psi(\theta) \equiv \ln \int_{\mathcal{Y}}{h(y)\exp{\left( \theta^{\top} \cdot T(y) \right)}}\,dy,
    \end{aligned}
\end{align}
and
\begin{align}
    \begin{aligned}
    \nonumber
        P_\theta(y) \equiv h(y)\exp{\left( \theta^{\top} \cdot T(y) - \psi(\theta) \right)}. 
    \end{aligned}
\end{align}

The collection of densities \( \{ P_\theta : \theta \in \Theta \} \) constitutes an \emph{exponential family}. When \( \Theta \) is an open subset of \( \mathbb{R}^d \), this family is termed a \emph{regular exponential family}. For convenience, we will denote this family simply by \( \Theta \), with a slight abuse of notation.

For such families, the score function is as follows:
\[
    \nabla_\theta \log P_\theta(y) = T(y) - \nabla_\theta \psi(\theta),
\]
and the normalized expected Hessian of the log-loss under \( P_\phi \) is:
\[
    J_\phi(\theta) = \nabla^2_\theta \psi(\theta) = I(\theta),
\]
where \( I(\theta) \) denotes the Fisher information matrix. 
%For exponential families, these quantities coincide at the MLE \( \hat{\theta} \), i.e., \( I(\hat{\theta}) = \hat{I}(\hat{\theta}) \).
Moreover, for exponential families, the empirical Fisher information coincides with the Fisher information at the MLE \( \hat{\theta} \), i.e., \( \hat{I}(\hat{\theta}) = I(\hat{\theta}) \).

Finally, the normalized score variance is given by:  
\[
    K_\phi(\theta) = \operatorname{Cov}_{P_\phi}\big(T(Y)\big).
\]
    
    \section{Misspecified Universal Learning} 
    \label{sec:misspecified_setting}
    The misspecified universal prediction framework was initially introduced by Barron et al.~\cite{robust1998minimax} in 1998. However, it remained largely unexplored until its resurgence in 2021, when Feder and Polyanskiy~\cite{online21misspecified}, as well as Painsky and Feder~\cite{robust21inference}, revisited this setting in the context of online prediction. In this Section, we first summarize the existing results on misspecified online prediction in the unsupervised case and establish preliminary extensions to the supervised case. Furthermore, we provide a detailed and in-depth extension of the analysis to the misspecified batch setting, which to the best of our knowledge, has not been investigated previously.
\subsection{Online Setting}
\label{subsec:misspecified_online_setting}
Using the minimax theorem for concave-convex functions~\cite{sion1958general} and standard tools from information theory, it was shown in \cite{online21misspecified}, \cite{robust21inference} that the misspecified regret and the corresponding universal predictor are characterized by the following:
\begin{thm}[\cite{robust21inference}, Theorem 2]
\label{thm:online_misspecified_regret}
Let $\Theta \subseteq \Phi$ be a hypothesis class and a model class, respectively. Then, the misspecified regret in the online prediction setting under log-loss is given by:
\begin{align}
    F_{o,n}(\Theta, \Phi) = \max_{\pi(\phi)} \left( I(Y^n ; \Phi) - \mathbb{E}_{\pi(\phi)} \left\{ D(P_\phi \| \Theta) \right\} \right)
\end{align}
where the divergence term is defined as:
\begin{align}
    D(P_\phi \| \Theta) \equiv \min_{P_\theta \in \Theta} D(P_\phi \| P_\theta)
\end{align}
and the universal predictor that achieves this regret is:
\begin{align}
    Q(y^n) = \int_{\Phi} \pi(\phi) P_\phi(y^n) \, d\phi
\end{align}
with \( \pi(\phi) \) being a prior distribution over \( \Phi \).
\end{thm}
Theorem~\ref{thm:online_misspecified_regret} demonstrates that, analogously to the well-specified setting, the misspecified regret can be interpreted as a constrained form of capacity between the observed data samples \( Y^n \) and the model class \( \Phi \). In this formulation, the capacity-achieving prior is given by \( \pi(\phi) \), and the optimal universal predictor corresponds to a mixture distribution over models in \( \Phi \), weighted by this prior.

An analytical evaluation of the misspecified regret and the capacity-achieving prior distribution in Theorem~\ref{thm:online_misspecified_regret} is intractable in many cases. To address this, \cite{robust21inference} introduces an extension of the Arimoto-Blahut algorithm, enabling the numerical evaluation of these quantities, as formalized in the following theorem:

\begin{thm}[\cite{robust21inference}, Theorem 3]
Let \( \Phi \) and \( \Theta \) be two model classes, and let \( F_{o,n}(\Theta,\Phi) \) be defined as above. Assume that \( \Theta \) is bounded. Then, for \( F_{o,n}(\Theta,\Phi) < \infty \), the following holds:
\begin{align}
    \begin{aligned}
    \label{eq:misspecified_regret_arimoto_blahut}
        F_{o,n}(\Theta,\Phi) =& \sup_{\pi(\phi),\psi(\phi,y^n)} \int_{\Phi} \int_{y^n} \pi(\phi) P_\phi(y^n) \cdot 
        \\&\log \left( \frac{\psi(\phi,y^n)}{\pi(\phi)} \frac{P_{\theta^*}(y^n)}{P_\phi(y^n)} \right) \, dy^n \, d\phi
    \end{aligned}
\end{align}
where \( \pi(\phi) \) and \( \psi(\phi,y^n) \) are probability distributions over \( \phi \), for each fixed \( y^n \). Furthermore, the solution to \eqref{eq:misspecified_regret_arimoto_blahut} may be attained via the following iterative projection algorithm:
\begin{enumerate}
    \item For a fixed \( \pi(\phi) \), set
    \[
    \psi(\phi,y^n) = \frac{\pi(\phi) P_\phi(y^n)}{\int_\Phi \pi(\phi) P_\phi(y^n) \, d\phi}
    \]
    \item For a fixed \( \psi(\phi,y^n) \), set
    \[
    \pi(\phi) = \frac{ \prod_{y^n} \tilde{\psi}(\phi,y^n)^{P_\phi(y^n)} }{ \int_\Phi \prod_{y^n} \tilde{\psi}(\phi,y^n)^{P_\phi(y^n)} \, d\phi }, \quad 
    \]
    where $\tilde{\psi}(\phi,y^n) = \psi(\phi,y^n) \cdot \frac{P_{\theta^*}(y^n)}{P_\phi(y^n)}$.
\end{enumerate}
Finally, the distribution \( Q \) that achieves \( F_{o,n}(\Theta,\Phi) \) is given by
\[
Q(y^n) = \int_{\phi} \pi^*(\phi) P_\phi(y^n) \, d\phi,
\]
where \( \pi^*(\phi) \) is the final value of \( \pi(\phi) \) at the last iteration of the algorithm.
\end{thm}

Note that, by definition, the minimax regrets associated with the various universal prediction settings satisfy the following relation:
\begin{align}
    \begin{aligned}
        \nonumber
        C_n(\Theta) \equiv F_{o,n}(\Theta,\Theta) \le F_{o,n}(\Theta,\Phi)
        \le F_{o,n}(\Theta,\mathcal{P}) \equiv \Gamma_n(\Theta).
    \end{aligned}
\end{align}
Here, $C_n(\Theta)$ denotes the redundancy-capacity of the well-specified setting, $F_{o,n}(\Theta,\Phi)$ is the misspecified minimax regret when the true distribution lies in $\Phi$, and $\Gamma_n(\Theta)$ is the worst-case regret over all probability distributions on $\mathcal{Y}^n$, denoted by $\mathcal{P}$.
The following relation
\[
    \Gamma_n(\Theta) = F_{o,n}(\Theta,\mathcal{P}) = \log\sum_{y^n}{\max_{\theta}P_\theta(y^n)}
\]
was also demonstrated in (\cite{robust21inference}, Theorem~4) and was shown to be achieved by the NML universal predictor introduced by Shtarkov \cite{shtar1987universal}.
We now provide an alternative proof of Shtarkov’s individual-setting regret result, formulated using principles of the misspecified setting:
\begin{align}
    \begin{aligned}
        \nonumber
        F_{o,n}(\Theta,\mathcal{P}) &= \min_{Q(y^n)} \max_{P(y^n)\in\mathcal{P}} \max_{P_\theta\in\Theta} \sum_{y^n} P(y^n) \log \frac{P_\theta(y^n)}{Q(y^n)} \\
        &= \min_{Q(y^n)} \max_{y^n} \log \frac{\max_{\theta\in\Theta}{P_\theta(y^n)}}{Q(y^n)}.
    \end{aligned}
\end{align}
Replacing the maximization over $y^n$ with a maximization over a mixture of sequences, we have:
\begin{align}
    \begin{aligned}
    \nonumber
        F_{o,n}(\Theta,\mathcal{P}) &= \min_{Q(y^n)} \max_{\pi(y^n)}  \sum_{y^n} \pi(y^n) \log \frac{\max_{\theta\in\Theta}{P_\theta(y^n)}}{Q(y^n)}.
    \end{aligned}
\end{align}
Setting $Q(y^n) = \pi(y^n)$, we obtain:
\begin{align}
    \begin{aligned}
        F_{o,n}(\Theta,\mathcal{P}) &\leq \max_{\pi(y^n)} \sum_{y^n} \pi(y^n) \log \frac{\max_{\theta\in\Theta} P_\theta(y^n)}{\pi(y^n)}.
    \end{aligned}
\end{align}
Applying Jensen's inequality gives the following:
\begin{align}
    \begin{aligned}
        F_{o,n}(\Theta,\mathcal{P}) &\leq \log \sum_{y^n} \pi(y^n) \cdot \frac{\max_{\theta\in\Theta} P_\theta(y^n)}{\pi(y^n)} \\
        &= \log \sum_{y^n} \max_{\theta \in \Theta} P_\theta(y^n) = \Gamma_n(\Theta).
    \end{aligned}
\end{align}
Since $F_{o,n}(\Theta,\mathcal{P}) \geq \Gamma_n(\Theta)$ (as the set of all individual deterministic sequences is a subset of $\mathcal{P}$), we conclude:
\[
F_{o,n}(\Theta,\mathcal{P}) = \Gamma_n(\Theta) = \log \sum_{y^n} \max_{\theta \in \Theta} P_\theta(y^n) \equiv \log K_n,
\]
and clearly the NML is given by,
\begin{align}
    \begin{aligned}
        Q(y^n) = \pi(y^n) = \frac{\max_{\theta\in\Theta}{P_\theta(y^n)}}{K_n}.
    \end{aligned}
\end{align}

In~\cite{online21misspecified}, the authors investigated whether the PAC-style regret \( F_{o,n}^{(\mathrm{PAC})}(\Theta, \Phi) \), relevant from a learning-theoretic perspective where \( \Phi \) is the set of all i.i.d. distributions, is more closely aligned with \( C_n(\Theta) \) or with \( \Gamma_n(\Theta) \). %This question is central to understanding the robustness of universal prediction strategies under model misspecification in practical learning scenarios.
Interestingly, under mild regularity conditions, it has been shown that the misspecified regret is closely approximated by the well-specified regret associated with the hypothesis class \( \Theta \) and a small neighborhood of distributions surrounding it within the model class \( \Phi \). In other words, the complexity of the prediction task is primarily governed by the hypothesis class \( \Theta \), rather than the broader model class \( \Phi \). Consequently, the regret is well approximated by the well-specified capacity of \( \Theta \), and the capacity-achieving prior \( \pi(\phi) \) concentrates its probability mass predominantly on distributions in \( \Theta \). This insight highlights the robustness of universal prediction strategies, even in the presence of model misspecification.

This behavior is formally captured in the following theorem:

\begin{thm}[\cite{online21misspecified}, Theorem 4]
\label{thm:misspecified_online_regret_bound}
Suppose that \( \Phi \subset \mathcal{P} \) is the set of all i.i.d. distributions such that \( C_n(\Phi) = \tau_n \cdot n \), with \( \tau_n \to 0 \). Then, for every \( \epsilon_n \gg \tau_n \), we have:
\begin{align}
    C_n(\Theta) \leq F_{o,n}(\Theta, \Phi) \leq C_n(\Theta_{\epsilon_n}) + o(1)
\end{align}
where \( \Theta_{\epsilon_n} = \{ P_\phi(y^n) \in \Phi : D(P_\phi \| \Theta) \leq \epsilon_n \} \).
\end{thm}

In many relevant scenarios, the condition \( \epsilon_n \to 0 \) holds while still satisfying \( \epsilon_n \gg \tau_n \), as required by Theorem~\ref{thm:misspecified_online_regret_bound}. 

An illustrative example is provided in \cite{online21misspecified}, which considers the Gaussian Location Model (GLM), where \( \Theta \) is a compact subset of \( \mathbb{R}^d \), and the likelihood function is given by
\[
f_\theta(y) = (2\pi)^{-\frac{d}{2}} e^{-\frac{1}{2}\|y - \theta\|^2}.
\]
In this setting, it is shown that \( F_{o,n}(\Theta, \Phi) = C_n(\Theta) + o(1) \), while the individual setting regret term satisfies \( \Gamma_n(\Theta) = C_n(\Theta) + \frac{d}{2}\log{e} + o(1) \). Notably, this result also holds in the PAC setting, where \( \Phi \) is the set of all i.i.d. distributions.

Conversely, (\cite{online21misspecified}, Appendix F.1.) also presents a counterexample demonstrating that the growth rates of \( C_n(\Theta) \) and \( C_n(\Theta_\epsilon) \) can differ significantly. This example extends the GLM to an infinite-dimensional setting, where the observation vector is defined as
\[
Y = \phi + N, \quad N \sim \mathcal{N}(0, I_\infty),
\]
with \( \phi = (\phi_0, \phi_1, \dots) \), such that each \( Y \) is an infinite-dimensional vector. Define the parameter sets:
\[
\Phi = \left\{ \phi : 0 \leq \phi_j \leq 2^{-j}, \; j = 0, 1, \dots \right\},
\]
and
\[
\Theta = \left\{ \phi : 0 \leq \phi_0 \leq 1, \; \phi_1 = \phi_2 = \dots = 0 \right\}.
\]
It can be shown that the following asymptotic behaviors hold:
\[
C_n(\Phi) \asymp \log^2 n, \quad
C_n(\Theta) \asymp \log n, \quad
C_n(\Theta_\epsilon) \asymp \log^2 n,
\]
for any \( \epsilon > 0 \), where
\[
\Theta_\epsilon = [0,1] \times \{0\} \times \cdots \times \{0\} \times [0, 2^{-k}] \times [0, 2^{-k-1}] \times \cdots,
\]
and \( k = -\frac{1}{2} \log_2 \epsilon + O(1) \).

The above results provide a fairly complete picture of the misspecified online setting for unsupervised data. 
In contrast, the supervised online setting introduces additional challenges due to a fundamental limitation: the chain rule generally does not hold for conditional probabilities. Specifically, not every joint distribution $P(y^n|x^n)$ can be factorized into sequential conditionals $P(y_t|x^t, y^{t-1})$. 
As a result, the supervised online case remains largely unresolved, even under well-specified assumptions, with only partial results currently available, see  e.g., \cite{online17learning} analyzed several specific scenarios using a mixture-based approach when the feature sequence is either i.i.d. or generated 
by an adversary with i.i.d. noise satisfying mild regularity conditions in the well-specified setting.

Naturally, the misspecified online learning problem in the supervised setting introduces additional challenges compared to the well-specified case. In what follows, we establish results for the misspecified online minimax regret when the data-generating family $\Phi$ consists of causal conditional models. Interestingly, the resulting characterization involves a directed-information term, an information measure originating in causal communication with feedback, see \cite{massey1990causality}, \cite{kramer1998directed}, \cite{tatikonda2009capacity}, \cite{permuter2009finite}.
In addition, we derive lower bounds on the minimax regret for broader, noncausal classes of $\Phi$. These bounds are obtained under the assumption of a known feature distribution $P(x^n)$ and focus on a simple but practically important class of models, namely, the set of memoryless conditional distribution hypotheses $\Theta$,

The following theorem addresses the misspecified online supervised learning setting under the assumption that the data-generating family $\Phi$ consists of causal conditional distributions.
\begin{thm}
\label{thm:misspecified_online_supervised_causal_setting}
Assume a set of causal conditional data generating distributions $P_\phi(y^n | x^n) \in \Phi$, a hypothesis class consisting of 
conditional models $P_\theta(y^n | x^n) \in \Theta \subseteq\Phi$.
Assume that the feature sequence is i.i.d. by a known probability distribution $P(x)$ and independent of $\Phi$.
Then the online misspecified minimax regret satisfies
\begin{align}
    \begin{aligned}
        R_{o,n}^*(\Theta,\Phi) = \max_{\pi(\phi)}\left( I(\Phi \to Y^n | X^n) - \mathbb{E_{\pi(\phi)}}\left\{D(P_\phi\|\Theta)\right\} \right)
    \end{aligned}
\end{align}
where $I(Y^n \to \Phi | X^n)$ is the directed information from $Y^n$ to $\Phi$ conditioned on $X^n$ and the misspecified penalty term is given by:
\begin{align}
    \begin{aligned}
        D(P_\phi\|\Theta) \equiv \min_{P_\theta\in\Theta}\mathbb{E}_{P_\phi P}\left\{ \log{\frac{P_\phi(Y^n\|X^n)}{P_\theta(Y^n\|X^n)}}  \right\}.
    \end{aligned}
\end{align}
Moreover, the universal sequential predictor is given by
\begin{align}
    \begin{aligned}
        Q(y^n\|x^n) = \int_{\Phi}{\pi(\phi)P_\phi(y^n\|x^n)}\,d\phi.
    \end{aligned}
\end{align}
\end{thm}
\begin{IEEEproof}
    See Appendix \ref{subsec:appendix_misspecified_online_setting}.    
\end{IEEEproof}

Theorem~\ref{thm:misspecified_online_supervised_causal_setting} shows that, when the data-generating family $\Phi$ is causal, the minimax regret in the misspecified online supervised setting equals the causal capacity between $\Phi$ and the label sequence $Y^n$ conditioned on the features $X^n$. This causal capacity coincides with the directed information $I(\Phi \to Y^n | X^n)$, together with a misspecification penalty term given by $\mathbb{E}_{\pi(\phi)}\!\left\{D(P_\phi\Vert\Theta)\right\}$.

In contrast to the previous theorem, which assumed that the data-generating family $\Phi$ consists exclusively of causal conditional distributions, the following result extends the analysis to a more general setting where $\Phi$ may include both causal and noncausal conditional models. To simplify the analysis, we restrict $\Theta$ to memoryless conditional models and derive a lower bound on the corresponding minimax regret.
\begin{thm}
\label{thm:misspecified_online_supervised_setting}
Assume a set of conditional data-generating distributions $P_\phi(y^n | x^n) \in \Phi$, a hypothesis class consisting of 
conditional memoryless models $P_\theta(y | x) \in \Theta$, and a given known feature
distribution $P(x^n)$. Then the misspecified minimax regret under the online setting satisfies the lower bound
    \begin{align}
        \begin{aligned}
        \nonumber
            R_{o,n}^*(\Theta,\Phi) \geq \max_{\pi(\phi)}\left( I(Y^n;\Phi|X^n) - \mathbb{E}_{\pi(\phi)}\left\{ D(P_\phi\|\Theta) \right\} \right)
        \end{aligned}
    \end{align}
where, $D(P_\phi\|\Theta) \equiv \min_{P_\theta\in\Theta}D(P_\phi(Y^n|X^n)\|P_\theta(Y^n|X^n))$.
\end{thm}
\begin{IEEEproof}
    See Appendix \ref{subsec:appendix_misspecified_online_setting}.    
\end{IEEEproof}

Theorem \ref{thm:misspecified_online_supervised_setting} demonstrates that under causal constraints in online learning, the predictor is restricted to base its prediction at time $t$ only on past observations $(x^{t},y^{t-1})$. If the true distributions $P_\phi$ are non-causal, such a predictor cannot match the optimal non-causal mixture achieving the lower bound, leading to additional regret beyond the mutual-information term.
The quantity $I(Y^n ; \Phi | X^n)$ reflects the inherent uncertainty about the parameter $\phi$ given the observed data, while $\mathbb{E}_{\pi}\{D(P_\phi \Vert \Theta)\}$ captures the additional penalty due to misspecification, representing the discrepancy between the true model family $\Phi$ and the hypothesis class $\Theta$.

\begin{cor}
\label{cor:supervised_misspecified_regret_online_setting}
If the families $\Phi$ and $\Theta$ consist of memoryless conditional distributions, then the minimax regret equals:
\begin{align}
    \begin{aligned}
\nonumber
    R_{o,n}^*(\Theta,\Phi) &= \max_{\pi(\phi)} \left( I(Y^n;\Phi|X^n) - \mathbb{E}_{\pi}\{ D(P_\phi \Vert \Theta) \} \right)
    \\ &=\max_{\pi(\phi)} \left( I(\Phi\to Y^n|X^n) - \mathbb{E}_{\pi}\{ D(P_\phi \Vert \Theta) \} \right).
    \end{aligned}
\end{align}
\end{cor}
\begin{IEEEproof}
    See Appendix \ref{subsec:appendix_misspecified_online_setting}.
\end{IEEEproof}
Corollary~\ref{cor:supervised_misspecified_regret_online_setting} shows that, when $\Phi$ is restricted to memoryless conditional distributions, the lower bound established in Theorem~\ref{thm:misspecified_online_supervised_setting} becomes tight. In this case, the conditions of Theorem~\ref{thm:misspecified_online_supervised_causal_setting} are also satisfied, and clearly the directed information and mutual information coincide.

\subsection{Batch Setting}
\label{subsec:misspecified_batch_setting}
In this Section, we present a detailed analysis of the misspecified batch setting. We begin with the unsupervised case and derive a closed-form expression for the minimax regret and the corresponding optimal universal predictor. We then establish tight bounds on the minimax regret relative to the well-specified setting, analogous to Theorem~\ref{thm:misspecified_online_regret_bound} (\cite{online21misspecified}, Theorem 4). In addition, since both the minimax regret and its optimal predictor are 
intractable to evaluate analytically in this setting, we develop an extension of the Arimoto-Blahut algorithm and illustrate our results numerically for the binary Bernoulli model. Finally, we provide preliminary results for the supervised misspecified batch setting.

Our first main result for universal misspecified batch learning is an analytical formulation of the minimax regret, stated as follows:
\begin{thm}\label{ThmFirst}
    The minimax regret of the universal misspecified batch learning setting is given by:
    \begin{align}									
        \begin{aligned}		
            \label{thmRegret}	            
        F_{b,n}(\Theta,\Phi) = \max_{\pi({\phi})}  \left( I(Y_n;\Phi|Y^{n-1}) - \mathbb{E}_{\pi(\phi)} \{  D({{P_{\phi}}\|\Theta}) \} \right)
    \end{aligned}
    \end{align}
    where the divergence term is defined as:
    \begin{align}
        \begin{aligned}
            \nonumber
            D(P_\phi \| \Theta) \equiv \min_{P_\theta \in \Theta} D(P_\phi \| P_\theta)
        \end{aligned}
    \end{align}    
    and the universal distribution for a given $\pi(\phi)$ is given by:
    \begin{align}
        \begin{aligned}
        \nonumber
            Q_{\pi}({y_n|y^{n-1}}) = \frac{\int_{\phi}\pi(\phi)P_{\phi}(y^{n})  \,d\phi}{\int_{\phi}\pi(\phi)P_{\phi}(y^{n-1})  \,d\phi}.
        \end{aligned}
    \end{align}		
\end{thm}
\begin{IEEEproof}
    See Appendix \ref{subsec:appendix_misspecified_batch_setting}.
\end{IEEEproof}    

Theorem \ref{ThmFirst} shows that the minimax regret in misspecified batch learning can be interpreted as a constrained version of the conditional capacity, that is, as a constrained variant of the batch learning regret in the classical well-specified stochastic setting \cite{batch18learning}. Moreover, the structure of this result aligns with that of the misspecified online learning minimax regret, as shown in Theorem \ref{thm:online_misspecified_regret} (\cite{robust21inference}, Theorem 2) and in \cite{online21misspecified}. In both settings the form of the regret is a combination of two terms: 
a mutual-information term between the samples and the data-generating distribution set $\Phi$, given by $I(Y_n;\Phi|Y^{n-1})$ in the batch setting, and an additional penalty term, denoted by $\mathbb{E}_{\pi(\phi)} \{  D({{P_{\phi}}\|\Theta}) \}$.
This penalty quantifies the mismatch between the hypothesis class $\Theta$ and the data-generating set $\Phi$, measured by the closest projected distribution from the set $\Phi$ onto the set of hypotheses $\Theta$, in the sense of KL divergence.

One might mistakenly conjecture that the regret is approximately equal to the conditional capacity between the data samples and the family of data-generating distributions $\Phi$. However, this is not the case, since the optimal prior distribution $\pi(\phi)$ is chosen to maximize the difference between these two terms.

To illustrate this, consider the following extreme example: let $\Phi$ be the class of $d$-parameter multinomial distributions over an alphabet of size $d+1$, while $\Theta$ consists of $d'$-parameter multinomial distributions over an alphabet of size $d'+1$, where $d' < d$. In this setting, $D(P_{\phi}\|P_{\theta}) = 0$ whenever $P_{\phi} \in \Theta$. However, if $P_{\phi} \notin \Theta$, then there exist symbols assigned positive probability under $P_{\phi}$ but zero probability in all $P_{\theta} \in \Theta$, which implies $D(P_{\phi}\|\Theta) = \infty$.

Consequently, to maximize (\ref{thmRegret}), the prior $\pi(\phi)$ must assign zero mass outside the hypothesis class $\Theta$, i.e., $\pi(\phi) = 0$ for all $\phi \notin \Theta$. In other words, in this example the mixture distribution is effectively supported only on $\Theta$. As a result, the regret coincides with the conditional capacity of $\Theta$, which, according to \cite{bernsten04polynomials}, satisfies $C_{c,n}(\Theta) = \frac{d'}{2n} + o(n^{-1})$, matching the batch learning minimax regeret in the well-specified stochastic setting where $\Phi \equiv \Theta$.

Moreover, even in cases where $D(P_{\phi}\|\Theta) < \infty$, we show in Theorem \ref{Thm2} that the optimizing prior $\pi(\phi)$ still concentrates most of its mass on $\Theta$, rather than on the full family $\Phi$. Thus, the effective complexity of the problem is governed by the hypothesis class $\Theta$, rather than by the larger set of data-generating distributions $\Phi$. This behavior is consistent with the misspecified online setting, as shown in Theorem \ref{thm:misspecified_online_regret_bound} (\cite{online21misspecified}, Theorem 4), where the complexity is dictated similarly by the hypothesis class rather than the full model family.

In addition, Theorem \ref{ThmFirst} shows that the universal distribution $Q$ is given by a mixture over the set of the data-generating distributions, consistent with all known unsupervised universal learning settings. This includes online learning in both the classical well-specified stochastic and deterministic individual sequence settings \cite{universal98prediction}, the misspecified setting \cite{robust21inference}\cite{online21misspecified}, and batch learning in the well-specified stochastic setting \cite{batch18learning}.

Our next main contribution in this setting is to lower bound the minimax regret by the conditional capacity of the distribution class $\Theta$. Moreover, we upper bound the regret by approximately the conditional capacity of a slight enlargement of $\Theta$. 
The proof of this upper bound relies on the following lemma, which bounds
\[
J(\pi) \equiv I(Y_n;\Phi|Y^{n-1}) - \mathbb{E}_{\pi}\{D(P_\phi\|\Theta)\}.
\]
The argument is similar to that used in the proof of Lemma 11 in \cite{online21misspecified}.

\begin{lem}
    \label{lemma1}
    For any $\pi_{0}(\phi)$, $\pi_{1}(\phi)$ and $\lambda \in [0,1]$ 
    \begin{align}
        \begin{aligned}
            J(\lambda \pi_1 + (1-\lambda) \pi_{0}) \leq \lambda J(\pi_1) + (1-\lambda) J(\pi_{0}) + h(\lambda)
        \end{aligned}
    \end{align}
    where $h(\lambda)$ is the binary entropy of a Bernoulli distribution $Ber(\lambda)$.
\end{lem}	
\begin{IEEEproof}
    See Appendix \ref{subsec:appendix_misspecified_batch_setting}.
\end{IEEEproof}

 We are now ready to state the following Theorem which bounds the minimax regret in terms of the conditional capacities of the set $\Theta$ and a slightly larger set:
\begin{thm}\label{Thm2}
    Suppose $\Phi$ and $\Theta$ are sets of distributions s.t. $C_{c,n}(\Phi) = \tau_n \to 0$ and $\Theta \subseteq \Phi$. Then for every $\epsilon_n \gg \tau_n$ we have
    \begin{align}
        \begin{aligned}
            \label{RegretBounds}
            C_{c,n}(\Theta) \leq F_{b,n}(\Theta,\Phi) \leq C_{c,n}(\Theta_{{\epsilon}_n}) + o(1)
        \end{aligned}
    \end{align}		
    where $\Theta_{\epsilon} \equiv \{P_{\phi} \in \Phi: D(P_\phi||\Theta) < \epsilon \}$.
\end{thm}
\begin{IEEEproof}
    See Appendix \ref{subsec:appendix_misspecified_batch_setting}.
\end{IEEEproof}

%Theorem \ref{Thm2} provides for the batch learning case a similar result to what was shown for the online learning under misspecification in (\cite{online21misspecified}, Theorem 4). 

Interestingly, it should be noted that while $\epsilon_n \gg \tau_n = C_{c,n}(\Phi)$, it does not mean that the conditional capacity of $\Theta_{\epsilon_n}$ is greater than the conditional capacity of $\Phi$. Its true meaning is that only a shell extension to the set $\Theta$, quantified by $\epsilon_n$, of distributions from $\Phi$ affects the minimax regret performance. If in addition $\epsilon_n \to 0$, the set $\Theta_{\epsilon_n}$ is a small extension of $\Theta$, implying that the resulting conditional capacities of $\Theta$ and $\Theta_{\epsilon_n}$ might be close. An illustration of this phenomenon is given in Figure \ref{FigIllustration}.

In many interesting examples we indeed have $\epsilon_n \to 0$ and the conditional capacities coincide for large $n$. Such an example is where the observations come from a distribution in the set $\Phi$ of $d$-parameters multinomial distributions of the form: $(\phi_0,\phi_1,\dots,\phi_d)$, s.t $\sum_{k=0}^{d} \phi_k=1$. In \cite{bernsten04polynomials} it was shown that the minimax regret, which is the conditional capacity of $\Phi$ in the well-specified stochastic setting of batch learning equals to $C_{c,n}(\Phi) = \frac{d}{2n} + o(n^{-1})$. Therefore, by choosing $\epsilon_n = \frac{1}{n^{1-\alpha}}$, for any $0<\alpha<1$, we have both $\epsilon_n \gg C_{c,n}(\Phi)$ and $\epsilon_n \to 0$.    
As noted, this may imply that $C_{c,n}(\Theta_{\epsilon_n}) \to C_{c,n}(\Theta)$ and according to a sandwich argument and (\ref{RegretBounds}), the minimax regret tends to $F_{b,n}(\Theta,\Phi) \to C_{c,n}(\Theta)$, which is the minimax regret of the well-specified stochastic batch learning setting, \cite{batch18learning}. 

To demonstrate this phenomenon, we can choose as an example $d=1$, i,e., the Bernoulli distribution $Ber(\phi)$, where $\phi\in[0,1]$, to be the set of all the data generating distributions $\Phi$, and the set of hypotheses $\Theta$, to be the set of all $Ber(\theta)$, where $\theta\in[a,b]$ s.t. $0 \leq a < b \leq 1$. In this case, assuming the data samples are i.i.d and $\epsilon$ is small enough, it can be verified that $\Theta_{\epsilon}$ is the set of all $Ber(\theta_\epsilon)$, where $\theta_\epsilon \in\left[a- \delta_\epsilon(a),b+\delta_\epsilon(b)\right]$ and $\delta_\epsilon(c) = \sqrt{2{c(1-c){\epsilon}}}$ for $c\in[0,1]$, i.e., a small extension of the set $\Theta$. Note that in this example the  conditional capacity is a continuous function of $\epsilon$ (it is a composition of elementary functions). Thus, by setting $\epsilon = \epsilon_n$ as explained above, we get $C_{c,n}(\Theta_{\epsilon_n}) \to C_{c,n}(\Theta)$, and finally the minimax regret is approximately equal to $F_{b,n}(\Theta,\Phi) \approx C_{c,n}(\Theta)$.

Another, even more extreme example arises in the misspecified batch setting over the GLM. In this case, when the true distribution belongs to the PAC class, namely, the set of all i.i.d. distributions with finite second moment, Mourtada~\cite{mourtada2022improper} showed that the misspecified regret coincides \emph{exactly} with the well-specified conditional redundancy-capacity expression, rather than only admitting an asymptotic approximation. Specifically,
\[
F_{b,n}^{(\mathrm{PAC})}(\Theta,\Phi)
= C_{c,n}(\Theta)
= \frac{d}{2}\log\left(1+\frac{1}{n}\right)
\to \frac{d}{2n}.
\]

These observations mean that under the conditions specified, the minimax regret in the misspecified setting converges to the regret in the case where the data generating distributions are approximately the ``smaller'' set of hypotheses $\Theta$ and not the ``bigger'' set $\Phi$. Another consequence is the fact that the universal distribution $Q$, is approximately a mixture distribution where the prior conditional capacity achieving distribution, $\pi(\phi)$, is concentrated mostly over the set $\Theta$.

%Note that $\epsilon_n \to 0$ does not necessarily imply that $C_{c,n}(\Theta_{\epsilon_n}) \to C_{c,n}(\Theta)$, although such pathological cases are typically rare. An extreme example illustrating this phenomenon was given in \cite[Appendix F.1]{online21misspecified} for the online setting, as discussed in Section \ref{subsec:misspecified_online_setting}.
Note that $\epsilon_n \to 0$ does not necessarily imply that $C_{c,n}(\Theta_{\epsilon_n}) \to C_{c,n}(\Theta)$, although such pathological cases are typically rare. An extreme example illustrating this phenomenon was given in (\cite{online21misspecified}, Appendix F.1) for the online setting, as discussed in Section \ref{subsec:misspecified_online_setting}.

\begin{figure}[ht]        
    \centering
    \includegraphics[width=0.45\textwidth]{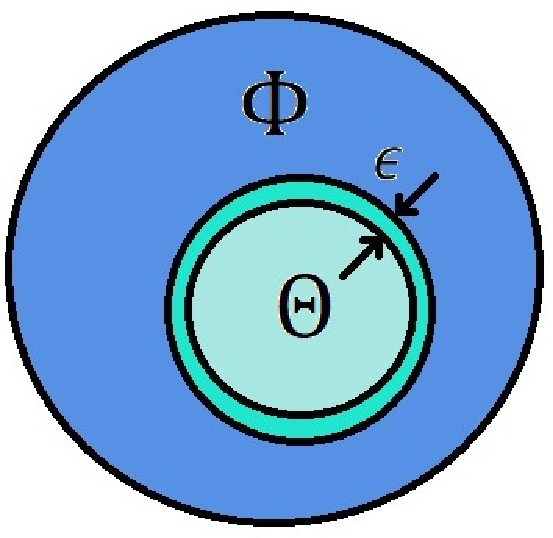}
    \caption{Geometric illustration of Theorems~\ref{thm:misspecified_online_regret_bound} and~\ref{Thm2}. Only an \( \epsilon \)-shell extension of the hypothesis set \( \Theta \) into \( \Phi \) contributes to the minimax regret.}
    \label{FigIllustration}
\end{figure}

\emph{Arimoto-Blahut Algorithm Extension:}\label{Arimoto Blahut Algorithm Extension}
As shown above, the regret can be interpreted as a constrained version of the conditional capacity between $Y_n$ and $\Phi$, where $\pi(\phi)$ serves as the capacity-achieving prior distribution. 
In general, obtaining closed-form expressions for either the capacity-achieving prior or the resulting capacity is analytically intractable.
In classical communication theory and, more recently, in universal prediction theory, several extensions of the Arimoto-Blahut algorithm \cite{blahut72,arimoto72} have been developed to numerically compute capacity-achieving priors and their associated capacities \cite{combined24batch_online,arimoto2004blahut_alg,arimoto2018blahut_alg,
arimoto2013blahut_alg,arimoto2019blahut_alg,robust21inference}. 

Motivated by this line of work, we develop an Arimoto-Blahut type iterative algorithm to numerically evaluate the prior $\pi(\phi)$ and the regret $F_{b,n}(\Theta,\Phi)$ in the misspecified batch learning setting. Corollary~\ref{cor:thmArimotoBlahut} provides upper and lower bounds on the regret, which serve as convergence criteria for the proposed algorithm. Furthermore, the structure of these bounds expressed as differences between two conditional divergences naturally guides the form of the iterative update 
steps, as detailed in the sequel.
\begin{cor}	    
    \label{cor:thmArimotoBlahut}
    The minimax regret of the misspecified batch learning setting holds the following for any $\Phi$, $\Theta$ and $\pi(\phi)$:		
    \begin{align}									
        \begin{aligned}	
            \nonumber
            R_L \leq F_{b,n}(\Theta,\Phi) \leq R_U 
          \end{aligned}
    \end{align}           
    where
    \begin{align}
        \begin{aligned}
            \nonumber
            R_L &\equiv \mathbb{E}_{\pi(\phi)}\left\{ D(P_\phi \| Q_{\pi}) -  D(P_\phi \| \Theta) \right\}, \\
            R_U &\equiv \max_{P_{\phi}}\left( D(P_\phi \| Q_{\pi}) -  D(P_\phi \| \Theta) \right),
        \end{aligned}
    \end{align}    
    and
    \begin{align}
    \nonumber
        Q_\pi(y^n) = \int_{\Phi} \pi(\phi) P_\phi(y^n)\, d\phi.
    \end{align}
\end{cor}
\begin{IEEEproof}
    See Appendix \ref{subsec:appendix_misspecified_batch_setting}.
\end{IEEEproof}
Using Corollary~\ref{cor:thmArimotoBlahut}, we extend the Arimoto-Blahut algorithm to the misspecified batch learning setting, as described in Algorithm~\ref{alg:ArimotoBlahutAlg_misspecified}. The inputs to the algorithm are the batch size $n$, an optimization parameter $\lambda$, a convergence accuracy $\epsilon$, the sets $\Phi$ and $\Theta$, and an initial prior distribution $\pi^{(0)}(\phi)$ (a uniform prior over $\Phi$ is a common practical choice). In the initialization step, we compute the corresponding lower and upper bounds, $R_L$ and $R_U$, under the initial prior $\pi^{(0)}(\phi)$. 

An iterative procedure is then applied to the prior $\pi^{(i)}(\phi)$, where $i$ indexes the iteration, until convergence is achieved, i.e., when $R_U - R_L \le \epsilon$. The output of the algorithm is the capacity-achieving prior $\pi(\phi)$, and the regret can be approximated by $F_{b,n}(\Theta,\Phi) \approx \frac{R_L + R_U}{2}$.

Finally, we note that the second divergence term in the algorithm (and in the regret), $D(P_\phi \,\|\, \Theta)$, exponentially lowers the prior mass assigned to distributions $P_\phi$ that lie far from the hypothesis class $\Theta$, as implied by Theorem~\ref{Thm2}.

\begin{algorithm}        
\caption{Arimoto-Blahut Algorithm for Misspecified Batch Learning}
\label{alg:ArimotoBlahutAlg_misspecified}
\begin{algorithmic}
\State \textbf{Input}: \\$n, \lambda, \epsilon, \Phi=\{\phi_m\}_{m=1}^{M_\phi}, \Theta=\{\theta_m\}_{m=1}^{M_\theta}, \pi^{(0)}(\phi)$
\State \textbf{Output}: \\$\pi(\phi)$
\State \textbf{Initialization}:
\State $i \gets 0$
\State {$R_{U}^{(0)} = \max_{\phi} \left( D(P_\phi \| Q_{\pi^{(0)}}) - D(P_\phi \| \Theta) \right)$}   
\State {$R_{L}^{(0)} = \mathbb{E}_{\pi^{(0)}(\phi)}\left\{ D(P_\phi \| Q_{\pi^{(0)}}) - D(P_\phi \| \Theta) \right\}$}
\State \textbf{Loop}:
\begin{algorithmic}
\While {$R^{(i)}_{U} - R^{(i)}_{L} > \epsilon$}         
        \State \textbf{Iterate:}
        \State $\tilde{\pi}^{(i+1)}(\phi_j) = \pi^{(i)}(\phi_j) \cdot 
        e^{\lambda \left(D(P_{\phi_j} \| Q_{\pi^{(i)}}) - D(P_{\phi_j} \| \Theta) \right)}$
        \State $\pi^{(i+1)}(\phi_j) = \frac{\tilde{\pi}^{(i+1)}(\phi_j)}{ \sum_{j'=1}^{M_{\phi}} \tilde{\pi}^{(i+1)}(\phi_{j'})}$        
        \State \textbf{Bounds Update:}
        \State {$R_{U}^{(i+1)} = \max_{\phi} \left( D(P_\phi \| Q_{\pi^{(i+1)}}) - D(P_\phi \| \Theta) \right)$}        
        \State {$R_{L}^{(i+1)} = \mathbb{E}_{\pi^{(i+1)}(\phi)}\left\{ D(P_\phi \| Q_{\pi^{(i+1)}}) - D(P_\phi \| \Theta) \right\} $}
        \State $i \gets i+1$
\EndWhile
\State \textbf{end}
\State \textbf{Return}: \\$\pi(\phi)$
\end{algorithmic}
\end{algorithmic}
\end{algorithm}
To demonstrate the proposed Arimoto-Blahut algorithm extension, we apply the procedure to Bernoulli distributions, where $y \in \{0,1\}$. In this example, $\Phi$ denotes the set of all Bernoulli distributions with parameter $\phi \in [\phi_{\min}, \phi_{\max}]$, and the hypothesis set $\Theta$ is restricted to the interval $\theta \in [a,b] \subset \Phi$, where $0 \le \phi_{\min} \le a \le b \le \phi_{\max} \le 1$. 

The numerical values of the resulting regret for several representative choices of $\Phi$ and $\Theta$ are summarized in Table~\ref{NumericalResultsTable}.

\begin{table}[!ht]
\centering
%\small
\begin{tabular}{||c c c c||}
\hline
$\Phi$ & $\Theta$ & $n$ & Regret \\
\hline\hline

$[a-\delta_n,\,b+\delta_n]$ & $[a-\delta_n,\,b+\delta_n]$ & $10^2$ & $\nicefrac{0.9171}{2n}$ \\

$[0,1]$ & $[a,b]$ & $10^2$ & $\nicefrac{0.8728}{2n}$ \\

$[a,b]$ & $[a,b]$ & $10^2$ & $\nicefrac{0.8710}{2n}$ \\

$[0,1]$ & $[0,1]$ & $10^2$ & $\nicefrac{0.9908}{2n}$ \\

$[0,1]$ & $[0.01,0.99]$ & $10^2$ & $\nicefrac{0.9766}{2n}$ \\

$[0.01,0.99]$ & $[0.01,0.99]$ & $10^2$ & $\nicefrac{0.9763}{2n}$ \\

$[a-\delta_n,\,b+\delta_n]$ & $[a-\delta_n,\,b+\delta_n]$ & $10^3$ & $\nicefrac{0.9837}{2n}$ \\

$[0,1]$ & $[a,b]$ & $10^3$ & $\nicefrac{0.9816}{2n}$ \\

$[a,b]$ & $[a,b]$ & $10^3$ & $\nicefrac{0.9798}{2n}$ \\

$[0,1]$ & $[0,1]$ & $10^3$ & $\nicefrac{1.0027}{2n}$ \\

$[0,1]$ & $[0.01,0.99]$ & $10^3$ & $\nicefrac{0.9970}{2n}$ \\

$[0.01,0.99]$ & $[0.01,0.99]$ & $10^3$ & $\nicefrac{0.9970}{2n}$ \\

\hline
\end{tabular}

\caption{Summary of Arimoto-Blahut numerical results for various $(\Phi,\Theta)$ and sample sizes. Parameters: $a=0.25$, $b=1-a$, and $\delta_n=\sqrt{2a(1-a)\,\epsilon_n}$ with $\epsilon_n = n^{\alpha-1}$ and $\alpha=0.1$.}

\label{NumericalResultsTable}
\end{table}

To compare the minimax regret in the well-specified and misspecified stochastic settings, consider the following three scenarios with $n = 10^3$:
\begin{enumerate}[label=(\alph*)]
    \item \emph{Well-specified stochastic setting:}
    \[
        \Phi = \Theta = [0.25, 0.75].
    \]

    \item \emph{Misspecified stochastic setting:}
    \[
        \Phi = [0,1], \quad \Theta = [0.25, 0.75].
    \]

    \item \emph{Well-specified stochastic setting:}
    \[
        \Phi = \Theta_{\epsilon_n} = [0.25-\delta_n,\, 0.75+\delta_n],
    \]
    where $\epsilon_n = n^{\alpha-1}$ with $\alpha = 0.1$, so that
    \[
        \delta_n
        = \sqrt{2 a(1-a)\,\epsilon_n}
        = \sqrt{2 b(1-b)\,\epsilon_n}
        \approx 0.0274.
    \]
    Note that
    \[
        \Theta_{\epsilon_n}
        = \left\{
            P_\phi \in \Phi : D(P_\phi \,\|\, \Theta) < \epsilon_n
          \right\}
         ~\text{for } n = 10^3.
    \]
\end{enumerate}

As expected from Theorem~\ref{Thm2}, the minimax regrets of the three settings satisfy
\begin{align}
\label{Thm2Demonstration}
    C_{c,n}(\Theta)
    \equiv
    \underbrace{F_{b,n}(\Theta)}_{\frac{0.9798}{2n}}
    <
    \underbrace{F_{b,n}(\Theta, \Phi)}_{\frac{0.9816}{2n}}
    <
    \underbrace{F_{b,n}(\Theta_{\epsilon_n})}_{\frac{0.9837}{2n}}
    \equiv
    C_{c,n}(\Theta_{\epsilon_n}).
\end{align}
A comparable result for $n = 10^2$ and $\delta_n \approx 0.0771$ appears in Table~\ref{NumericalResultsTable}.

Numerically, we also obtain
\[
    F_{b,n}(\Theta,\Phi) - C_{c,n}(\Theta)
    \approx \frac{0.0018}{2n}
\]
for both $n=10^2$ and $n=10^3$. Combined with (\ref{Thm2Demonstration}), this indicates that there exists a sequence $\epsilon_n \to 0$ such that $C_{c,n}(\Theta_{\epsilon_n}) \to C_{c,n}(\Theta)$ and $F_{b,n}(\Theta,\Phi) \approx C_{c,n}(\Theta)$.

A further illustration comes from the similarity of the capacity-achieving prior distributions $\pi(\phi)$ across these examples, shown in Figure~\ref{FigurePriorTheta_025_075_N1000}. In the well-specified case $\Phi = \Theta$, the prior assigns zero mass outside $[0.25, 0.75]$, while in the misspecified case $\Phi = [0,1]$, it decays rapidly outside this interval as expected. 

\begin{figure}[ht]
    \centering
    \includegraphics[width=0.4\textwidth]{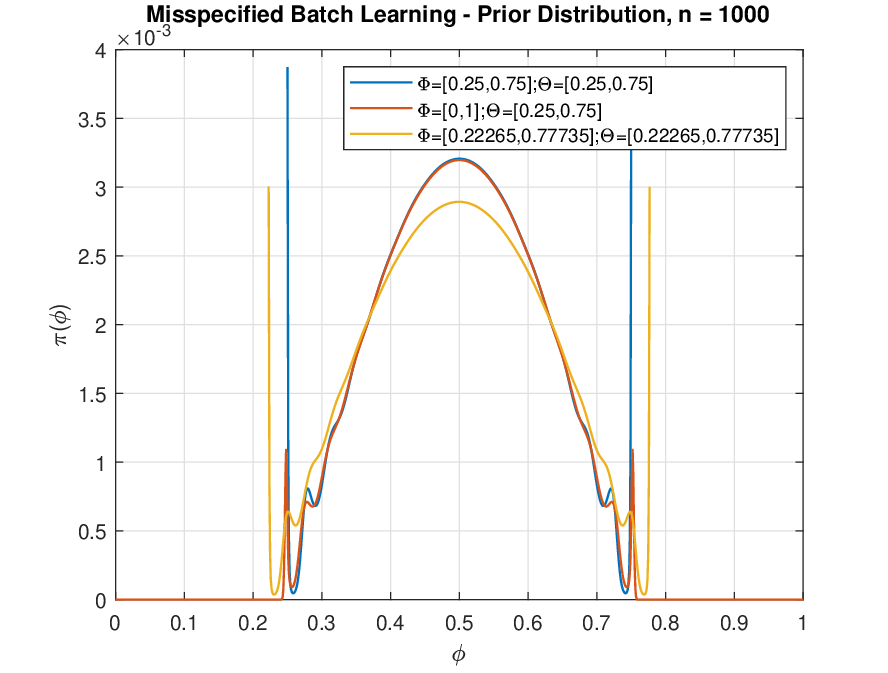}
    \caption{Comparison of the capacity achieving prior distributions in the well-specified stochastic and misspecified settings for cases (a), (b), and (c), with \( n = 10^3 \).}
    \label{FigurePriorTheta_025_075_N1000}
\end{figure}

\emph{add--$\beta$ Factor Analysis:}
Another quantity of interest is the \emph{add--$\beta$} factor, expressed as a function of the empirical distribution
\[
\hat{p} \equiv \frac{\sum_{t=1}^{n-1} y_t}{n-1},
\]
computed from the first $n-1$ samples. Since $\hat{p}$ is the sufficient statistic for the universal predictor $Q_{\pi}(y_n=1 | y^{n-1})$, we may equivalently write
\begin{align}
\begin{aligned}
Q_{\pi}(y_n=1 | y^{n-1})
    &= Q_{\pi}\Big(y_n=1 | \sum_{t=1}^{n-1} y_t\Big) \\
    &= Q_{\pi}\big(y_n=1 | \hat{p}\big) \\
    &= \frac{(n-1)\hat{p} + \beta}{\,n-1 + 2\beta\,}.
\end{aligned}
\end{align}
A straightforward algebraic rearrangement yields
\begin{align}
\beta(\hat{p}) 
    = (n-1)\frac{Q_{\pi}(y_n=1 | \hat{p})-\hat{p}}{1 - 2\,Q_{\pi}(y_n=1 | \hat{p})}.
\end{align}

Figure~\ref{FigureBeta_Theta_001_099_N100} illustrates the resulting empirical \emph{add--$\beta$} factor $\beta(\hat{p})$ for $n=10^2$ in the following scenarios:
\begin{enumerate}[label=(\alph*)]
    \item Misspecified stochastic setting: $\Phi=[0,1]$ and $\Theta = [0.01,0.99]$.
    \item Well-specified stochastic setting: $\Phi=\Theta=[0.01,0.99]$.
    \item Well-specified stochastic setting: $\Phi=\Theta=[0,1]$.
\end{enumerate}

\begin{figure}[ht]
    \centering
    \includegraphics[width=0.4\textwidth]{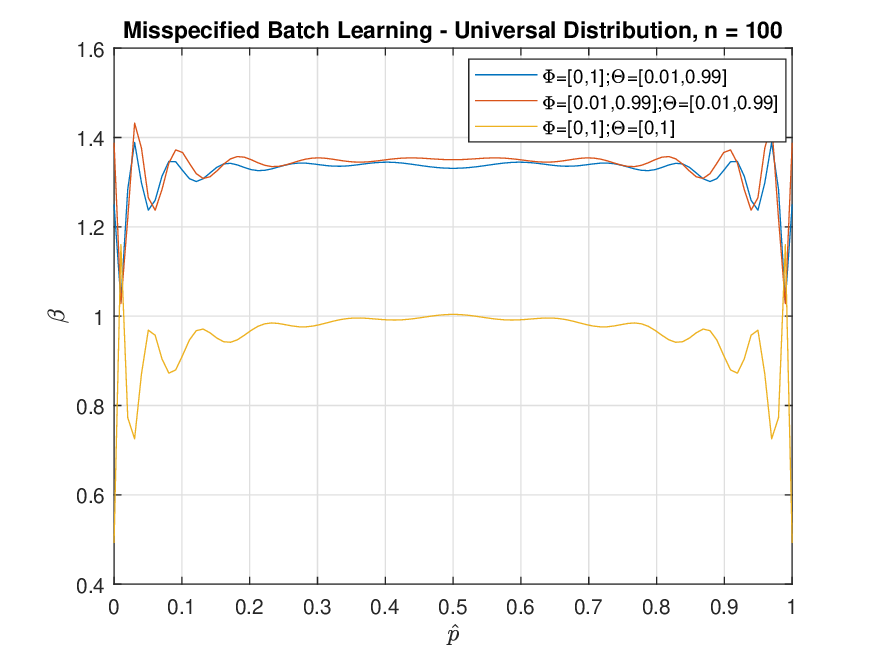}
    \caption{Comparison of the add-\(\beta\) bias factor across settings (a), (b), and (c), for \( n = 10^2 \).}
    \label{FigureBeta_Theta_001_099_N100}
\end{figure}

As shown in the figure, the settings (a) and (b) yield nearly identical \emph{add--$\beta$} factors, both fluctuating around approximately $1.3$. This behavior aligns with Komaki's analysis~\cite{komaki2012minimax}, which demonstrated that for multinomial models in the interior of a $d$-dimensional well-specified parameter space, the constant \emph{add--$\beta$} estimator achieves asymptotically the optimal leading regret term $d/(2n)$ when $\beta = 1 + \sqrt{1/6} \approx 1.4$. For completeness, we note that related properties of the \emph{add--$\beta$} family have also been examined in the batch multinomial setting by Krichevsky and Trofimov~\cite{krichevskiy1998laplace}, and more recently in analyses by Bondaschi and Gastpar~\cite{BondaschiGastpar2024ISIT,BondaschiGastpar2024arXiv}, although these settings differ from the misspecified framework considered here.

In contrast, setting (c) produces a noticeably smaller bias, with \emph{add--$\beta$} factors fluctuating around approximately $1$. This distinction illustrates that the magnitude of the \emph{add--$\beta$} factor is primarily dictated by the class of underlying hypotheses $\Theta$. In particular, smaller \emph{add--$\beta$} values indicate lower model complexity, since simpler classes require less smoothing to achieve minimal regret.

Finally, consider the extreme case $\hat{p} = 0$. Let $\beta(s)$ denote the empirical \emph{add--$\beta$} factor extracted under scenario $s \in \{a,b,c\}$. In the misspecified scenario (a), we obtain $\beta(a) \approx 1.25$. For the well-specified stochastic scenarios (b) and (c), the corresponding values are $\beta(b) \approx 1.38$ and $\beta(c) \approx 0.49$, respectively. These differences are consistent with the theoretical behavior of predictive distributions near the boundary: when the empirical observation resembles a boundary case (i.e., $\hat{p} \approx 0$), a model that excludes boundary parameters, as in (b), must rely on stronger smoothing to mitigate the mismatch, resulting in larger \emph{add--$\beta$} values. In contrast, when the hypothesis class includes the boundary, as in (c), such observations are fully compatible with the model, and considerably less smoothing is required, leading to much smaller \emph{add--$\beta$} factors.

We further note that, in the classical well-specified stochastic online prediction setting with $\Theta = [0,1]$, the optimal universal predictor is the estimator introduced by Krichevsky and Trofimov~\cite{krichevskiy1998laplace}, which corresponds asymptotically to a constant smoothing parameter $\beta = 0.5$. When $\hat{p} = 0$, the batch and online predictors coincide in functional form, and scenario (c) exhibits precisely this behavior, yielding the empirical value $\beta(c) \approx 0.49$ in agreement with this classical result.

\emph{Supervised Batch Learning:}
The above analysis provides a detailed and comprehensive characterization of the misspecified batch setting in the unsupervised case. We now present preliminary results for the misspecified batch setting in the supervised scenario. In this setting, each data feature, denoted by $x_n$, is associated with a corresponding label $y_n$. The objective is to assign a universal conditional distribution to the next label $y_n$, given its associated feature $x_n$ and the training data $(x^{n-1}, y^{n-1})$. We denote this universal predictor by $Q(y_n | x^n, y^{n-1})$.
Throughout this setting, we assume that the features are generated according to a known distribution $P(x^n)$.

Under the above definitions, the minimax regret of the misspecified supervised batch learning setting is given by the following Theorem:
\begin{thm}
\label{thm:supervised_batch_learning}
The minimax regret of the problem defined above is given by:		
\begin{align}									
    \begin{aligned}		
        \nonumber    
        R_{b,n}^*(\Theta,\Phi) &= \max_{\pi({\phi})}  \Big( I(Y_n;\Phi|X^n,Y^{n-1}) 
        \\&- \mathbb{E}_{\pi(\phi)} \{  D({{P_{\phi}}\|\Theta}) \} \Big)
\end{aligned}
\end{align}
where
\begin{align}
    \begin{aligned}
        \nonumber
        D({{P_{\phi}}\|\Theta}) \equiv \min_{P_\theta\in\Theta}{D(P_\phi(Y_n|X^n,Y^{n-1}) \| P_\theta(Y_n|X^n,Y^{n-1}))}
    \end{aligned}
\end{align}
and the universal distribution for a given $\pi(\phi)$ is given by:
\begin{align}
    \begin{aligned}
        \nonumber
        Q_{\pi}({y_n|x^n,y^{n-1}}) = \frac{\int_{\phi}\pi(\phi)P_{\phi}(y^{n}|x^n)  \,d\phi}{\int_{\phi}\pi(\phi)P_{\phi}(y^{n-1}|x^{n})  \,d\phi}.
    \end{aligned}
\end{align}
\end{thm}
\begin{IEEEproof}        
    See Appendix \ref{subsec:appendix_misspecified_batch_setting}.
\end{IEEEproof}    

\subsection{Combined Batch and Online Setting}        
Let us now extend the basic problem formulation to the setting of a universal prediction of $l$ future outcomes, denoted by $y^l = y_{n+1}, y_{n+2}, \dots, y_{n+l}$, given $n$ previously observed training samples $y^n$ under the misspecification framework. This formulation naturally constitutes a combined batch and online prediction setting: the learner first observes a batch of $n$ samples and then sequentially predicts the next $l$ outcomes. A closely related extension in the stochastic setting was introduced and analyzed in \cite{combined24batch_online}, and further investigated in \cite{BondaschiGastpar2024ISIT}, \cite{BondaschiGastpar2024arXiv}, \cite{BondaschiGastpar2025ITW}, \cite{BondaschiGastpar2025arXiv}.

To this end, we define the minimax regret as
\begin{align}
    F_{n,l}(\Theta,\Phi)
    = 
    \min_{Q} \max_{P_{\phi} \in \Phi}
    \frac{1}{l}
    \sum_{y^{n+l}}
    P_{\phi}(y^{n+l})
    \log
    \frac{P_{\theta^*}(y^{l} | y^{n})}
         {Q(y^l | y^{n})},
\end{align}
where $P_{\theta^*}(y^{l} | y^{n})$ denotes the projection of 
$P_\phi(y^{l} | y^{n})$ onto the hypothesis class $\Theta$, that is,
\begin{align}
    P_{\theta^*}
    = 
    \arg\min_{P_\theta\in\Theta}
    D_{l,n}\left(
        P_\phi(Y^l | Y^n)
        \|
        P_\theta(Y^l | Y^n)
    \right).
\end{align}

Following the same steps used in the derivation of the misspecified batch learning setting,
and using the definitions above, we obtain the following theorem.

\begin{thm}\label{thmCombinedSetting}
The minimax regret for the combined batch-online prediction problem described above is given by
\begin{align}
\nonumber
    F_{n,l}(\Theta,\Phi)
    &=
    \max_{\pi(\phi)}
    \frac{1}{l}
    \Big(
        I(Y^l ; \Phi | Y^n)
     \\&   -
        \mathbb{E}_{\pi(\phi)}
        \left\{
            D_{n,l}(P_{\phi}\,\|\,\Theta)
        \right\}
    \Big),
\end{align}
where
\begin{align}
    D_{n,l}(P_{\phi}\|\Theta)
    \equiv
    \min_{P_{\theta} \in \Theta}
    D\left(
        P_{\phi}(Y^{l} | Y^n)
        \|
        P_{\theta}(Y^{l} | Y^n)
    \right),
\end{align}
and for a given prior $\pi(\phi)$, the universal predictor is
\begin{align}
\label{CombinedUniversalDistribution}
    Q_{\pi}(y^{l} | y^{n})
    =
    \frac{
        \int \pi(\phi)P_{\phi}(y^{n+l})\,d\phi
    }{
        \int \pi(\phi)P_{\phi}(y^{n})\,d\phi
    }.
\end{align}
\end{thm}

The form of the minimax regret in Theorem~\ref{thmCombinedSetting} is consistent with the expression obtained for the combined batch and online prediction in the well-specified stochastic framework~\cite{combined24batch_online}. In particular, for $l=1$, the result reduces to Theorem~\ref{ThmFirst}, corresponding to the misspecified batch learning problem, while for $l \gg n$, it approaches the minimax regret associated with misspecified online learning, as derived in Theorem \ref{thm:online_misspecified_regret}.

We now turn to establishing upper and lower bounds for the misspecified combined batch and online minimax regret. These bounds can be viewed as a natural generalization of Theorem~\ref{thm:misspecified_online_regret_bound} (\cite{online21misspecified}, Theorem~4) and Theorem~\ref{Thm2}, which address the bounds for the misspecified online and misspecified batch minimax regrets, respectively.

\begin{thm}\label{thm:thmCombinedSettingBounds}
Suppose $\Theta \subseteq \Phi$ are families of distributions and the data samples are i.i.d. Then the minimax regret in the combined batch online setting satisfies
\begin{align}
    C_{n,l}(\Theta)
    \le
    F_{n,l}(\Theta,\Phi)
    \le
    \frac{1}{l}\sum_{t=1}^l F_{b,n+t}(\Theta,\Phi),
\end{align}
where $C_{n,l}(\Theta)$ denotes the well-specified stochastic capacity of the combined setting, and $F_{b,m}(\Theta,\Phi)$ is the misspecified batch minimax regret for a sample of size $m$.

Moreover, assume that the conditional capacity of the data-generating family satisfies  
$C_{c,n}(\Phi) \equiv \tau_n \to 0$.  
Then, for any sequence $\{\epsilon_{n+t}\}_{t=1}^l$ with $\epsilon_{n+t} \gg \tau_{n+t}$, we have
\begin{align}
    C_{n,l}(\Theta)
    \le
    F_{n,l}(\Theta,\Phi)
    \le
    \frac{1}{l}\sum_{t=1}^l C_{c,n+t}(\Theta_{\epsilon_{n+t}})
    + o(1),
\end{align}
where $\Theta_{\epsilon} \equiv \{P_{\phi}\in\Phi : D(P_{\phi}\,\|\,\Theta) < \epsilon \}$.
\end{thm}
\begin{IEEEproof}
    See Appendix \ref{subsec:appendix_misspecified_combined_setting}.
\end{IEEEproof}

To illustrate these results, consider the case where $\Phi$ is the family of
$d$-dimensional Multinomial distributions and $\Theta \subset \Phi$ is the
subset of $d'$-dimensional Multinomial distributions, with $d' < d$. As shown
in~\cite{combined24batch_online}, the well-specified combined batch-online
capacity in this setting is
\[
C_{n,l}(\Theta)
=
\frac{1}{l}\left[
    \sum_{t=1}^l \frac{d'}{2(n+t)}
    + o\left(\frac{1}{n+t}\right)
\right].
\]
This expression serves as the lower bound on the misspecified combined batch and online minimax regret according to Theorem~\ref{thm:thmCombinedSettingBounds}.

In the misspecified batch setting, we also established that
\[
F_{b,n+t}(\Theta,\Phi)
= \frac{d'}{2(n+t)} + o\left(\frac{1}{n+t}\right).
\]
Therefore, applying Theorem~\ref{thm:thmCombinedSettingBounds} yields the upper bound
\begin{align}
    \begin{aligned}
        \nonumber
        \frac{1}{l}\sum_{t=1}^l F_{b,n+t}(\Theta,\Phi)
        &=
        \frac{1}{l}\left[
            \sum_{t=1}^l \frac{d'}{2(n+t)}
            + o\left(\frac{1}{n+t}\right)
        \right]
    \end{aligned}
\end{align}
which coincides with the lower bound $C_{n,l}(\Theta)$. In this extreme misspecified scenario,
the minimax regret necessarily equals the capacity of $\Theta$, since
$D(P_\phi\|\Theta)=\infty$ for every $\phi \notin \Theta$ and zero otherwise, as
already discussed in the misspecified batch setting Section \ref{subsec:misspecified_batch_setting}.

Combining these observations, the misspecified combined batch and online minimax regret is
\begin{align}
    \begin{aligned}
        F_{n,l}(\Theta,\Phi)
        &=
        \frac{1}{l}\left[
            \sum_{t=1}^l \frac{d'}{2(n+t)}
            + o\left(\frac{1}{n+t}\right)
        \right]
        \\
        &\to
        \frac{d'}{2l}\,
        \log\left(1 + \frac{l}{n}\right),
    \end{aligned}
\end{align}
for $n,l \gg 1$, consistent with the asymptotic analysis in~\cite{combined24batch_online}.

Interestingly, the behavior interpolates smoothly between the online and batch
regimes. When $l \gg n \gg 1$, we have
\[
F_{n,l}(\Theta,\Phi)
=
\frac{d'}{2l}\log l
+ O\left(\frac{\log n}{l}\right),
\]
matching the online setting. In contrast, when $n \gg l \gg 1$,
\[
F_{n,l}(\Theta,\Phi)
\approx
\frac{d'}{2l},
\]
as in the batch setting.

    \section{Constrained Misspecified Universal Learning} 
    \label{sec:constrained_setting}
    In this Section, we introduce a new framework for the misspecified setting, termed the \emph{constrained misspecified setting}, in which the universal predictor is restricted to be a mixture over the convex hull of the hypothesis class \( \Theta \). Our analysis focuses primarily on the unsupervised online setting, with preliminary results for the batch case. In Section~\ref{subsec:constrained_online}, we derive a general closed-form expression for the minimax regret and analyze Bernoulli and Markov models. These examples show that the constrained misspecified minimax regret matches the well-specified capacity up to a fixed penalty term, and that the optimal constrained prior coincides with the capacity-achieving prior of the well-specified stochastic setting. We then prove in Section~\ref{subsec:constrained_smooth_parametric_models} that this equivalence generally holds for smooth parametric models. Beyond the general result, we present a simplified proof for exponential families using their distinctive structural characteristics. Finally, we illustrate the theory in the Gaussian location model, where the constrained minimax regret is shown to interpolate, under mild regularity conditions, between the well-specified capacity and the individual-sequence minimax regret.

%\subsection{Online Setting}
\subsection{General Analysis of the Constrained Misspecified Setting}
\label{subsec:constrained_online}
One of our main contributions to the constrained misspecified universal learning setting is the derivation of a general analytical expression for the minimax regret, given as follows:
\begin{thm}\label{thm:thmOnlineConstrainedRegret}
The constrained misspecified minimax regret of the online setting is given by:		
\begin{align}
    \begin{aligned}
    \nonumber
        R_{o,n}^*(\Theta,\Phi) &= \max_{\pi(\phi)} 
        \Big( I\left(Y^n;\Phi\right) -  
        E_{\pi(\phi)} \left\{  D\left(P_\phi \| \Theta\right) \right\} 
        \\&+ \min_{Q_{\pi_0(\theta)}}D\left(Q_{\pi(\phi)} \| Q_{\pi_0(\theta)}\right) \Big)
    \end{aligned}
\end{align}
where
\begin{align}
    \begin{aligned}
    \nonumber
        Q_{\pi_0(\theta)}(y^n) = \int{\pi_0(\theta)P_\theta(y^n)}d\theta
    \end{aligned}
\end{align}
and
\begin{align}
    \begin{aligned}
    \nonumber
        Q_{\pi(\phi)}(y^n) = \int{\pi(\phi)P_\phi(y^n)}d\phi.
    \end{aligned}
\end{align}
\end{thm}
\begin{IEEEproof}            
    See Appendix \ref{subsec:appendix_constrained_setting}.
\end{IEEEproof}

This paper focuses on the constrained misspecified learning under the online setting. A similar result by using the same information theory tools can be derived for the batch learning setting, and is given by the following Theorem:
\begin{thm}\label{thm:thmBatchConstrainedRegret}
The constrained misspecified minimax regret of the batch setting is given by:		
\begin{align}
    \begin{aligned}
    \nonumber
        R_{b,n}^*(\Theta,\Phi) &= \max_{\pi(\phi)} 
        \Big( I\left(Y_n;\Phi|Y^{n-1}\right) -  
        E_{\pi(\phi)} \left\{  D\left(P_\phi \| \Theta\right) \right\} 
        \\&+ \min_{Q_{\pi_0(\theta)}}D\left(Q_{\pi(\phi)} \| Q_{\pi_0(\theta)}\right) \Big)
    \end{aligned}
\end{align}
where
\begin{align}
    \begin{aligned}
    \nonumber
        Q_{\pi_0(\theta)}(y_n|y^{n-1}) = \frac{\int{\pi_0(\theta)P_\theta(y^n)}d\theta}{\int{\pi_0(\theta)P_\theta(y^{n-1})}d\theta}
    \end{aligned}
\end{align}
and
\begin{align}
    \begin{aligned}
    \nonumber
        Q_{\pi(\phi)}(y_n|y^{n-1}) = \frac{\int{\pi(\phi)P_\phi(y^n)}d\phi}{\int{\pi(\phi)P_\phi(y^{n-1})}d\phi}.
    \end{aligned}
\end{align}
\end{thm}

Theorems \ref{thm:thmOnlineConstrainedRegret} and \ref{thm:thmBatchConstrainedRegret} show that the constrained misspecified learning minimax regret can be regarded as a constrained version of the capacity, i.e., a constrained version of the minimax regret of the classical well-specified stochastic setting, see \cite{universal98prediction} and \cite{batch18learning} for the online and batch settings, respectively.
More precisely, the form of the result is similar to the misspecified learning minimax regret, as shown in Theorem \ref{thm:online_misspecified_regret} (\cite{robust21inference}, Theorem 2) and in \cite{online21misspecified} for the online setting and in Theorem \ref{ThmFirst} for the batch setting. In both cases the form of the regret is a combination of two terms: one is mutual information between the samples and the data source of distributions $\Phi$, given by $I(Y^n;\Phi)$ in the online setting and by $I(Y_n;\Phi|Y^{n-1})$ in the batch setting, and an additional penalty term, denoted by $E_{\pi(\phi)} \{  D({{P_{\phi}}\|\Theta}) \}$, which quantifies the mismatch between the set of hypotheses $\Theta$ and the data generating distributions set $\Phi$, by the closest projected distribution from the set $\Phi$ onto the set of hypotheses $\Theta$, in the sense of KL divergence. In the constrained misspecified setting, there is an additional penalty term, denoted by $\min_{Q_{\pi_0(\theta)}}D\left(Q_{\pi(\phi)} \| Q_{\pi_0(\theta)}\right)$, which quantifies the constraint penalty by the closest projected mixture distribution from the set $\Phi$ onto the mixture distributions over the hypotheses set $\Theta$, in the sense of KL divergence.              
                                
\emph{Two-Stage Arimoto-Blahut Algorithm Extension:}
As we have shown, the constrained minimax regret can be interpreted as a constrained version of the capacity between \( Y^n \) and the data-generating class \( \Phi \). Consequently, the prior \( \pi(\phi) \) can be interpreted as a capacity-achieving prior distribution. In general, obtaining a closed-form expression for either the capacity or the corresponding capacity-achieving prior is analytically intractable.

Therefore, following the approach of Algorithm~\ref{alg:ArimotoBlahutAlg_misspecified} in the
misspecified batch setting, we develop an extension of the Arimoto-Blahut algorithm for the
numerical evaluation of both the prior \( \pi(\phi) \) and the constrained minimax regret
\( R_{o,n}^*(\Theta,\Phi) \) in the constrained misspecified universal learning setting. In our
formulation, the regret depends on two coupled priors, \( \pi(\phi) \) and \( \pi_0(\theta) \), which are to be optimized jointly. Accordingly, the proposed algorithm proceeds in two alternating
steps: first, optimizing \( \pi(\phi) \) to maximize the regret, and second, optimizing
\( \pi_0(\theta) \), which defines the constrained universal predictor, to minimize it.
Corollary \ref{cor:corArimotoBlahut} derives an upper and lower bounds for the minimax regret. These bounds are used as convergence criteria for the iterative Arimoto-Blahut algorithm extension. Moreover, the form of the bounds as a linear combination of KL divergences implies the structure of the algorithm.
    
\begin{cor}
    \label{cor:corArimotoBlahut}
    The minimax regret of the constrained online learning under the misspecification setting holds the following for any $\Phi$, $\Theta$, $\pi(\phi)$ and $\pi_0(\theta)$:		
    \begin{align}									
        \begin{aligned}	                            
            \nonumber
            R_L(\pi,\Theta,\Phi) \leq R^*_{o,n}(\Theta,\Phi) \leq R_U(\pi_0,\Theta,\Phi) 
          \end{aligned}
    \end{align}           
    where,
    \begin{align}
        \begin{aligned}
        \nonumber
            R_L(\pi,\Theta,\Phi) &\equiv E_{\pi(\phi)}\left\{ D(P_\phi \| Q_{\pi(\phi)})\right\} - E_{\pi(\phi)}\left\{D(P_\phi \| \Theta)\right\}
            \\&+ \min_{\pi_0(\theta)}D(Q_{\pi(\phi)} \| Q_{\pi_0(\theta)})
        \end{aligned}
    \end{align}
    and 
    \begin{align}
        \begin{aligned}
        \nonumber
            R_U(\pi_0,\Theta,\Phi) \equiv \max_{P_{\phi}}\left( D(P_\phi \| Q_{\pi_0(\theta)}) -  D(P_\phi \| \Theta) \right).
        \end{aligned}
    \end{align}
\end{cor}
\begin{IEEEproof}
    Similar steps as in \cite{arimoto2013blahut_alg},\cite{combined24batch_online} and Corollary \ref{cor:thmArimotoBlahut}.%\cite{misspecified24learning_arxiv}.
\end{IEEEproof}

Using Corollary \ref{cor:corArimotoBlahut} we can extend the Arimoto-Blahut algorithm to the constrained misspecified learning setting by Algorithm \ref{alg:ArimotoBlahutAlg_constrained}. The inputs to the algorithm are the number of samples $n$, an optimization parameter $\lambda$, a required convergence accuracy parameter $\epsilon$, the sets $\Phi$ and $\Theta$ and initial prior distributions $\pi^{(0)}(\phi)$ and $\pi_0^{(0)}(\theta)$ (the uniform distribution over the sets $\Phi$ and $\Theta$ respectively, is a common practical choice). In the initialization stage, we calculate the lower and upper bounds $R_L$ and $R_U$, respectively, under the initial priors.
Then an iterative procedure is applied first over $\pi^{(i)}(\phi)$ and then, as a second update stage, the normalized portion of $\pi^{(i)}(\phi)$ in the range of $\phi\in\Theta$, is set to $\pi^{(i)}_0(\theta)$, until the convergence of the algorithm.
The outputs of the algorithm are the prior distributions $\pi(\phi)$ and $\pi_0(\theta)$.
        
\begin{algorithm}        
\caption{Two-Stage Arimoto-Blahut Algorithm}\label{alg:ArimotoBlahutAlg_constrained}
\begin{algorithmic}
\State \textbf{Input}: \\$n, \lambda, \epsilon, \Phi=\{\phi_m\}_{m=1}^{M_\phi}, \Theta=\{\theta_m\}_{m=1}^{M_\theta}, \pi^{(0)}(\phi), \pi_0^{(0)}(\theta)$
\State \textbf{Output}: \\$\pi(\phi)$ and $\pi_0(\theta)$
\State \textbf{Initialization}:
\State $i \gets 0$
%\State {$R_{U}^{(0)} = \max_{\phi}\left( D(P_\phi \| Q_{\pi_0^{(0)}}) -  D(P_\phi \| \Theta) \right)$}  
\State {$R_{U}^{(0)} = R_U(\pi_0^{(0)},\Theta,\Phi)$}
%\State {$R_{L}^{(0)} = E_{\pi^{(0)}}\left\{ D(P_\phi \| Q_{\pi^{(0)}}) - D(P_\phi \| \Theta)\right\} + D(Q_{\pi^{(0)}} \| Q_{\pi^{*(0)}_0})$}
\State {$R_{L}^{(0)} = R_L(\pi^{(0)},\Theta,\Phi)$}
\State \textbf{Loop}:
\begin{algorithmic}
\While {$R^{(i)}_{U} - R^{(i)}_{L} > \epsilon$}         
        \State \textbf{First Stage:}
        \State $\tilde{\pi}^{(i+1)}(\phi_j) = \pi^{(i)}(\phi_j) \cdot 
        e^{\lambda \left(D(P_{\phi_j} \| Q_{\pi_0^{(i)}}) - D(P_{\phi_j} \| \Theta) \right)}$
        \State $\pi^{(i+1)}(\phi_j) = \frac{\tilde{\pi}^{(i+1)}(\phi_j)}{ \sum_{j'=1}^{M_{\phi}} \tilde{\pi}^{(i+1)}(\phi_{j'})}$
        \State \textbf{Second Stage:}
        \State {$\tilde{\pi}_0^{(i+1)}(\theta_j) = \pi(\theta_j)$}
        \State $\pi_0^{(i+1)}(\theta_j) = \frac{\tilde{\pi}_0^{(i+1)}(\phi_j)}{ \sum_{j'=1}^{M_{\theta}} \tilde{\pi}_0^{(i+1)}(\theta_{j'})}$
        \State \textbf{Bounds Update:}
        \State {$R_{U}^{(i+1)} = R_U(\pi_0^{(i+1)},\Theta,\Phi)$}
        %\State {$R_{U}^{(i+1)} = \max_{\phi}\left( D(P_\phi \| Q_{\pi_0^{(i+1)}}) -  D(P_\phi \| \Theta) \right)$}    
        %\State {$R_{L}^{(i+1)} = E_{\pi^{(i+1)}}\left\{ D(P_\phi \| Q_{\pi^{(i+1)}}) - D(P_\phi \| \Theta)\right\} + D(Q_{\pi^{(i+1)}} \| Q_{\pi^{*(i+1)}_0})$}
        \State {$R_{L}^{(i+1)} = R_L(\pi^{(i+1)},\Theta,\Phi)$}
        \State $i \gets i+1$
\EndWhile
\State \textbf{end}
\State \textbf{Return}: \\$\pi(\phi)$ and $\pi_0(\theta)$
\end{algorithmic}
\end{algorithmic}
\end{algorithm}
To demonstrate our results, we apply the algorithm to the Bernoulli distribution sets, where $y\in \{0,1\}$, $\Phi$ is the set of all the distributions with probability of success $\phi \in [0,1]$ and $\Theta$ is the set of all the hypotheses with probability of success $\theta \in [0.25,0.75]$.
Figure \ref{figRegretVsN} shows regret versus $n$ for the well-specified stochastic, misspecified, and constrained misspecified settings. It can be shown that the constrained regret is greater than the well-specified stochastic setting over $\Theta$ only by a fixed constant, denoted by $\Delta R(\Theta,\Phi)\approx 0.3[bits]$. Note that by simulations this bound is true for any set $\Theta = 0.5 + [-\Delta/2,\Delta/2], \Delta\in[0,1]$ and shows that the optimal constrained learner is equal to the optimal learner of the well-specified stochastic setting of the set $\Theta$, as can be seen in Figure \ref{figPriors}. 

\begin{figure}[ht]        
        \centering
        \includegraphics[width=0.4\textwidth]{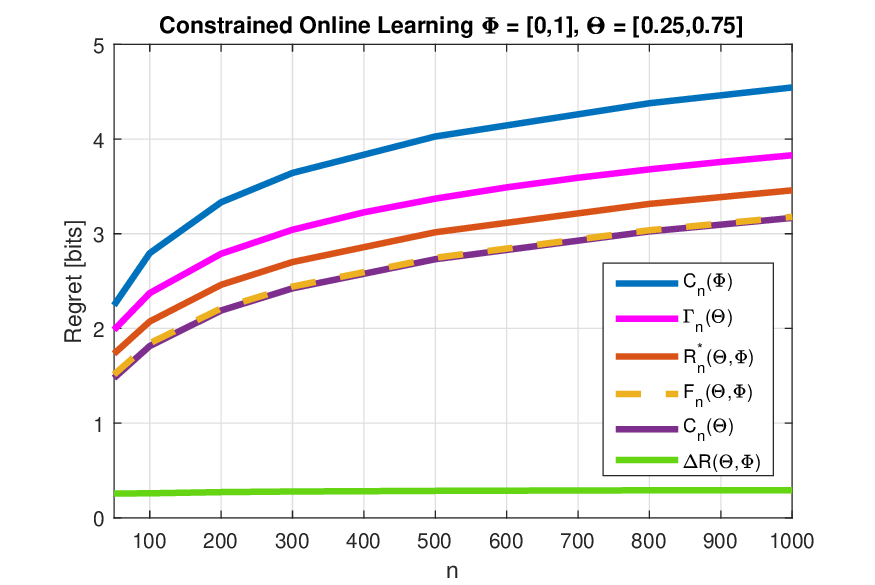}
        \caption{Numerical evaluation of the minimax regret as a function of \( n \) for various online learning settings in the Bernoulli case, with \( \Phi = [0,1] \) and \( \Theta = [0.25,0.75] \). The constrained minimax regret exceeds the well-specified capacity by a constant penalty of \( \Delta R(\Theta,\Phi) = 0.3~\text{[bits]} \).}
        \label{figRegretVsN}
\end{figure}
\begin{figure}[ht]        
        \centering
        \includegraphics[width=0.4\textwidth]{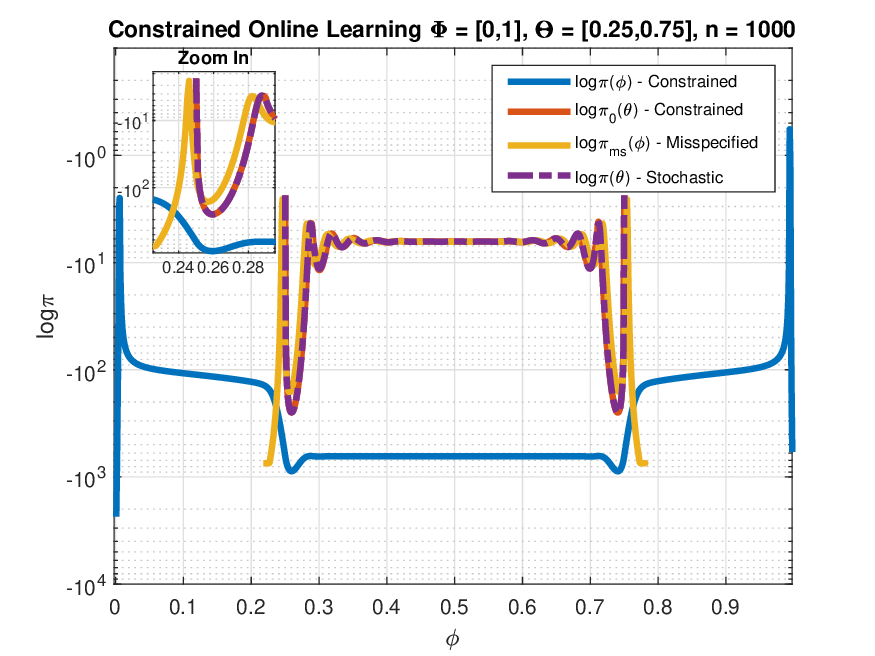}
        %\caption{The prior distributions of the various settings: $\pi(\phi)$ and $\pi_0(\theta)$ are the priors of the constrained setting, $\pi_{ms}(\phi)$ is the misspecified setting prior and $\pi(\theta)$ is the stochastic setting prior. Zooming in highlights the identity between $\pi_0(\theta)$ and $\pi(\theta)$.  }
        \caption{Prior distributions for the different settings. The priors \(\pi(\phi)\) and \(\pi_0(\theta)\) correspond to the constrained setting, \(\pi_{ms}(\phi)\) corresponds to the misspecified setting, and \(\pi(\theta)\) denotes the prior in the well-specified stochastic setting. The zoomed-in view highlights the equivalence between \(\pi_0(\theta)\) and \(\pi(\theta)\).}
        \label{figPriors}
\end{figure}                         

Another noteworthy observation, also supported numerically in Figure~\ref{figRegretVsN}, is that the constrained misspecified regret is upper bounded by the individual-sequence regret \( \Gamma_n(\Theta) \). In the following, we establish this upper bound analytically for the Bernoulli model class, consisting of all hypotheses with success probability \( \theta \in \Theta = (a,b) \), where \( 0 < a < b < 1 \).

To this end, let \( h(\cdot) \) denote the binary entropy function. For sequences of length \( n \), let \( \hat{p}_k \equiv  \frac{k}{n} \in \{0, \tfrac{1}{n}, \ldots, 1\} \) denote the empirical distribution. In the individual-sequence setting, the minimax regret is given by
\begin{align}
    \begin{aligned}
    \nonumber
        \Gamma_n(\Theta) &= \log\sum_{k=0}^n\binom{n}{\hat{p}_kn}\max_\theta2^{-n(D(\hat{p}_k\|\theta)+h(\hat{p}_k))}.
        %\\& = \log\sum_{k=0}^n\max_\theta A_n(\hat{p}_k,\theta),
    \end{aligned}
\end{align}
%where, $A_n(\hat{p}_k,\theta) \equiv \binom{n}{\hat{p}_kn}2^{-n(D(\hat{p}_k\|\theta)+h(\hat{p}_k))}$.

For large $n$, the dominant contribution comes from $\hat p_k \in (a,b)$. 
Hence,
\begin{align}
    \nonumber
    \Gamma_n(\Theta) = \log\bigg(\sum_{k=an}^{bn} \binom{n}{\hat{p}_kn}\, 2^{-n h(\hat{p}_k)}\bigg) + o(1).
\end{align}
Applying Stirling’s approximation to the binomial coefficient,
\[
    \binom{n}{\hat{p}_kn} = \frac{1}{\sqrt{2\pi n \hat{p}_k (1-\hat{p}_k)}}\,2^{n h(\hat{p}_k)} + O\left(\frac{1}{n}\right),
\]
yields the following
\begin{align}
    \begin{aligned}
        \nonumber
        \Gamma_n(\Theta) = \log\bigg(\sqrt{\frac{n}{2\pi}} \sum_{k=an}^{bn} \frac{1}{\sqrt{\hat{p}_k(1-\hat{p}_k)}}\cdot\frac{1}{n}\bigg) + o(1).
    \end{aligned}            
\end{align}
Approximating the Riemann sum by an integral gives,
\begin{align}
    \begin{aligned}
        \nonumber
        &\sum_{\hat{p}_k=a/n}^{b/n} \frac{1}{\sqrt{\hat{p}_k(1-\hat{p}_k)}}\cdot\frac{1}{n} = \int_a^b \frac{1}{\sqrt{p(1-p)}}\,dp 
        \\&+ O(1/n) = \arcsin(2b-1) - \arcsin(2a-1) + O(1/n).
    \end{aligned}
\end{align}
Combining the steps above,
\begin{align}
    \nonumber
    \Gamma_n(\Theta) &= \frac{1}{2}\log\frac{n}{2\pi} + \log(\arcsin(2b-1) 
    \\&- \arcsin(2a-1)) + o(1).
\end{align}
We can interpret the result as follows:
\begin{align}
    \begin{aligned}
        \Gamma_n(\Theta) &= \min_{Q} \max_{y^n} \log \frac{P_{\hat{\theta}}(y^n)}{Q(y^n)} 
        \approx \min_{Q} \max_{\hat{p}_k \in \Theta} \log \frac{P_{\hat{\theta}}(\hat{p}_k)}{Q(\hat{p}_k)}
        \\ &=  \min_{\pi_0(\theta)} \max_{\hat{p}_k \in \Theta} \log \frac{P_{\hat{\theta}}(\hat{p}_k)}{Q(\hat{p}_k)},
    \end{aligned}
\end{align}
where \( Q \) is taken to be a Bayesian mixture distribution induced by the Jeffreys prior \( \pi_0(\theta) \propto |I(\theta)|^{\frac{1}{2}} = \frac{1}{\sqrt{\theta(1-\theta)}} \), rather than the NML distribution originally derived by Shtarkov~\cite{shtar1987universal}. This choice follows the approach of Barron \emph{et al.}, who showed that Bayesian mixture distributions based on the Jeffreys prior achieve the same minimax regret asymptotically, see, e.g., \cite{ClarkeBarron1990,clarke1994jeffreys,barron2000individual,barron2013markov}.
In addition, the constrained misspecified regret, where $\Phi = \mathcal{P}$, the set of all probability distributions on $\mathcal{Y}^n$, can be approximated as follows:
\begin{align}
    \begin{aligned}
        R^*_{o,n}(\Theta,\mathcal{P}) &= \min_{\pi_0(\theta)} \max_{y^n} \log \frac{P_{\hat{\theta}}(y^n)}{Q(y^n)}
        \approx \min_{\pi_0(\theta)} \max_{\hat{p}_k \in \Theta} \log \frac{P_{\hat{\theta}}(\hat{p}_k)}{Q(\hat{p}_k)}.
    \end{aligned}
\end{align}
Therefore, we have asymptotically
\begin{align}
    \begin{aligned}
        \nonumber
        R^*_{o,n}(\Theta,\mathcal{P}) &= \Gamma_n(\Theta) = \frac{1}{2} \log\frac{n}{2\pi} + \log(\arcsin(2b - 1) 
        \\&- \arcsin(2a - 1)) + o(1),
    \end{aligned}
\end{align}
and in summary,
\begin{align}
    \begin{aligned}
        \nonumber
        C_n(\Theta) \leq R^*_{o,n}(\Theta,\Phi) \leq R_{o,n}^*(\Theta,\mathcal{P}) = \Gamma_n(\Theta).
    \end{aligned}
\end{align}    
If we further assume \( b = 1 - a \), then the expression simplifies to:
\begin{align}
    \begin{aligned}
        \nonumber
        R^*_{o,n}(\Theta,\Phi) \leq \Gamma_n(\Theta) &= \frac{1}{2} \log\frac{n}{2\pi} 
        \\&+ \log(2 \arcsin(1 - 2a)) + o(1),
    \end{aligned}    
\end{align}
for any \( a \in (0,1) \), as illustrated numerically for \( a = 0.25 \) in Figure~\ref{figRegretVsN}.

\emph{Markov Chain Models Example:}                         
Another interesting example is the constrained misspecified regret where $\Phi$ is the set of all ergodic Markov chain models of order $m$, while the set of hypotheses $\Theta$ contains only the subset of the ergodic Markov chain distributions of order $k$ for $k<m$ and alphabet $A=\{0,1,\dots,d\}$. In \cite{markov99redundancy} the constrained misspecified regret, where $\Phi$ is the set of $\phi$-mixing distributions, see \cite{processes69learning}, and $\Theta$ is the family of Markov chains of order $k$, was analyzed.
This regret for a given distribution $P\in\phi$-mixing and a prior distribution $\pi_0(\theta)$ was shown to be given by: 
\begin{align} 
    \begin{aligned} 
    \label{phiMixingRegret}
        R_n(P,\pi_0) &= \frac{(d+1)^kd}{2}\log \frac{n}{2\pi} 
        + \log \frac{{|I(\theta^*)|^{\frac{1}{2}}}}{\pi_0(\theta^*)}
        \\& - \sum_{t \in A^k}\sum_{u \in A} \frac{\tau^2_{t,u}}{2\eta^*_{t,u}}\log{e} + o(1),
    \end{aligned}
\end{align}
where $\theta^*$ is the projection of $P$ onto the set $\Theta$, which is given by the truncation of $P$ to a Markov process of order $k$. 
We explicitly denote by $\theta^*_{t,u}$ the transition probability from string $t\in A^k$ to symbol $u\in A$.  In addition, $\eta^*_{t,u}$ and $\tau_{t,u}^2 = \lim_{n\to\infty}n\mathbb{E}_{\phi}\left\{ (\hat{\eta}_{t,u} - \hat{\eta}_t\theta_{t,u}^*)^2 \right\}$ are the expectation and the variance growth rate of the empirical distribution of the string $tu$ of length $k+1$, denoted by $\hat{\eta}_{t,u}$ and $I(\theta^*)$ is the Fisher information of $\theta^*$.        
Since the family of Markov chain distributions is included in the set of $\phi$-mixing, by applying the same steps as in \cite{markov99redundancy}, we can conclude that the constrained misspecified regret where $\Phi$ and $\Theta$ are the sets of Markov chain distributions of orders $m$ and $k$, respectively, is also given by (\ref{phiMixingRegret}) for any $\phi\in\Phi$ and prior distribution $\pi_0(\theta)$.

We can immediately conclude that the constrained misspecified regret for any given $\phi$ and under the asymptotically stochastic prior distribution, also known as \emph{Jeffreys' prior} \cite{jeffreys46priors},\cite{rissanen96cimplexity},\cite{universal98prediction}, $\pi_0^*(\theta) = c |I(\theta|^{\frac{1}{2}}$ of the set of Markov chain families of order $k$ holds the upper bound $R_n(\phi,\pi_0^*) \leq C_{n}(\Theta) + \frac{(d+1)^kd}{2}\log{e},$        
where $C_{n}(\Theta) = \frac{(d+1)^kd}{2}\log \frac{n}{2\pi e}  - \log(c) + o(1)$ is the capacity of the well-specified stochastic setting of the set $\Theta$, according to (\cite{markov99redundancy}, Corollary 1). Therefore, we have the following
\begin{align}
    \begin{aligned}
    \label{markovConstrainedBound}
        R^*_{o,n}(\Theta,\Phi) \leq C_{n}(\Theta) + \frac{(d+1)^kd}{2}\log{e}.
    \end{aligned}
\end{align}        
In words, we see again that also in the constrained misspecified setting of Markov chain models, the usage of the well-specified stochastic mixture distribution over the set of hypotheses $\Theta$, has a fixed bounded penalty relative to the well-specified stochastic minimax regret of $\Theta$. 

Now we turn to demonstrate that this bound is tight by analyzing the binary Markov case where $m=1$ and $k=0$. To do so, we need to evaluate the moments of the empirical distribution, $\hat{\eta}_u = \frac{\sum_{t=1}^n \mathds{1}_{\{Y_t=u\}} }{n}$, for $u=0,1$. Note that the projection of any $\phi = (\phi_{01},\phi_{{10}})\in\Phi$ onto $\Theta$, where $\phi_{u,1-u}$ is the transition probability from symbol $u$ to symbol $1-u$, is given by the stationary distribution, denoted by $\theta^*(\phi)=\frac{\phi_{01}}{\phi_{01}+\phi_{10}}$. Asymptotically, the first moment of $\hat{\eta}_u$ equals ${\theta^*}^u(1-\theta^*)^{1-u}$ and its constant growth rate of the variance is given by the following:
\begin{align}
\label{tau}
    \begin{aligned}
        \tau_u^2 \equiv \lim_{n\to\infty}n\mathbb{E}_{\phi}\left\{ (\hat{\eta}_u - {\theta^*}^u(1-\theta^*)^{1-u})^2 \right\}.
    \end{aligned}
\end{align}
To evaluate (\ref{tau}) we use \cite{markov99redundancy} Corollary 1 proof, which is based on (\cite{processes69learning}, proposition 1.1.20), to get:
\begin{align}
    \begin{aligned}
        \tau_u^2 &=  
        {\theta^*}^u(1-\theta^*)^{1-u}
         \sum_{t=1}^\infty(P_{uu}^{(t+1)} - {\theta^*}^u(1-\theta^*)^{1-u})
         \\&+ \theta^*(1-\theta^*)
    \end{aligned}
\end{align}
where, $P_{uu}^{(t+1)} \equiv Pr\{Y_{t+1}=u|Y_1=u\}$. 

Therefore, the constrained misspecified regret for a given $\phi\in\Phi$ under the well-specified stochastic capacity achieving prior $\pi_0^*(\theta) = \frac{c}{\sqrt{\theta(1-\theta)}}$, where $c^{-1} = {\int_0^1{ \frac{d\theta}{\sqrt{\theta(1-\theta)}}}}$, is given by:
\begin{align}
    \begin{aligned}
    \label{binaryMarkovRegret}
        R_{n}(\phi,\pi^*_0(\theta)) = C_{n}(\Theta) 
         + \sum_{t=1}^\infty a_t,
    \end{aligned}
\end{align}
where, $a_t = \frac{1}{2}\left( 1 - P_{11}^{(t+1)} - P_{00}^{(t+1)} \right)\log{e}$.
In order to maximize (\ref{binaryMarkovRegret}) by $\phi$, and for symmetry reasons, we can maximize it by one parameter $0<\delta<1$, such that $\phi_{01}=\phi_{10}=1-\delta$. Note that for any $\delta$ we have $\theta^* = \frac{1}{2}$.
It can be verified that when $\delta\to 0^+$, there is no limit to the sum $\sum_{t=1}^{\infty}a_t$ since we have lost the ergodicity of the process, but it tends to $\sum_{t=1}^{\infty}\frac{1}{2}(-1)^{t+1}\log{e}$, which is bounded by $\frac{1}{2}\log{e}$, as expected by (\ref{markovConstrainedBound}). The intuition behind this result is the fact that when $\delta\to 0^+$, there are approximately only two possible series of data, one is $101010\cdots$ and the other is $010101\cdots$. In such a case, the projection $\theta^*=\frac{1}{2}$ approximately fits the real data generating distribution. Hence, it is difficult for the stochastic mixture distribution to compete with the compression performance of $\theta^*$, leading to this extra penalty.

Interestingly, on the other extreme, when $\delta\to 1^-$, the process also loses its ergodicity and $\sum_{t=1}^{\infty}a_t\to -\infty$. The intuition behind this result is the fact that when $\delta\to 1^{-}$, there are approximately only two possible series of data: $111111\cdots$ and $000000\cdots$. Therefore, the projection $\theta^*=\frac{1}{2}$ is a very poor predictor of the potential data sequences, whereas the mixture distribution performs much better.        
For a deeper understanding of the result, we remind the reader that regret can only be interpreted as the compression redundancy of the universal learner in relation to the compression performance of $\theta^*$, while the compression performance of $\theta^*$ is given by $D(P_\phi\|P_{\theta^*}) = n-1 \to \infty$, and that corresponds to $D(P_\phi\|P_{\theta^*}) = D(P_\phi\|Q_{\pi_0^*}) - R_n(\phi,\pi_0^*) \to \infty$.

\subsection{Constrained Misspecified Analysis in Smooth Parametric Models}
\label{subsec:constrained_smooth_parametric_models}    
As an initial step, we examine the relationship between the constrained misspecified regret and the individual setting regret in smooth parametric models. 
Trivially, the following inequality holds: 
\begin{align}
\label{eq:constraineVsindividual}
R_{o,n}^*(\Theta, \Phi) 
\leq R_{o,n}^*(\Theta, \mathcal{P}) 
\equiv \min_{\pi(\theta)} \max_{y^n} \log \frac{\max_{\theta\in\Theta}P_{\theta}(y^n)}{Q(y^n)},
\end{align}
where $\mathcal{P}$
%\begin{align}
%    \begin{aligned}
%        \nonumber
%        \mathcal{P} = \left\{ P : P(y^n) \geq 0, \int_{\mathcal{Y}^n}{P(y^n)}\,dy^n = 1 \right\},
%    \end{aligned}
%\end{align}
is the set of all probability distributions on $\mathcal{Y}^n$.

In addition, by definition, we have the following:
\[
\nonumber
R_{o,n}^*(\Theta, \mathcal{P}) \geq \min_{Q(y^n)}\max_{y^n}\log\frac{\max_{\theta\in\Theta}P_{\theta}(y^n)}{Q(y^n)} = \Gamma_n(\Theta),
\]
where, $Q(y^n)$ is the NML universal predictor, see \cite{shtar1987universal}.
However, in many smooth parametric settings of \( \Theta \), the asymptotic behavior satisfies the following:
\begin{align}
    \begin{aligned}
        \label{eq:constrainedSmoothParametricAsymptotic}
        R_{o,n}^*(\Theta, \mathcal{P}) \to \Gamma_n(\Theta) = \frac{d}{2}\log\frac{n}{2\pi} + \log\int_{\Theta}{|I(\theta)|^{\frac{1}{2}}}\,d\theta + o(1),
    \end{aligned}
\end{align}
where $d$ is the model dimension and the prior distribution \( \pi(\theta) \propto |I(\theta)|^{\frac{1}{2}} \) corresponds to the Jeffreys prior. This prior is known to be the capacity achieving distribution in the well-specified stochastic setting over the interior of the set \( \Theta \). Moreover, $\Gamma_n(\Theta) = C_n(\Theta) + \frac{d}{2}\log{e} + o(1)$ in these cases, where $C_n(\Theta)$ denotes the capacity of the well-specified setting of $\Theta$.

Such smooth parametric models, where the assumptions of the Laplace method hold, include discrete memoryless multinomial distributions \cite{barron2000individual}, exponential families \cite{barron2013nonExponential}, Markov models \cite{markov99redundancy, barron2013markov}, and finite-state machine (FSM) models \cite{gotoh1998fsmx}.

Note that by combining (\ref{eq:constraineVsindividual}) and (\ref{eq:constrainedSmoothParametricAsymptotic}) we get the following asymptotic result for such a smooth parametric hypothesis class of $\Theta$:
\begin{align}
    \begin{aligned}
    %\nonumber
         \Gamma_n(\Theta) \leq R_{o,n}^*(\Theta,\mathcal{P}) \to \Gamma_n(\Theta),
    \end{aligned}
\end{align}
and
\begin{align}
    \begin{aligned}
    \label{eq:constrained_regret_upper_bound}
         R_{o,n}^*(\Theta,\Phi) \leq R_{o,n}^*(\Theta,\mathcal{P}) \to \Gamma_n(\Theta).
    \end{aligned}
\end{align}

As an example, in Section \ref{subsec:constrained_online}, we investigated numerically the constrained misspecified setting for Bernoulli distributions, where $\Phi = [0,1]$ and $\Theta = [a,b] \subset [0,1]$ is a restricted subset. For the case $a = 1 - b = 0.25$, the following inequality was established:
\[
R_{o,n}^*(\Theta,\Phi) = C_n(\Theta) + 0.3~\text{[bits]} \leq \Gamma_n(\Theta),
\]
where
\[
\Gamma_n(\Theta) = \frac{1}{2}\log\frac{n}{2\pi} + \log\big(2 \arcsin(1 - 2a)\big) + o(1),
\]
for any $a = 1 - b$, and satisfies $\Gamma_n(\Theta) = C_n(\Theta) + \frac{1}{2}\log{e} + o(1)$. 
Furthermore, numerical results indicated that the prior $\pi(\theta)$ coincides with the capacity achieving prior distribution in the well-specified stochastic setting over $\Theta$.

Note that, using the inequality (\ref{eq:constrained_regret_upper_bound}), we can alternatively derive the result for the Markov example presented in Section~\ref{subsec:constrained_online} as follows:
\begin{align}
\begin{aligned}
\nonumber
R_{o,n}^*(\Theta,\Phi)
&\leq R_{o,n}^*(\Theta,\mathcal{P}) \\
&\leq \frac{(d+1)^kd}{2}\log\!\frac{n}{2\pi}
    + \log\!\int_{\Theta} |I(\theta)|^{\frac{1}{2}}\, d\theta
    + o(1) \\
&= C_n(\Theta) + \frac{(d+1)^kd}{2}\log{e} + o(1).
\end{aligned}
\end{align}
Here, the second inequality follows from the results of \cite{barron2013markov}, which derive upper and lower bounds on the minimax and maximin regrets, respectively, for Markov chain models under Bayesian mixture distributions in the individual-sequence setting.

Inspired by the preceding results and inequalities, a more general formulation of the constrained misspecified regret and the capacity achieving prior distribution in smooth parametric models is presented in the following theorem.

\begin{thm}
\label{thm:main_result}
Let $\Theta \subseteq \Phi$ be a smooth parametric model of i.i.d. distributions such that within the interior of $\Theta$, the Laplace approximation and regularity conditions of the log-likelihood hold. Under these conditions, the constrained misspecified minimax regret in the online setting admits the following asymptotic characterization (in natural units):
\begin{align}
\begin{aligned}
\label{eq:constrainedSmoothParametricAsymptotic1}
&R_{o,n}^*(\Theta, \Phi) 
= \frac{d}{2} \log \frac{n}{2\pi} + \log \int_\Theta \left| I(\theta) \right|^{\frac{1}{2}} \, d\theta 
\\&+\frac{1}{2}\max_{P_\phi \in \Phi} \left(  \log\frac{|J_\phi(\theta^*)|}{|I(\theta^*)|}
-\mathrm{trace}(K_\phi(\theta^*)J_\phi(\theta^*)^{-1}) \right) + o(1)
\end{aligned}
\end{align}
Moreover, the minimax regret is achieved by the Bayesian mixture distribution $Q(x^n) = \int_\Theta \pi(\theta)\, P_\theta(x^n)\, d\theta$, with prior \( \pi(\theta) \propto |I(\theta)|^{\frac{1}{2}} \).
\end{thm}

\begin{IEEEproof}
    See Appendix \ref{subsec:appendix_constrained_setting}.
\end{IEEEproof}
Note that the minimax regret expression in (\ref{eq:constrainedSmoothParametricAsymptotic1}) is in natural units, and multiplication of the trace term by the factor \( \log e \) makes it invariant to the choice of logarithmic base.

The expression in Theorem \ref{thm:main_result} can be simplified for exponential families of distributions, as can be shown by the following theorem.
\begin{thm}
    \label{thm:exponentialFamilyAsymptotic}
    Let $\Theta \subseteq \Phi$ be an exponential family of i.i.d. distributions of the form $P_\theta(y^n) = \prod_{i=1}^nh(x_i)\exp(\theta^{\top} \cdot T(x_i)-\psi(\theta))$, then within the interior of $\Theta$, the constrained misspecified minimax regret in the online setting admits the following asymptotic form:
    \begin{align}
    \begin{aligned}
    \label{eq:constrainedExponentialFamilyAsymptotic1}
    &R_{o,n}^*(\Theta, \Phi) 
    = \frac{d}{2} \log \frac{n}{2\pi} + \log \int_\Theta \left| \nabla^2_\theta \psi(\theta) \right|^{\frac{1}{2}} \, d\theta
    \\&-\frac{\log{e}}{2}\min_{P_\phi \in \Phi} \left( \mathrm{trace}(\operatorname{Cov}_{P_\phi}(T(X)) \nabla^2_\theta \psi(\theta^*)^{-1}) \right) + o(1).
    \end{aligned}
    \end{align}
    Moreover, the minimax regret is achieved by the Bayesian mixture distribution \( Q(y^n) = \int_\Theta \pi(\theta) P_\theta(y^n) \, d\theta \), with prior \( \pi(\theta) \propto |\nabla^2_\theta \psi(\theta)|^{\frac{1}{2}} \).
\end{thm}
\begin{IEEEproof}
    See Appendix \ref{subsec:appendix_constrained_setting}
\end{IEEEproof}

Theorems \ref{thm:main_result} and \ref{thm:exponentialFamilyAsymptotic} together provide a comprehensive characterization of the constrained misspecified framework in the online setting, for smooth parametric models in general, and in particular for exponential families.

We now turn to the constrained misspecified framework in the \emph{batch setting}, adopting a PAC perspective for smooth parametric models. Let \( \Phi \) denote the family of all i.i.d. distributions, and let \( \Theta \subseteq \Phi \) be a smooth parametric model of dimension \( d \).

Given a batch of \( n \) samples \( y^n \), the constrained universal predictor for the next observation \( y_{n+1} \) is defined as
\begin{align}
    \begin{aligned}
    \nonumber
        Q(y_{n+1}|y^n) = \frac{\int_\Theta{\pi(\theta)P_\theta(y^{n+1})}\,d\theta}{\int_\Theta{\pi(\theta)P_\theta(y^n)}\,d\theta}.
    \end{aligned}
\end{align}

Note that, according to our definitions, the appropriate notation of the constrained misspecified minimax regret in the batch setting is denoted by \( R_{n+1}^*(\Theta,\Phi) \), since the prediction involves a total of \( n+1 \) samples. However, by slight abuse of notation, we write it in the analysis below as \( R_n^*(\Theta,\Phi) \), consistent with conventions in the well-specified setting, where the batch size is sometimes taken to be \( n \) without explicitly indexing the next predicted sample. 

The asymptotic characterization of this minimax regret, along with the optimal constrained universal predictor, is given in the following theorem:
\begin{thm}
\label{thm:constrainedMisspecifiedBatch}
Let $\Phi$ be the set of all i.i.d. distributions and $\Theta \subseteq \Phi$ be a smooth parametric model such that within the interior of $\Theta$, the Laplace approximation and regularity conditions of the log-likelihood hold. Under these conditions, the constrained misspecified minimax regret in the batch setting admits the following asymptotic characterization (in natural units):
    \begin{align}
        \begin{aligned}
            \label{eq:asymptoticMinimaxRegretSmoothModelsBatch}
            &R_{b,n}^*(\Theta,\Phi) = \frac{d}{2}\log\left(1 + \frac{1}{n}\right) 
            \\&+ \frac{1}{4n^2}\max_{P_\phi\in\Phi}\mathrm{trace}\left(G(\theta^*) \nabla^2_\theta\log\frac{|I(\theta^*)|}{|J_\phi(\theta^*)|}  \right) 
            + o(n^{-2}),            
        \end{aligned}
    \end{align}
where $G(\theta^*) = J_\phi(\theta^*)^\top K_\phi(\theta^*)^{-1}J_\phi(\theta^*)$ is the Godambe information matrix. Moreover, the minimax regret is achieved by the Bayesian universal predictor $Q(y^{n+1} | y^n)$ with prior \( \pi(\theta) \propto |I(\theta)|^{1/2} \).
\end{thm}
\begin{IEEEproof}
    See Appendix \ref{subsec:appendix_constrained_setting}.
\end{IEEEproof}
Expanding the logarithmic term yields the equivalent form
\begin{align}
    \nonumber
    R_{b,n}^*(\Theta,\Phi)
    = \frac{d}{2n}
    + \frac{1}{4n^2}
    \max_{P_\phi\in\Phi}
    \mathrm{trace}\!\left(D(\theta^*)\right)
    + o(n^{-2}),
\end{align}
where
\[
D(\theta^*) \equiv
G(\theta^*)^{-1}
\nabla^2 \log\frac{|I(\theta^*)|}{|J_\phi(\theta^*)|}
- d.
\]
The quantity \( D(\theta^*) \) can be interpreted as a \emph{dimensional redundancy correction} induced by model misspecification.

This expansion highlights a key insight: in the PAC setting for smooth parametric models in the constrained misspecified batch setting, the leading term of the minimax regret coincides with the well-specified capacity, namely \( d/(2n) \). Misspecification affects only the second-order term of order \( O(n^{-2}) \). Moreover, the optimal constrained universal predictor employs the same prior that achieves capacity in the well-specified case.

In particular, when the model is well-specified, i.e. \( \Phi = \Theta \), we have \( D(\theta^*) = 0 \), and the minimax regret reduces exactly to \( d/(2n) \).

Similarly to Theorem \ref{thm:exponentialFamilyAsymptotic} for exponential families in the online setting, the identity \( I(\theta^*) = J_\phi(\theta^*) \) eliminates the second-order term in \eqref{eq:asymptoticMinimaxRegretSmoothModelsBatch} in the batch setting. Consequently, Theorem~\ref{thm:constrainedMisspecifiedBatch} simplifies to
\begin{align}
\nonumber
R_{b,n}^*(\Theta,\Phi)
= \frac{d}{2}\log\!\left(1 + \frac{1}{n}\right)
+ o(n^{-2})
= \frac{d}{2n} + O(n^{-2}).
\end{align}

An additional interesting insight in the constrained misspecified batch setting is that, under mild assumptions on the prior, namely smoothness and unimodality, misspecification influences only terms of order \( 1/n^2 \), while leaving the leading term \( d/(2n) \) unaffected. This property follows from the proof of Theorem~\ref{thm:constrainedMisspecifiedBatch}, provided in Appendix~\ref{subsec:appendix_constrained_setting}.

\emph{GLM Demonstration:}
The Gaussian Location Model is a fundamental parametric model in statistical theory and information-theoretic analysis. It assumes observations $Y_1, \dots, Y_n$ are i.i.d. from a normal distribution with unknown mean $\theta \in \Theta \subset \mathbb{R}^d$ and known covariance matrix $\Sigma$ (commonly defined as $\sigma^2I_d$). This model serves as a canonical example for studying asymptotic properties of estimators, minimax regret, and universal coding due to its simplicity and regularity.
The Fisher information for this model is constant, given by $I(\theta) = \Sigma^{-1}$, which implies that Jeffreys prior is uniform over $\Theta$.    
Consequently, the GLM provides a clean setting for deriving exact asymptotic minimax regret formulas and illustrating the role of information geometry in universal prediction and compression.

In the well-specified stochastic online setting, the asymptotic minimax regret, see e.g. \cite{clarke1994jeffreys}, is given by:
\[
C_n(\Theta) = \frac{d}{2}\log\frac{n}{2 \pi e} + \log{ \frac{\text{Vol}(\Theta)}{|\Sigma|^{\frac{1}{2}}} }   + o(1).
\]
where $\text{Vol}(\Theta) = \int_{\Theta} d\theta$, represents the Lebesgue measure (volume) of the parameter space $\Theta \subset \mathbb{R}^d$.

This expression aligns with the theoretical model of communication over an additive white Gaussian noise (AWGN) channel with amplitude constrained input, see \cite{smith1971amplitude}, where the transmitted signal $\theta \in [-\theta_0,\theta_0] \equiv \Theta$ and the received signal is $Y = \theta + N$, with $N \sim \mathcal{N}(0,\sigma^2)$ representing Gaussian noise. 
For $n$ i.i.d. observations of the received signal while transmitting $\theta$, and in the high signal to noise ratio (SNR) regime, specifically when $\frac{nE\{\theta^2\}}{\sigma^2} \gg 1$, the channel capacity asymptotically approaches:
\[
C_n(\Theta) \approx \log\left( 2\theta_0\sqrt{\frac{n}{2 \pi e \sigma^2}} \right) = \frac{1}{2}\log{\frac{n}{2\pi e}} + \log{\frac{2\theta_0}{\sigma}},
\]
and the capacity achieving input distribution converges to a uniform distribution over $[-\theta_0,\theta_0]$. 
This connection highlights the structural similarities among universal coding regret, universal prediction regret, and channel capacity in constrained settings.

In the deterministic individual sequence setting, see, e.g. \cite{barron2013nonExponential}, the minimax regret becomes:
\[
\Gamma_n(\Theta) = \frac{d}{2}\log\frac{n}{2 \pi} + \log{ \frac{\text{Vol}(\Theta)}{|\Sigma|^{\frac{1}{2}}} } + o(1).
\]    
The leading term $\frac{d}{2}\log n$ appears both in stochastic and in individual settings, reflecting the dimensional complexity. However, the absence of stochastic averaging in the individual sequence setting results in a slightly larger regret, as this criterion protects against the worst-case sequence rather than averaging over realizations.

For the misspecified online setting, see \cite{online21misspecified}, where the true distribution belongs to the PAC model, namely the set of all i.i.d. distributions with finite second moment, the minimax regret satisfies:
\[
F_{o,n}^{(\mathrm{PAC})}(\Theta,\Phi) = C_n(\Theta) + o(1),
\]
while in the well-specified and misspecified batch setting, see \cite{mourtada2022improper}, the minimax regret simplifies to:
\[
F_{b,n}(\Theta,\Phi) = C_{c,n}(\Theta) = \frac{d}{2}\log\left(1+\frac{1}{n}\right),
\]
where $\Phi$ denotes the set of all probability distributions with finite second moment.

The following theorem extends and generalizes the asymptotic minimax regret result for the GLM to the setting of constrained misspecified regret. 
We further show that, by suitably defining the sets $\Phi$ and $\Theta$, the constrained misspecified regret can be positioned between two fundamental extremes: the well-specified capacity and the deterministic individual sequence regret.
\begin{thm}
\label{thm:GLM}
Assume that the GLM parameter set is compact and given by:
\[
\Theta = \{\theta : \mathcal{N}(\theta,\Sigma)\} \subset \mathbb{R}^d.
\]
Let $\Phi \supseteq \Theta$ be a set of distributions such that for every $\phi \in \Phi$, the distribution $P_\phi$ has a well-defined covariance matrix $\Sigma_\phi$. Then, the constrained misspecified minimax regret in the online setting satisfies the following asymptotic expression:
\begin{align}
    \begin{aligned}
    \label{eq:GLMregret}
        R_{o,n}^*(\Theta,\Phi) &= \frac{d}{2}\log\frac{n}{2\pi} + \log{ \frac{\text{Vol}(\Theta)}{|\Sigma|^{\frac{1}{2}}} } 
        \\&- \frac{\log{e}}{2}\min_{P_\phi\in\Phi}{\mathrm{trace}(\Sigma_\phi\Sigma^{-1})} + o(1).
    \end{aligned}
\end{align}
%Moreover, $\pi(\theta) \propto |\Sigma|^{-\frac{1}{2}}$, which corresponds to a uniform distribution over the set $\Theta$.
Moreover, the minimax regret is achieved by the Bayesian mixture distribution \(Q(x^n)\) with prior $\pi(\theta) \propto |\Sigma|^{-\frac{1}{2}}$, which corresponds to a uniform distribution over the set $\Theta$.
\end{thm}
\begin{IEEEproof}
    See Appendix \ref{subsec:appendix_constrained_setting}.
\end{IEEEproof}

\begin{cor}
    For $\Phi = \mathcal{P}$, corresponding to the deterministic individual setting, we have the following:
    \begin{align}
        \begin{aligned}
        \nonumber
            R_{o,n}^*(\Theta,\Phi) = \Gamma_n(\Theta) = \frac{d}{2}\log\frac{n}{2 \pi } + \log{ \frac{\text{Vol}(\Theta)}{|\Sigma|^{\frac{1}{2}}} } + o(1).
        \end{aligned}
    \end{align}
\end{cor}        
\begin{IEEEproof}
    The result follows by setting $\Sigma_\phi = 0$ in (\ref{eq:GLMregret}).
\end{IEEEproof}

\begin{cor}
    For $\Phi=\Theta$, corresponding to the well-specified stochastic setting, we have the following:
    \begin{align}
        \begin{aligned}
        \nonumber
            R_{o,n}^*(\Theta,\Phi) = C_n(\Theta) = \frac{d}{2}\log\frac{n}{2 \pi e} + \log{ \frac{\text{Vol}(\Theta)}{|\Sigma|^{\frac{1}{2}}} } + o(1).
        \end{aligned}
    \end{align}
\end{cor}
\begin{IEEEproof}
    The result follows by setting $\Sigma_\phi=\Sigma$ in (\ref{eq:GLMregret}).
\end{IEEEproof}

Moreover, analogously to Theorem \ref{thm:GLM} in the online setting, applying the GLM case to Theorem \ref{thm:constrainedMisspecifiedBatch} in the PAC batch framework shows that the minimax regret coincides with the well-specified capacity up to an \( o(n^{-2}) \) correction:
\[
R_{b,n}^*(\Theta,\Phi)
= C_{c,n}(\Theta) + o(n^{-2})
= \frac{d}{2}\log\!\left(1+\frac{1}{n}\right) + o(n^{-2}).
\]

As an illustrative example showing that the constrained misspecified asymptotic regret can lie between the well-specified stochastic capacity and the deterministic individual regret, define
$\Phi = \Theta \cup \Lambda$, where
\[
\Lambda = \Bigg\{ \lambda : P_\lambda(y) = \prod_{i=1}^d \frac{\lambda_i^{y_i} e^{-\lambda_i}}{y_i!}, \; \lambda_i \geq \lambda_0 > 0,\; y \in \mathbb{N}^d \Bigg\},
\]
is the set of independent multivariate Poisson distributions with parameters $\lambda_i$ bounded below by a positive constant $\lambda_0$ for all $i = 1,\dots,d$. The covariance matrix of this distribution is $\Sigma_\lambda = \operatorname{diag}(\lambda_1,\dots,\lambda_d)$.    
In addition, for simplicity, consider the simplified GLM set $\Theta$ such that $\Sigma=\sigma^2I_d$.    
Then, according to Theorem~\ref{thm:GLM}, the asymptotic constrained misspecified regret requires evaluating the term:
\begin{align}
    \begin{aligned}
    \nonumber
        \min_{P_\phi} \operatorname{trace}(\Sigma_\phi / \sigma^2)
        &= \min\Bigg(d,\; \min_{\lambda}\sum_{i=1}^d \frac{\lambda_i}{\sigma^2}\Bigg)
        \\&= d \cdot \min(1,\lambda_0 / \sigma^2).
    \end{aligned}
\end{align}    
Hence, for $\lambda_0 \leq \sigma^2$, the minimax regret is:
\[
R_{o,n}^*(\Theta,\Phi) =
\frac{d}{2}\log\frac{n}{2\pi}
+ \log{ \frac{\text{Vol}(\Theta)}{|\Sigma|^{\frac{1}{2}}} }
- \frac{d\lambda_0}{2\sigma^2}\log{e}.
\]

Observe that setting $\lambda_0 = \sigma^2$ yields the lower bound corresponding to the well-specified stochastic capacity $C_n(\Theta)$, while letting $\lambda_0 \to 0$ converges to the upper bound associated with the deterministic individual regret $\Gamma_n(\Theta)$.

In all preceding GLM examples, we have shown that the constrained misspecified regret spans two extremes, the well-specified stochastic setting and the deterministic individual setting. This variability is determined by the relationship between the sets $\Phi$ and $\Theta$, and up to a constant factor, the constrained misspecified regret coincides with the well-specified stochastic capacity.

However, in certain extreme configurations of the parameter sets \( \Theta \) and \( \Phi \), the constrained misspecified regret can significantly exceed the well-specified capacity \( C_n(\Theta) \). To illustrate this phenomenon, we refer to the example by Feder and Polyanskiy in (\cite{online21misspecified}, Appendix F.1), which extends the GLM framework. 
In this setting, the observation vector is defined as an infinite dimensional vector, $Y = \phi + N$, where $N \sim \mathcal{N}(0, I_\infty)$, and the parameter sets \( \Phi \), \( \Theta \), and \( \Theta_\epsilon \) are defined in Section~\ref{subsec:misspecified_online_setting}.
It can be shown that the following asymptotic behaviors hold:    
\begin{align}
    \begin{aligned}
    \nonumber
        &C_n(\Phi) \asymp \log^2 n, \quad 
        C_n(\Theta) \asymp \log n, \quad
        C_n(\Theta_\epsilon) \asymp \log^2 n.
    \end{aligned}
\end{align}    
Therefore, according to Theorem \ref{thm:misspecified_online_regret_bound} (\cite{online21misspecified}, Theorem 4), since $C_n(\Phi)/n\to0$, we have:
\begin{align}
\nonumber
    R_{o,n}^*(\Theta,\Phi) 
    \geq F_n(\Theta,\Phi) = C_n(\Theta_\epsilon) + o(1) 
    \asymp \log^2{n}.
\end{align}    
Consequently, there exists an example for which    
\begin{align}    
    \begin{aligned}
    \nonumber
        \frac{R^*_{o,n}(\Theta,\Phi)}{C_n(\Theta)} = \Omega(\log{n}) \to \infty.
    \end{aligned}
\end{align}

    \section{Summary} 
    \label{sec:summary}
        This paper studies universal learning under model misspecification with log-loss, a framework originally introduced in the late 1990s and only recently revisited in the context of unsupervised online prediction. Misspecified universal learning provides a unifying perspective on universal prediction, encompassing both the well-specified stochastic and deterministic individual-sequence settings as special cases, and closely parallels agnostic statistical learning.

    We extend existing results beyond the unsupervised online setting to supervised learning and to both supervised and unsupervised batch scenarios. For all these regimes, we derive a closed-form expression of the misspecified minimax regret and identify the optimal universal predictor as a Bayesian mixture distribution over the data-generating class. In the misspecified unsupervised batch setting, we establish tight regret bounds and show that the minimax regret behaves similarly to that of the well-specified stochastic case, depending primarily on the complexity of the hypothesis class rather than on the broader data-generating family. This phenomenon is demonstrated numerically for Bernoulli models using an extension of the Arimoto-Blahut algorithm and analytically for general multinomial families.
    
    A central contribution of the paper is the introduction and analysis of the \emph{constrained misspecified setting}, in which the universal learner is restricted to mixtures over the hypothesis class rather than over the entire data-generating class. We show that under suitable regularity conditions, the constrained minimax regret coincides with the well-specified stochastic capacity up to a fixed penalty, and that the optimal constrained prior is given by the capacity-achieving prior. Through examples including Bernoulli, multinomial, Markov, exponential family models, and the Gaussian location model, we demonstrate that the constrained regret interpolates between stochastic capacity and individual-sequence regret, depending on the relationship between the hypothesis and data-generating classes.
    
    Overall, this work advances the theoretical understanding of universal learning under misspecification and contributes to the broader effort to establish an information-theoretic foundation for statistical machine learning.

\appendix

    \section{Appendix} 
    \label{sec:appendix}
    \subsection{Misspecified Online Setting Proofs}
\label{subsec:appendix_misspecified_online_setting}
% Proof of Theorem 4
\begin{IEEEproof}[\textbf{Proof of Theorem \ref{thm:misspecified_online_supervised_causal_setting}}]
The minimax regret is given by:
\begin{align}
    \begin{aligned}
        &R_{o,n}^*(\Theta,\Phi) = \min_{Q}\max_{P_\phi\in\Phi}\max_{P_\theta\in\Theta}\mathbb{E}_{P_\phi P}\left\{ \log{\frac{P_{\theta^*}(Y^n\|X^n)}{Q(Y^n\|X^n)}}  \right\}
        \\& = \min_{Q}\max_{P_\phi\in\Phi}\Bigg(\mathbb{E}_{P_\phi P}\left\{ \log{\frac{P_\phi(Y^n\|X^n)}{Q(Y^n\|X^n)}}  \right\} - D_c(P_\phi\|\Theta) \Bigg).
    \end{aligned}
\end{align}
Replacing the maximization over $\Phi$ by a maximization over priors $\pi$ on~$\Phi$ yields
\begin{align}
    \begin{aligned}
        R_{o,n}^*(\Theta,\Phi) &= \min_{Q}\max_{\pi(\phi)}\Bigg(\mathbb{E}_{\tilde{P}}\left\{ \log{\frac{P_\phi(Y^n\|X^n)}{Q(Y^n\|X^n)}}  \right\} 
        \\&- \mathbb{E_{\pi(\phi)}}\left\{D_c(P_\phi\|\Theta)\right\} \Bigg),
    \end{aligned}
\end{align}
where $\tilde{P}\equiv P(\phi,x^n,y^n)=\pi(\phi) P_\phi(y^n|x^n) P(x^n)$.

By Sion's minimax theorem \cite{sion1958general}, we may interchange the minimax into a maximin problem:
\begin{align}
    \begin{aligned}
        R_{o,n}^*(\Theta,\Phi) &= \max_{\pi(\phi)}\min_{Q}\Bigg(\mathbb{E}_{\tilde{P}}\left\{ \log{\frac{P_\phi(Y^n\|X^n)}{Q(Y^n\|X^n)}}  \right\} 
        \\&- \mathbb{E_{\pi(\phi)}}\left\{D_c(P_\phi\|\Theta)\right\} \Bigg).
    \end{aligned}
\end{align}
The inner minimization over $Q$ is achieved by the mixture
\[
Q(y^n\| x^n)
=
\int_{\Phi} \pi(\phi)\,P_\phi(y^n\| x^n)\,d\phi.
\]
A key observation of the proof is the causal factorization of both $Q$ and $P_\phi$:
\begin{align}
    \begin{aligned}
        Q(y_t|x^t,y^{t-1}) &= \frac{Q(y^t\|x^t)}{Q(y^{t-1}\|x^{t-1})} 
        \\&= \int_{\Phi}{\pi_t(\phi|x^{t},y^{t-1})P_\phi(y_t|x^t,y^{t-1})}\,d\phi,
    \end{aligned}
\end{align}
where the time‑updated prior is
\begin{align}
    \begin{aligned}
        \pi_t(\phi|x^t,y^{t-1}) &\equiv \frac{\pi(\phi)P_\phi(y^{t-1}\|x^{t-1})}{\int_{\Phi}{\pi(\phi)P_\phi(y^{t-1}\|x^{t-1})}\,d\phi} 
        \\ &= P(\phi|x^{t-1},y^{t-1}) = P(\phi|x^{t},y^{t-1}),
    \end{aligned}
\end{align}
where the last equality uses the assumption $X^n \perp\!\!\!\perp \Phi$ and the i.i.d.\ structure of~$X^n$.
We also use the causal chain rule \cite{massey1990causality}:
\begin{align}
    \begin{aligned}
        P_\phi(y^n\|x^n) = \prod_{t=1}^nP_\phi(y_t|x^t,y^{t-1}).
    \end{aligned}
\end{align}
Therefore,
\begin{align}
    \begin{aligned}
        &\mathbb{E}_{\tilde{P}}\left\{ \log{\frac{P_\phi(Y^n\|X^n)}{Q(Y^n\|X^n)}}  \right\} 
        = \sum_{t=1}^n\mathbb{E}_{\tilde{P}}\left\{ \log{\frac{P_\phi(Y_t|X^t,Y^{t-1})}{Q(Y_t|X^t,Y^{t-1})}}  \right\}
        \\ &= \sum_{t=1}^n\mathbb{E}_{\tilde{P}}\left\{ \log{\frac{P_\phi(Y_t|X^t,Y^{t-1})}{ \int_{\Phi}{\pi_t(\phi|X^{t},Y^{t-1})P_\phi(Y_t|X^t,Y^{t-1})}\,d\phi }}  \right\}
    \end{aligned}
\end{align}
and 
\begin{align}
    \begin{aligned}
        &\sum_{t=1}^n\mathbb{E}_{\tilde{P}}\left\{ \log{\frac{P_\phi(Y_t|X^t,Y^{t-1})}{ \int_{\Phi}{\pi_t(\phi|X^{t},Y^{t-1})P_\phi(Y_t|X^t,Y^{t-1})}\,d\phi }}  \right\}
        \\ &= \sum_{t=1}^n{ I(Y_t;\Phi | X^t, Y^{t-1}) } 
        = \sum_{t=1}^n{ I(Y_t;\Phi | X^n, Y^{t-1}) } 
        \\ &= I(\Phi \to Y^n|X^n),
    \end{aligned}
\end{align}
where the final equality follows from the definition of conditional directed information in \cite{massey1990causality}.

Hence, 
\begin{align}
    \begin{aligned}
        \mathbb{E}_{\tilde{P}}\left\{ \log{\frac{P_\phi(Y^n\|X^n)}{Q(Y^n\|X^n)}}  \right\} = I(\Phi \to Y^n|X^n),
    \end{aligned}
\end{align}
which completes the proof.
\end{IEEEproof}

% Proof of Theorem 5
\begin{IEEEproof}[\textbf{Proof of Theorem \ref{thm:misspecified_online_supervised_setting}}]
    Let us denote by $Q_c(y^n\|x^n)$ a causal universal predictor. Then trivially, we have
    \begin{align}
        \begin{aligned}
        \nonumber
            &R_{o,n}^*(\Theta,\Phi) = \min_{Q_c}\max_{P_\phi\in\Phi}\mathbb{E}_{P_\phi P}\left\{ \log{\frac{\prod_{t=1}^nP_{\theta^*}(Y_t|X_t)}{Q_c(Y^n\|X^n)}} \right\}
            \\ &\geq \min_Q\max_{P_\phi\in\Phi}\mathbb{E}_{P_\phi P}\left\{ \log{\frac{\prod_{t=1}^nP_{\theta^*}(Y_t|X_t)}{Q(Y^n|X^n)}} \right\}
            \\ &\geq  \min_Q\max_{\pi(\phi)}\left( \int_{\Phi}{\pi(\phi)D(P_\phi\|Q)}\,d\phi - \mathbb{E}_{\pi(\phi)}\{D(P_\phi\|\Theta)\} \right)
        \end{aligned}
    \end{align}
where $Q(y^n|x^n)$ is any conditional universal predictor (not necessarily causal). 

By exchanging the minimax and maximin operations, we obtain
\begin{align}
    \begin{aligned}
    \nonumber
        R_{o,n}^*(\Theta,\Phi) &\geq \max_{\pi(\phi)}\min_Q\Big( \int_{\Phi}{\pi(\phi)D(P_\phi\|Q)}\,d\phi 
        \\ &- \mathbb{E}_{\pi(\phi)}\{D(P_\phi\|\Theta)\} \Big)
        \\ &\geq \max_{\pi(\phi)}\left( I(Y^n;\Phi|X^n) - \mathbb{E}_{\pi(\phi)}\{D(P_\phi\|\Theta)\} \right),
    \end{aligned}    
\end{align}
where the inner minimization over $Q$ is given by the universal mixture distribution $Q(y^n|x^n) = \int_{\Phi}{\pi(\phi)P_\phi(y^n|x^n)}\,d\phi$ under the prior distribution $\pi(\phi)$.
\end{IEEEproof}

%Proof of Corollary 1
\begin{IEEEproof}[\textbf{Proof of Corollary \ref{cor:supervised_misspecified_regret_online_setting}}]
Under the memoryless assumption, causal conditioning coincides with standard conditioning:
\begin{align}
\nonumber
    Q(y^n \Vert x^n) = Q(y^n|x^n),
\end{align}
since $P_\phi(y_t|x^t,y^{t-1}) = P_\phi(y_t|x_t)$ for all $t$. 

Therefore, the causal predictor achieves the same performance as the optimal noncausal mixture predictor. Consequently, the conditions of Theorem~\ref{thm:misspecified_online_supervised_causal_setting} are satisfied, and all inequalities appearing in the proof of 
Theorem~\ref{thm:misspecified_online_regret_bound} hold with equality under the memoryless model. This establishes the claimed result.
\end{IEEEproof}

\subsection{Misspecified Batch Setting Proofs}
\label{subsec:appendix_misspecified_batch_setting}

%Proof of Theorem 6
\begin{IEEEproof}[\textbf{Proof of Theorem \ref{ThmFirst}}]
    The misspecified minimax regret, defined in (\ref{eq:misspecified_minimax_regret_definition}), is given by
    \begin{align}
        \begin{aligned}
        \nonumber
        F_{b,n}(\Theta,\Phi) = \min_{Q} \max_{P_{\phi} \in \Phi} \left(D({{P_{\phi}}\|{Q}}) - D({{P_{\phi}}\|P_{\theta^*}})\right).
        \end{aligned}
    \end{align}    
    Let us translate the minimax problem into a mixture minimax problem by the following:
    \begin{align}
        \begin{aligned}
            \nonumber
            F_{b,n}(&\Theta,\Phi) = \min_{Q} \max_{\pi({\phi})} 
            \int_{\phi}\pi(\phi) \left( D({{P_{\phi}}\|{Q}}) - D({{P_{\phi}}\|\Theta}) \right) \,d\phi
            \\&= \min_{Q} \max_{\pi({\phi})} 
             \underbrace{ \left( \mathbb{E}_{\pi(\phi)} \{  D({{P_{\phi}}\|{Q}}) \} 
        - \mathbb{E}_{\pi(\phi)} \{  D({{P_{\phi}}\|{\Theta}}) \} \right) }_{R_n(\pi(\phi),Q)}.
        %\\& \equiv \min_{Q} \max_{\pi({\phi})} R_n(\pi(\phi),Q),
        \end{aligned}
    \end{align}    
    where by definition:
    \begin{align}
        \label{TwoTermsRegret}
        \begin{aligned}
            R_n(\pi(\phi),Q) \equiv
            \underbrace{\mathbb{E}_{\pi(\phi)} \{  D({{P_{\phi}}\|{Q}}) \}}_{\propto \pi(\phi);\propto -\log{Q}} -
            \underbrace{\mathbb{E}_{\pi(\phi)} \{  D({{P_{\phi}}\|{\Theta}}) \}}_{\propto \pi(\phi);\text{not a function of } Q}.
        \end{aligned}
    \end{align}
    
    Since the first term of $R_n(\pi(\phi),Q)$ is proportional to $-\log{Q}$ and the second term is not a function of $Q$, (\ref{TwoTermsRegret}) is a convex function w.r.t $Q$. 
    Moreover, both the first and second terms are concave functions w.r.t $\pi(\phi)$, due to their linearity in $\pi(\phi)$. Thus, $R_n(\pi(\phi),Q)$ is also a concave function w.r.t $\pi(\phi)$.
    Hence, according to Sion's minimax Theorem \cite{sion1958general} for convex-concave functions, the minimax problem can be translated into a maximin problem as follows:
    \begin{align}
        \begin{aligned}
            \label{maxminRegret}
            F_{b,n}(\Theta,\Phi) = \max_{\pi({\phi})} \min_{Q} R_n(\pi(\phi),Q).
        \end{aligned}
    \end{align}	
    Now let us minimize $R_n(\pi(\phi),Q)$ w.r.t $Q$ by zeroing the derivative of the following Lagrangian:
    \begin{align}
        \begin{aligned}
        \nonumber
        L = R_n(\pi(\phi),Q) + \sum_{y^{n-1}}\lambda_{y^{n-1}}\sum_{y_n}Q(y_n|y^{n-1})
        \end{aligned}
    \end{align}	
    \begin{align}
        \begin{aligned}
        \nonumber
        \frac{\partial L}{\partial Q} \Big{|}_{Q_{\pi}} = - \frac{1}{Q_{\pi}(y_n|y^{n-1})}\int_{\phi}\pi(\phi)P_{\phi}(y^{n})  \,d\phi + \lambda_{y^{n-1}} = 0
        \end{aligned}
    \end{align}	
    Therefore, we get:
    \begin{equation}
    \nonumber
        Q_{\pi}({y_n|y^{n-1}}) = \frac{\int_{\phi}\pi(\phi)P_{\phi}(y^{n})  \,d\phi}{\lambda_{y^{n-1}}}
    \end{equation}	
    and in order to meet the constraint of $\sum_{y_{n}}Q_{\pi}({y_n|y^{n-1}}) = 1$ we get:
    \begin{align}
        \begin{aligned}
        \nonumber
        \sum_{y_n}Q_{\pi}({y_n|y^{n-1}}) & = \frac{\int_{\phi}\pi(\phi)\sum_{y_{n}}P_{\phi}(y^{n})  \,d\phi}{\lambda_{y^{n-1}}} \\& = 
        \frac{\int_{\phi}\pi(\phi)P_{\phi}(y^{n-1})  \,d\phi}{\lambda_{y^{n-1}}} = 1
        \end{aligned}
    \end{align}	
    which leads immediately to the following Lagrange multiplier:
    \begin{equation}
        \nonumber
        \lambda_{y^{n-1}} = \int_{\phi}\pi(\phi)P_{\phi}(y^{n-1})  \,d\phi.
    \end{equation}	
    Combining all the above, we get the minimizer $Q$ as:
    \begin{equation}
        \label{condQ}
        Q_{\pi}({y_n|y^{n-1}}) = \frac{\int_{\phi}\pi(\phi)P_{\phi}(y^{n})  \,d\phi}{\int_{\phi}\pi(\phi)P_{\phi}(y^{n-1})  \,d\phi}
    \end{equation}		
    or in other words,
    \begin{equation}
            \label{directQ}
        Q_{\pi}({y^{n}}) = \int_{\phi}\pi(\phi)P_{\phi}(y^{n})  \,d\phi,
    \end{equation}
    where similarly to the online case, (\ref{directQ})  is a mixture distribution over the set $\Phi$, see for example \cite{universal98prediction} and \cite{robust21inference}. As will be discussed later, the choice of the mixture prior is making the difference.
    
    Using (\ref{condQ}) We observe that:
    \begin{align}
        \begin{aligned}
        \label{FirstTerm}
        &\mathbb{E}_\pi \{  D({{P_{\phi}}\|{Q_\pi}}) \} = \int_{\phi} \pi{(\phi)} D({{P_{\phi}}\|{Q_\pi}}) \,d\phi 
        %\\& = \int_{\phi} \sum_{y^n} \pi{(\phi)}P_{\phi}(y^n)       \log{\frac{P_{\phi}(y_n|y^{n-1})}{Q_{\pi}(y_n|y^{n-1})}}        \,d\phi
        \\& = \int_{\phi} \sum_{y^n} P(y^n,\phi)       \log{\frac{P(y_n|y^{n-1},\phi)}{P({y^{n})/P(y^{n-1})}}}        \,d\phi  
        \\&  = \mathbb{E}_{P(Y^n,\Phi)}  \left\{     \log{\frac{P(Y_n|Y^{n-1},\Phi)}{P({Y_{n}|Y^{n-1})}}}  \right\} = I(Y_n;\Phi|Y^{n-1})
        \end{aligned}
    \end{align}
    where $P(Y^n=y^n,\Phi=\phi) \equiv \pi(\phi)P_\phi(y^n)$. Combining (\ref{TwoTermsRegret}), (\ref{maxminRegret}) and (\ref{FirstTerm}) we get (\ref{thmRegret}).
\end{IEEEproof}

%Proof of Lemma1
\begin{IEEEproof}[\textbf{Proof of Lemma 1}]
    Let us define the Markov chain triplet $B \to \Phi \to Y^n$, where $B \sim Ber(\lambda)$, $\lambda \in [0,1]$, $\Phi = \phi$ is conditionally distributed according to $\pi_{b}(\phi)$ given $B=b$ and $Y^n=y^n$ is conditionally distributed according to $P_{\phi}(y^n)$ given $\phi$.
    Note that the induced prior distribution over $\Phi$ is given by $\pi(\phi) = \lambda \pi_{1}(\phi) + (1-\lambda) \pi_{0}(\phi)$,
    and by definition we get for any $b \in \{0,1\}$:
    \begin{align}			
        \begin{aligned}				
            \label{lem1}
            J(\pi_{b}|B=b) & = I(Y_n;\Phi|Y^{n-1},B=b) 
            \\& - \mathbb{E}_{\pi_{b}}\{ D(P_{\phi} \| \Theta ) | B=b \}
        \end{aligned}
    \end{align}
    By taking the expectation of (\ref{lem1}) according to the $B$ distribution we get:
    \begin{align}
        \begin{aligned}
            \label{lem2}
            \mathbb{E}_B\{ J(\pi_{b}|B=b) \} = \lambda J(\pi_{1}) + (1-\lambda) J(\pi_{0})				
        \end{aligned}
    \end{align}
    while on the other hand,
    \begin{align}
        \begin{aligned}
            \label{lem3}
            \mathbb{E}_B\{ J(\pi_{b}|B=b) \} &= I(Y_n;\Phi | B, Y^{n-1}) 
            \\& - \mathbb{E}_{\pi}\{ D(P_{\phi} \| \Theta ) \}
        \end{aligned}
    \end{align}
    In addition, due to the Markov chain characteristics, we have $B\perp Y^n|\Phi$, which leads to $I(Y_n;B|\Phi,Y^{n-1}) = 0$. Therefore, by applying twice the mutual information chain rule we get:
    \begin{align}
        \begin{aligned}
            \label{lem4}
            I(Y_n;\Phi|Y^{n-1}) &= I(Y_n;B,\Phi|Y^{n-1}) - I(Y_n;B|\Phi,Y^{n-1})				
            \\& = I(Y_n;B,\Phi|Y^{n-1})
            \\& = I(Y_n;\Phi|B,Y^{n-1}) + I(Y_n;B|Y^{n-1}) 
            \\&\leq I(Y_n;\Phi|B,Y^{n-1}) + h(\lambda) 
        \end{aligned}
    \end{align}		
    Combining (\ref{lem1}), (\ref{lem2}), (\ref{lem3}) and (\ref{lem4}) with $J(\pi) = I(Y_n;\Phi|Y^{n-1}) - \mathbb{E}_{\pi}\{ D(P_{\phi} \| \Theta ) \}$ completes the proof.
\end{IEEEproof}

%Proof of Theorem 7
\begin{IEEEproof}[\textbf{Proof of Theorem \ref{Thm2}}]
    Since $F_{b,n}(\Theta,\Phi) = \max_{\pi({\phi})} J(\pi) \geq 0$ there exists a $\pi(\phi)$ such that
    \begin{align}
        \begin{aligned}
        \nonumber
            0 \leq J(\pi) & = I(Y_n;\Phi|Y^{n-1}) - E_\pi\{D_{c,n}(P_\phi \| \Theta)\} \\& \leq C_{c,n}(\Phi) -  \mathbb{E}_\pi\{D(P_\phi \| \Theta)\}.
        \end{aligned}
    \end{align}

    Therefore, the fraction of models of distributions, denoted by $\lambda$, from the set $\Phi$ that are not included in $\Theta_{\epsilon}$ can be upper bounded by Markov's inequality as follows:
    \begin{align}
        \begin{aligned}
        \nonumber
            \lambda = P(D(P_\phi \| \Theta) > \epsilon) \leq \frac{\mathbb{E}_\pi\{D(P_\phi \| \Theta)\}}{\epsilon} \leq \frac{C_{c,n}(\Phi)}{\epsilon}
        \end{aligned}
    \end{align}		
    Let us now define $\pi_{0}(\phi),\pi_{1}(\phi)$ as the distributions implied by $\pi(\phi)$ over the sets $\Theta_{\epsilon}$ and its complement, i.e., $\pi_{0}(\phi)=\pi(\phi)/(1-\lambda),\;\phi\in \Theta_\epsilon$ and zero otherwise, while $\pi_{1}(\phi)=\pi(\phi)/\lambda,\;\phi\notin \Theta_\epsilon$ and zero otherwise. As a consequence, $\pi(\phi) = \lambda \pi_{1}(\phi) + (1-\lambda) \pi_{0}(\phi),\;\forall \phi\in \Phi$.	
        Applying Lemma \ref{lemma1} and maximizing over $\pi_{0}(\phi)$ and $\pi_{1}(\phi)$ gives us the following:
    \begin{align}
        \label{SecResult1}
        \begin{aligned}
            J(\pi) &\leq \lambda J(\pi_1) + (1-\lambda) J(\pi_0) + h(\lambda)
            \\& \leq \lambda F_{b,n}(\Theta,\Phi) + (1-\lambda) F_{b,n}(\Theta,\Theta_{\epsilon}) + h(\lambda).
        \end{aligned}
    \end{align}		
    Maximizing (\ref{SecResult1}) over $\pi(\phi)$, combined with simple algebraic manipulations gives us the following:
    \begin{align}
        \begin{aligned}
        \nonumber
            F_{b,n}(\Theta,\Phi) \leq F_{b,n}(\Theta,\Theta_{\epsilon}) + \underbrace{ \frac{h(\lambda)}{1-\lambda} }_{A(\lambda)}.
        \end{aligned}
    \end{align}
    By taking $\epsilon = \epsilon_n \gg \tau_n$, we get $\lambda = \lambda_n \to 0$, which leads to $A(\lambda_n) \to 0$. 
    Furthermore, since $\Theta \subseteq \Theta_{\epsilon}$ for any $\epsilon \geq 0$, then clearly $F_{b,n}(\Theta,\Theta_{\epsilon}) \leq C_{c,n}(\Theta_{\epsilon})$.		
    Hence, we get:
    \begin{align}
        \begin{aligned}
        \nonumber
            C_{c,n}(\Theta) \leq F_{b,n}(\Theta,\Phi) \leq C_{c,n}(\Theta_{\epsilon_n}) + o(1)
        \end{aligned}
    \end{align}
    where the lower bound is given by the definition of the regret.
\end{IEEEproof}

%Proof of Corollary 2
\begin{IEEEproof}[\textbf{Proof of Corollary \ref{cor:thmArimotoBlahut}}]
    The regret can be written as follows:
    \begin{align}
        F_{b,n}(\Theta,\Phi) = \max_{\pi(\phi)}\left( \mathbb{E}_{\pi(\phi)}\left\{ D(P_\phi \| Q_{\pi}) -  D(P_\phi \| \Theta) \right\} \right).
    \end{align}
    Let us denote by $\pi^*(\phi)$ the maximizer prior distribution of the regret. Therefore, $\forall \pi(\phi) \neq \pi^*(\phi)$ we get the following lower bound:
    \begin{align}
        \begin{aligned}                    
        F_{b,n}(\Theta,\Phi) & = \mathbb{E}_{\pi^*(\phi)}\left\{ D(P_\phi \| Q_{\pi^*}) -  D(P_\phi \| \Theta) \right\} 
        \\& \geq \mathbb{E}_{\pi(\phi)}\left\{ D(P_\phi \| Q_{\pi}) -  D(P_\phi \| \Theta) \right\}.
        \end{aligned}         
    \end{align}                    
    On the other hand, by definition, we get the following upper bound:
    \begin{align}
        \begin{aligned}                    
        F_{b,n}(\Theta,\Phi) & = \min_{Q} \max_{P_{\phi} \in \Phi} \max_{P_{\theta} \in \Theta} R_n(\theta,\phi,Q) 
        \\& \leq \max_{P_{\phi} \in \Phi} \max_{P_{\theta} \in \Theta} R_n(\theta,\phi,Q_{\pi}) 
        \\& = \max_{P_{\phi}}\left( D(P_\phi \| Q_{\pi}) -  D(P_\phi \| \Theta) \right).
        \end{aligned}         
    \end{align}
\end{IEEEproof}

%Proof of Thorem 8
\begin{IEEEproof}[\textbf{Proof of Theorem \ref{thm:supervised_batch_learning}}]
    The minimax regret in this setting is given by the following:
    \begin{align}
        \begin{aligned}
            &R_{b,n}^*(\Theta,\Phi) = \min_{Q}\max_{P_\phi \in \Phi}\max_{P_\theta \in \Theta}  R(\theta,\phi,Q,P)             
            \\&= \min_{Q} \max_{P_\phi \in \Phi} \Bigg( \sum_{x^n} \sum_{y^n} {P(x^n)}{P_{\phi}(y^n|x^n)} \cdot 
            \\& \log\frac{P_{\phi}(y_n|x^n,y^{n-1})}{Q(y_n|x^n,y^{n-1})} - \min_{P_\theta \in \Theta}D(P_\phi \| P_\theta) \Bigg)            
            \\&\equiv \min_{Q}\max_{\pi(\phi)}R(\pi(\phi),Q)
        \end{aligned}
    \end{align}
    where
    \begin{align}
        \begin{aligned}
            &R(\pi(\phi),Q) \equiv \int \sum_{x^n} \sum_{y^n} \pi(\phi){P(x^n)}{P_{\phi}(y^n|x^n)} \cdot 
            \\& \log\frac{P_{\phi}(y_n|x^n,y^{n-1})}{Q(y_n|x^n,y^{n-1})}\,d\phi
            - \mathbb{E}_{\pi(\phi)}\left\{D_(P_\phi \| \Theta)\right\}.
        \end{aligned}
    \end{align}
    
    Using Sion's minimax theorem, see \cite{sion1958general}, we get:
    \begin{align}
        \begin{aligned}
            R_{b,n}^*(\Theta,\Phi) &= \min_{Q}\max_{\pi(\phi)}R(\pi(\phi),Q) 
            \\&= \max_{\pi(\phi)}\min_{Q}R(\pi(\phi),Q).
        \end{aligned}
    \end{align}        
    Defining the following Lagrangian:    
    \begin{align}
        \begin{aligned}
            L = R(\pi(\phi),Q) + \sum_{x^n}\sum_{y^{n-1}}\lambda_{x^n,y^{n-1}}\sum_{y_n}Q(y_n | x^n, y^{n-1})
        \end{aligned}
    \end{align}
    and zeroing the derivative of $L$ w.r.t $Q$ we get:    
    \begin{align}
        \begin{aligned}
            \frac{\partial L}{\partial Q} &= -\frac{1}{Q(y_n | x^n, y^{n-1})} \cdot 
            \\&\int_{\phi}\pi(\phi)P(x^n)P_\phi(y^n|x^n)\,d\phi + \lambda_{x^n,y^{n-1}} = 0.
        \end{aligned}
    \end{align}
    By the constraint $\sum_{y_n}Q(y_n | x^n, y^{n-1}) = 1$ we get:
    \begin{align}
        \begin{aligned}
            \lambda_{x^n,y^{n-1}} = \int_{\phi}{\pi(\phi)P(x^n)P_\phi(y^{n-1}|x^n)}\,d\phi.
        \end{aligned}
    \end{align}    
    Combining all the above gives us:
    \begin{align}
        \begin{aligned}
            {Q(y_n | x^n, y^{n-1})} &= \frac{\int_{\phi}{\pi(\phi)P_\phi(y^n|x^n) \,d\phi}}{\int_{\phi}{\pi(\phi)P_\phi(y^{n-1}|x^n) \,d\phi}}.
        \end{aligned}
    \end{align}    
    Let us define the joint probability distribution function of $\Phi$, $X^n$ and $Y^n$ by:
    \begin{align}
        \begin{aligned}
            \tilde{P} = P(\phi,x^n,y^n) = \pi(\phi)P(x^n)P_\phi(y^n|x^n),
        \end{aligned}
    \end{align}    
    then
    \begin{align}
        \begin{aligned}
            R(\pi(\phi),Q) &= \mathbb{E}_{\tilde{P}}\left\{ \log\frac{P_\phi(Y_n|X^n,Y^{n-1})}{Q(Y_n|X^n,Y^{n-1})} \right\} 
            \\&- \mathbb{E}_{\pi(\phi)}\left\{ D( P_\phi \| \Theta ) \right\}        
        \end{aligned}
    \end{align}
    and it can be verified that
    \begin{align}
        \mathbb{E}_{\tilde{P}}\left\{ \log\frac{P_\phi(Y_n|X^n,Y^{n-1})}{Q(Y_n|X^n,Y^{n-1})} \right\} = I(Y_n;\Phi|X^n,Y^{n-1}).
    \end{align}
    Therefore,
    \begin{align}
        \begin{aligned}
            R_{b,n}^*(\Theta,\Phi) &= \max_{\pi(\phi)}\big( I(Y_n;\Phi|X^n,Y^{n-1})
            \\&- E_{\pi(\phi)}\left\{ D_{c}( P_\phi \| \Theta ) \right\} \big),
        \end{aligned}
    \end{align}    
    which completes the proof.
\end{IEEEproof}

\subsection{Misspecified Combined Batch and Online Proofs}
\label{subsec:appendix_misspecified_combined_setting}

%Proof of Theorem 10
\begin{IEEEproof}[\textbf{Proof of Theorem \ref{thm:thmCombinedSettingBounds}}]
    The proof begins by applying the chain rule to the universal distribution defined in (\ref{CombinedUniversalDistribution}), yielding the following decomposition:
    \begin{align}
        \begin{aligned}
            Q_\pi(y^l | y^n) &= \Pi_{t=1}^l Q_\pi(y_{n+t} | y^{n+t-1}) 
            \\&= \Pi_{t=1}^l \frac{\int{\pi(\phi)P_{\phi}(y^{n+t})d\phi}}{\int{\pi(\phi)P_{\phi}(y^{n+t-1})d\phi}}.
        \end{aligned}
    \end{align}
    The misspecified combined batch and online minimax regret is given by the following:
    \begin{align}
        \begin{aligned}
            F_{n,l}(\Theta,\Phi) &= \max_{\pi(\phi)} \frac{1}{l} \int \pi(\phi) \sum_{y^{n+l}} P_\phi(y^{n+l}) \cdot 
            \\& \log \frac{P_{\theta^*}(y^l | y^n)}{Q_\pi(y^l | y^n)} \, d\phi.
        \end{aligned}
    \end{align}
    
    Moreover, under the assumption of the theorem that the hypotheses generate i.i.d. data, we have 
    \[P_\theta(y^l | y^n) = \prod_{t=1}^l P_\theta(y_{n+t}).\]
    Combining this with the preceding steps yields the following:
    \begin{align}
        \begin{aligned}
        \nonumber
            F_{n,l}(\Theta,\Phi) &\leq \frac{1}{l} \sum_{t=1}^l \max_{\pi(\phi)} \int \pi(\phi) \sum_{y^{n+t}} P_\phi(y^{n+t}) \cdot
            \\& \log \frac{P_{\theta^*}(y_{n+t})}{Q_\pi(y_{n+t} | y^{n+t-1})} d\phi,
        \end{aligned}
    \end{align}
    where the inequality follows from the fact that the sum of maximized terms is always no smaller than the maximization of their sum. Therefore,
    \begin{align}
        \begin{aligned}
        \nonumber
            F_{n,l}(\Theta,\Phi) &\leq \frac{1}{l} \sum_{t=1}^l F_{b,n+t}(\Theta,\Phi 
            \\& \leq \frac{1}{l}\sum_{t=1}^lC_{c,n+t}\left( \Theta_{\epsilon_{n+t}} \right) + o(1),
        \end{aligned}
    \end{align}
    where the first inequality follows directly from the definition of the misspecified batch minimax regret, and the final inequality follows from Theorem~\ref{Thm2} under the theorem's assumptions.
    
    The lower bound on the misspecified combined batch and online minimax regret follows directly from the definition: by taking $\Phi \equiv \Theta$, we obtain the combined batch and online minimax regret of the well-specified stochastic setting.
\end{IEEEproof}

\subsection{Constrained Misspecified Setting Proofs}
\label{subsec:appendix_constrained_setting}

%Proof of Theorem 11
\begin{IEEEproof}[\textbf{Proof of Theorem \ref{thm:thmOnlineConstrainedRegret}}]
By the translation of the the maximization over the set $\Phi$ in (\ref{eq:minmaxRegretDefinition_constrained}) into a mixture maximization over the prior probability distribution $\pi(\phi)$ and by using the mixture distribution definition $Q_{\pi(\phi)} \equiv \int{\pi(\phi)}P_\phi(y^n)d\phi$, 
it can be easily verified that the minimax regret is given by the following:
\begin{align}
    \begin{aligned}
    \nonumber
        R_{o,n}^*(\Theta,\Phi) &= \min_{Q_{\pi_0(\theta)}}\max_{\pi(\phi)} 
        \int{\pi(\phi)R_n\left( P_\phi, P_{\theta^*}, Q_{\pi_0(\theta)} \right)}d\phi
        \\&= \min_{Q_{\pi_0(\theta)}}\max_{\pi(\phi)} 
        \Big( I\left(Y^n;\Phi\right) - E_{\pi(\phi)} \{ D\left(P_\phi \| \Theta\right) \}  
        \\&+  D\left(Q_{\pi(\phi)} \| Q_{\pi_0(\theta)}\right) \Big)
        \\&\equiv \min_{Q_{\pi_0(\theta)}}\max_{\pi(\phi)} R_n\left(\pi(\phi),Q_{\pi_0(\theta)}\right).
    \end{aligned}
\end{align}                   
Since the term $D\left(Q_{\pi(\phi)} \| Q_{\pi_0(\theta)}\right)$ is proportional to $-\log(Q_{\pi_0(\theta)})$ and the other terms are not a function of $Q_{\pi_0(\theta)}$, $R_n\left(\pi(\phi),Q_{\pi_0(\theta)}\right)$ is a convex function w.r.t $Q_{\pi_0(\theta)}$. 
  Moreover, all the three terms are concave functions w.r.t $\pi(\phi)$, due to     
their linearity in $\pi(\phi)$. Thus, $R_n\left(\pi(\phi),Q_{\pi_0(\theta)}\right)$ is also a concave      function w.r.t $\pi(\phi)$.
  Hence, according to Sion's minimax theorem \cite{sion1958general} for convex-concave         
functions, the minimax problem can be translated into a max-min problem as follows:            
\begin{align}
    \begin{aligned}
    \nonumber
        R_{o,n}^*(\Theta,\Phi) &= \max_{\pi(\phi)}\min_{Q_{\pi_0(\theta)}} 
        \Big( I\left(Y^n;\Phi\right) - E_{\pi(\phi)} \{ D\left(P_\phi \| \Theta\right) \}  
        \\&+  D\left(Q_{\pi(\phi)} \| Q_{\pi_0(\theta)}\right) \Big)
        \\&= \max_{\pi(\phi)} 
        \Big( I\left(Y^n;\Phi\right) - E_{\pi(\phi)} \left\{  D\left(P_\phi \| \Theta\right) \right\}
        \\&+  \min_{Q_{\pi_0(\theta)}}D\left(Q_{\pi(\phi)} \| Q_{\pi_0(\theta)}\right) \Big).
    \end{aligned}
\end{align}
\end{IEEEproof}

%Proof of Theorem 13
\begin{IEEEproof}[\textbf{Proof of Theorem \ref{thm:main_result}}]
    By applying the minimax theorem \cite{sion1958general} to interchange the order of minimization and maximization the constrained misspecified regret can be expressed by:
    \begin{align}
        \begin{aligned}
        \nonumber
        R_{o,n}^*(\Theta, \Phi) 
        &= \max_{P_\phi \in \Phi} \left(
        \min_{\pi(\theta)} \mathbb{E}_{P_\phi} \bigg\{ \log \frac{P_{\hat{\theta}}(y^n)}{Q(y^n)} \right\} 
        \\&+ \mathbb{E}_{P_\phi} \left\{ \log \frac{P_{\theta^*}(y^n)}{P_{\hat{\theta}}(y^n)} \right\} 
        \bigg).
        \end{aligned}
    \end{align}

    %\vspace*{10pt}
    To approximate the first term, we assume that the conditions required for the Laplace method hold for smooth parametric models within the interior of \( \Theta \). Under these assumptions, the following approximation is obtained:
    \begin{align}
        \begin{aligned}
        \nonumber
        \log \frac{P_{\hat{\theta}}(y^n)}{Q(y^n)} 
        = \frac{d}{2} \log \frac{n}{2\pi}
        + \log \frac{|\hat{I}(\hat{\theta})|^{\frac{1}{2}}}{\pi(\hat{\theta})} + o(1).
        \end{aligned}
    \end{align}        
    This expression is minimized asymptotically, independently of $P_\phi$, by choosing the Jeffreys prior $\pi(\theta) \propto |I(\theta)|^{\frac{1}{2}}$. Using this choice and the properties established in \cite{White1982}, the expectation under \( P_\phi \) converges to:
    \begin{align}
        \begin{aligned}
        \nonumber
            \mathbb{E}_{P_\phi} \left\{ \log \frac{|\hat{I}(\hat{\theta})|^{\frac{1}{2}}}{|I(\hat{\theta})|^{\frac{1}{2}}}  \right\} = \frac{1}{2}\log{\frac{|J_\phi(\theta^*)|}{|I(\theta^*)|}} + o(1).
        \end{aligned}
    \end{align}
    
    Using the Godambe Information Approximation \cite{Godambe1960,White1982,vdV1998}, we have: 
    \begin{align}
        \begin{aligned}
        \nonumber
            \mathbb{E}_{P_\phi} \left\{ \log \frac{P_{\theta^*}(y^n)}{P_{\hat{\theta}}(y^n)} \right\} 
            = -\frac{1}{2}\mathrm{trace}(K_\phi(\theta^*)J_\phi(\theta^*)^{-1}) + o(1).
        \end{aligned}
    \end{align}
    Note that this expression is in natural units, and multiplication by the factor \( \log e \) makes it invariant to the choice of logarithmic base.
   
    Combining these results yields the stated asymptotic characterization.
\end{IEEEproof}

%Proof of Theorem 14
\begin{IEEEproof}[\textbf{Proof of Theorem \ref{thm:exponentialFamilyAsymptotic}}]
    In exponential families, the following identities hold:
    \[
    J_\phi(\theta) = I(\theta) = \nabla^2_\theta \psi(\theta), 
    \qquad
    K_\phi(\theta) = \operatorname{Cov}_{P_\phi}(T(X)).
    \]
    Thus, (\ref{eq:constrainedSmoothParametricAsymptotic1}) simplifies to:
    \begin{align}
        \begin{aligned}
        \nonumber
        &R_{o,n}^*(\Theta, \Phi) 
        = \frac{d}{2} \log \frac{n}{2\pi} + \log \int_\Theta \left| \nabla^2_\theta \psi(\theta^*) \right|^{\frac{1}{2}} \, d\theta
        \\&-\frac{1}{2}\min_{P_\phi \in \Phi} \left( \mathrm{trace}(\operatorname{Cov}_{P_\phi}(T(X)) \nabla^2_\theta \psi(\theta^*)^{-1}) \right)\log{e} + o(1),
        \end{aligned}
    \end{align}
    and, according to Theorem~\ref{thm:main_result}, the optimal prior is given by \(\pi(\theta) \propto |\nabla^2_\theta \psi(\theta)|^{\frac{1}{2}}\).
\end{IEEEproof}

%Proof of Thorem 15
\begin{IEEEproof}[\textbf{Proof of Theorem \ref{thm:constrainedMisspecifiedBatch}}]
    The constrained universal predictor in the batch setting is defined as
    \begin{align}
        \nonumber
        Q(y_{n+1} | y^n) \equiv \frac{q_{n+1}(y^{n+1})}{q_n(y^n)},
    \end{align}
    where
    \begin{align}
        \nonumber
        q_n(y^n) = \int_{\Theta} \pi(\theta) P_\theta(y^n)\, d\theta.
    \end{align}
    
    Under the theorem conditions, Laplace's method yields    
    \begin{align}
        \begin{aligned}
            \nonumber
            &\log Q(y_{n+1} | y^n)
            = \log q_{n+1}(y^{n+1}) - \log q_n(y^n)
            \\&= l_{n+1}(\hat{\theta}_{n+1}) - l_n(\hat{\theta}_n)
            - \frac{d}{2}\log\!\left(1+\frac{1}{n}\right)
            \\&\quad + \log \frac{\pi(\hat{\theta}_{n+1})}{\pi(\hat{\theta}_n)}
            - \frac{1}{2}\log \frac{|\hat{I}_{n+1}(\hat{\theta}_{n+1})|}{|\hat{I}_n(\hat{\theta}_n)|}
            + o(n^{-2}).
        \end{aligned}
    \end{align}
    where
    \[
        l_n(\theta) \equiv \sum_{t=1}^n \log P_\theta(y_t),
        \quad
        \hat{I}_n(\theta) = -\frac{1}{n} \nabla_\theta^2 l_n(\theta).
    \]
    
    Let $\theta^*$ denote the KL projection of $P_\phi$ into the hypotheses set $\Theta$. A Taylor expansion gives
    \begin{align}
        \begin{aligned}
            \nonumber
            l_{n+1}(\hat{\theta}_{n+1})
            &= l_{n+1}(\theta^*)
            + \nabla_\theta l_{n+1}(\theta^*)^\top (\hat{\theta}_{n+1} - \theta^*) \\
            &\quad - \frac{n+1}{2} (\hat{\theta}_{n+1} - \theta^*)^\top 
            \hat{I}_{n+1}(\theta^*)(\hat{\theta}_{n+1} - \theta^*),
        \end{aligned}
    \end{align}
    with
    \begin{align}
        \nonumber
        l_{n+1}(\theta^*) = l_n(\theta^*) + \log P_{\theta^*}(y_{n+1}).
    \end{align}
    Similarly,
    \begin{align}
        \begin{aligned}
            \nonumber
            l_n(\hat{\theta}_n)
            &= l_n(\theta^*)
            + \nabla_\theta l_n(\theta^*)^\top (\hat{\theta}_n - \theta^*) \\
            &\quad - \frac{n}{2} (\hat{\theta}_n - \theta^*)^\top 
            \hat{I}_n(\theta^*)(\hat{\theta}_n - \theta^*).
        \end{aligned}
    \end{align}
    
    Define
    \begin{align}
        \nonumber
        \Delta_n \equiv l_{n+1}(\hat{\theta}_{n+1}) - l_n(\hat{\theta}_n)
        - \log P_{\theta^*}(y_{n+1}).
    \end{align}
    Then,
    \begin{align}
        \begin{aligned}
            \nonumber
            \log Q(y_{n+1} | y^n)
            &= \log P_{\theta^*}(y_{n+1})
            + \Delta_n
            - \frac{d}{2}\log\!\left(1+\frac{1}{n}\right) \\
            &\quad + A_{n+1}(\hat{\theta}_{n+1}) - A_n(\hat{\theta}_n)
            + o(n^{-2}),
        \end{aligned}
    \end{align}
    where
    \begin{align}
        \nonumber
        A_n(\theta) \equiv \log \frac{\pi(\theta)}{|\hat{I}_n(\theta)|^{1/2}}.
    \end{align}
    Under misspecification \cite{White1982,Godambe1960},
    \begin{align}
        \nonumber
        \hat{\theta}_n - \theta^* \;\overset{d}{\approx}\;
        \mathcal{N}\!\left(0,\, G(\theta^*)^{-1}/n\right),
    \end{align}
    where
    \begin{align}
        \nonumber
        G(\theta^*) = J_\phi(\theta^*)^\top K_\phi(\theta^*)^{-1} J_\phi(\theta^*).
    \end{align}
    
    Using standard quadratic-form identities, 
    \begin{align}
        \nonumber
        \mathbb{E}_{P_\phi}\!\{ ( \hat{\theta}_n - \theta^* )^\top \hat{I}_n(\theta^*) ( \hat{\theta}_n - \theta^* ) \} = \mathrm{trace}( K_\phi(\theta^*) J_\phi(\theta^*)^{-1})
    \end{align}
    we get
    \begin{align}
        \nonumber
        \mathbb{E}_{P_\phi}\{\Delta_n\} = 0 + o(n^{-2}).
    \end{align}
    In addition, a second-order expansion yields
    \begin{align}
        \begin{aligned}
            \nonumber
            A_n(\hat{\theta}_n)
            &= A_n(\theta^*)
            + \nabla_\theta A_n(\theta^*)^\top (\hat{\theta}_n - \theta^*) \\
            &\quad + \frac{1}{2}(\hat{\theta}_n - \theta^*)^\top
            \nabla_\theta^2 A_n(\theta^*)(\hat{\theta}_n - \theta^*).
        \end{aligned}
    \end{align}
        
    Taking expectations,
    \begin{align}
        \begin{aligned}
            \nonumber
            \mathbb{E}_{P_\phi}\{A_n(\hat{\theta}_n)\}
            &= A_n(\theta^*)
            + \frac{1}{2n} \operatorname{trace}\!\left(
            G(\theta^*)^{-1} \nabla_\theta^2 A(\theta^*)
            \right),
        \end{aligned}
    \end{align}
    where asymptotically
    \begin{align}
        \nonumber
        A(\theta^*) = \log \frac{\pi(\theta^*)}{|J_\phi(\theta^*)|^{1/2}}.
    \end{align}
    Similarly,
    \begin{align}
        \begin{aligned}
            \nonumber
            \mathbb{E}_{P_\phi}\{A_{n+1}(\hat{\theta}_{n+1})\}
            &= A(\theta^*)
            \\&+ \frac{1}{2(n+1)} \operatorname{trace}\!\left(
            G(\theta^*)^{-1} \nabla_\theta^2 A(\theta^*)
            \right).
        \end{aligned}
    \end{align}    
    
    Note that, to minimize regret, the prior is chosen as the Jeffreys prior,
    \begin{align}
        \nonumber
        \pi(\theta) \propto |I(\theta)|^{1/2}.
    \end{align}
    
    Combining all terms, the minimax regret is
    \begin{align}
        \begin{aligned}
            \nonumber
            &R_{b,n}^*(\Theta,\Phi) = \frac{d}{2}\log\left(1 + \frac{1}{n}\right) 
            \\&+ \frac{1}{4n^2}\max_{P_\phi\in\Phi}\mathrm{trace}\left(G(\theta^*)^{-1} \nabla^2_\theta\log\frac{|I(\theta^*)|}{|J_\phi(\theta^*)|}  \right) 
            + o(n^{-2}).          
        \end{aligned}
    \end{align}
\end{IEEEproof}

%Proof of Thorem 16
\begin{IEEEproof}[\textbf{Proof of Theorem \ref{thm:GLM}}]
    According to Theorem \ref{thm:exponentialFamilyAsymptotic}, for an exponential family of distributions, evaluating the minimax regret requires analyzing the relevant terms under the natural parametrization, where \( \eta = \Sigma^{-1}\theta \), $T(X) = X$, and $\psi(\eta) = \tfrac{1}{2}\eta^\top \Sigma \eta$, as follows:
    \begin{align}
        \begin{aligned}
        \label{eq:constrainedGLMstep1}
            \operatorname{Cov}_{P_\phi}(T(X)) = \Sigma_\phi, \quad 
            I(\eta^*) = \left. \nabla^2_{\eta} \psi(\eta) \right|_{\eta^* \equiv \Sigma^{-1}\theta^*} = \Sigma.
        \end{aligned}
    \end{align}            
    By substituting (\ref{eq:constrainedGLMstep1}) into (\ref{eq:constrainedExponentialFamilyAsymptotic1}) and noting that:
    \begin{align}
        \begin{aligned}
        \nonumber
            I(\theta) = \frac{\partial \eta}{\partial \theta}^\top \nabla^2_\eta \psi(\eta) \frac{\partial \eta}{\partial \theta} = \Sigma^{-1},
        \end{aligned}
    \end{align}
    we conclude that the optimal prior takes the form $\pi(\theta) \propto |I(\theta)|^{\frac{1}{2}} = |\Sigma|^{-\frac{1}{2}}$, and the minimax regret is given by (\ref{eq:GLMregret}).
\end{IEEEproof}

\bibliographystyle{IEEEbib}
\bibliography{references}

\vspace{11pt}


\vfill

\end{document}